\documentclass[twocolumn,10pt,notitlepage,nofootinbib,numerical]{revtex4-2}

\usepackage[utf8]{inputenc}
\usepackage[colorlinks=true,linkcolor=blue,citecolor=blue,urlcolor=blue]{hyperref}

\usepackage{amsmath,amssymb}
\usepackage{dsfont}
\usepackage{xcolor}
\usepackage{amsthm}

\usepackage{algorithm}
\usepackage{algpseudocode}

\algnewcommand\algorithmicforeach{\textbf{for each}}
\algdef{S}[FOR]{ForEach}[1]{\algorithmicforeach\ #1\ \algorithmicdo}

\usepackage{ragged2e}

\newtheorem{theorem}{Theorem}[section]

\newtheorem{corollary}[theorem]{Corollary}

\newtheorem{lemma}[theorem]{Lemma}

\theoremstyle{definition}
\newtheorem{definition}[theorem]{Definition}
\newtheorem{example}{Example}[section]
\newtheorem{remark}{Remark}

\definecolor{FGreen}{RGB}{1,100,10}

\usepackage{graphicx}
\usepackage[caption=false]{subfig}
\graphicspath{
    {./}
}

\usepackage{mathtools} 
\usepackage[shortlabels]{enumitem}
\usepackage{bm}
\usepackage{dsfont} 
\usepackage{etoolbox}

\docsvlist{A,B,C,D,E,F,G,H,I,J,K,L,M,N,O,P,Q,R,S,T,U,V,W,X,Y,Z}

\docsvlist{A,B,C,D,E,F,G,H,I,J,K,L,M,N,O,P,Q,R,S,T,U,V,W,X,Y,Z}

\docsvlist{A,B,C,D,E,F,G,H,I,J,K,L,M,N,O,P,Q,R,S,T,U,V,W,X,Y,Z,c,d}

\docsvlist{d}

\docsvlist{E,P,R,e,f,g,h,i,j,k,p,q,r,x,y}

\usepackage{ stmaryrd }
\def\llb{\llbracket}
\def\rrb{\rrbracket}
\def\bra{\langle}
\def\ket{\rangle}
\def\lr{\leftrightarrow}
\def\backslash{\symbol{92}}

\def\weight{\mathfrak{w}}
\def\qubit{\mathfrak{q}}

\graphicspath{
    {./}
}

\DeclareMathOperator{\im}{im}

\DeclareMathOperator{\supp}{supp}

\DeclareMathOperator{\poly}{poly}

\DeclareMathOperator{\row}{row}
\DeclareMathOperator{\col}{col}
\DeclareMathOperator{\Row}{Row}
\DeclareMathOperator{\Col}{Col}
\DeclareMathOperator{\AB}{AB}

\usepackage{tikz-cd} 
\usetikzlibrary{matrix,arrows,decorations.pathmorphing}
\usetikzlibrary{fit, shapes.geometric,calc}

\usepackage{xcolor}

\newcommand{\red}[1]{{\color{red} #1}}
\newcommand{\gray}[1]{{\color{gray} #1}}
\newcommand{\blue}[1]{{\color{blue} #1}}

\makeatletter
\def\l@subsection#1#2{}
\def\l@subsubsection#1#2{}
\makeatother

\begin{document}

\title{4D and 5D Layer Codes through Color Routing}
\author{Andrew C.~Yuan}
\affiliation{Iceberg Quantum, Sydney}
\author{Nouédyn Baspin}
\affiliation{Iceberg Quantum, Sydney}
\date{\today}

\begin{abstract}
    We introduce an explicit Calderbank–Shor–Steane (CSS) code construction that generalizes the Layer codes to $D=4,5$ dimensions. 
    Much like its predecessor, the present construction is based on embedding quantum low-density parity check (qLDPC) codes; from an $[[n,k,d]]$ code with energy barrier $\Delta$, we obtain a $D=4,5$ dimensional Layer code with parameters $[[\Theta(n^{D/(D-2)}), k, \Theta(d n^{1/{(D-2)}})]]$ and energy barrier $\Omega(\Delta)$. Using good qLDPC codes as input, our construction saturates the $D=4,5$ dimensional BPT bounds exactly. 
    The higher dimensional Layer Codes are modular, and thus well suited to architectures composed of modular network patches, despite our physical limitation to three dimensions.
    We overcome the hurdles encountered by previous generalization attempts through the use of \textit{color routing}, allowing us to resolve the structure of the check layers and line defects.
\end{abstract}

\maketitle


\section{Introduction}
\label{sec:intro}

Quantum error correction protects fragile quantum information by diluting 
it across a very large physical system \cite{shor1996faultolerant,knill1998resilient,aharonov2008faulttolerant}. 
It is therefore of fundamental interest to understand which physical systems can reliably store the most information with the smallest physical footprint.
In particular a long line of research \cite{bravyi2009nogo, bravyi2010tradeoffs, flammia2017limits, baspin2022connectivity, dai2024localityvsquantumcodes} demonstrates that the efficiency of geometrically local systems is bounded, leading to the following question: What is the larged amount of information that can be reliably stored in a bounded volume of space?

In this work we focus on Pauli-based stabilizer codes, whose codespace corresponds to the groundspace of a commuting many-body Pauli Hamiltonian \cite{gottesman2024surviving}, where each Pauli interaction is local in a fixed Euclidean dimension $D$.
A check refers to a local Pauli interaction of that Hamiltonian, where its \emph{support} is the set of qubits on which it acts, while its \emph{weight} is the number of those qubits. 

We denote $n$ the number of physical qubits used and $k$ the number of logical qubits stored. 
The distance $d$ is the size of the smallest undetectable error, also referred to as a nontrivial \emph{logical operator} \cite{gottesman2024surviving}. These quantities are often concisely written $[[n,k,d]]$ and play a crucial role in analyzing a code's behavior for active error correction, where the checks are measured to estimate whether an error occurred, and then a correction is applied. 

In contrast, passive error correction is a regime where the dynamics of the Hamiltonian themselves are leveraged to suppress errors, very much like the ferromagnetic phase of the Ising model suppresses fluctuations that would reverse the encoded magnetization. 
In that case, the dynamical properties of the code also need to be characterized, which is often done through the energy barrier $\Delta$ associated with the code Hamiltonian. 
For a nontrivial logical error to build up, the environment has to pump at least energy $\Delta$ into the system, an event whose probability typically scales like $e^{-\beta \Delta}$ \cite{brown2016quantum, alicki2008thermal, temme2016thermalization} and is thus exponentially suppressed. However, note that a growing energy barrier alone does not establish self-correction, as the \textit{entropy} of available error paths may be larger than the probability supression \cite{brown2014entropic, siva2017topological, baspin2025free}.

\begin{table*}[t]
\caption{\label{tab:comparison}Local CSS embeddings from qLDPC inputs. Optimal means saturation up to constants with suitable input families.}
\begin{ruledtabular}
\begin{tabular}{p{0.28\textwidth}cp{0.20\textwidth}p{0.28\textwidth}}
Work & Dimension & Saturation & Input scope \\
\hline
Portnoy \cite{portnoy2023local} & $D\geq3$ & Up to polylog & Sparse integer lift$^{a}$ \\
Subdivision \cite{lin2023geometrically,li2026almost} & $D\geq3$ & Up to polylog$^{b}$ & Any qLDPC CSS code$^{c}$ \\
Layer codes \cite{williamson2023layer} & $3$ & Optimal & Any qLDPC CSS code \\
This work & $4,5$ & Optimal & Any qLDPC CSS code \\
\end{tabular}
\end{ruledtabular}
\smallskip
\parbox{0.96\textwidth}{\footnotesize $^{a}$The condition is defined in \cite{freedman2021building}. $^{b}$Theorems~1 and~4 of Ref.~\cite{li2026almost} retain polylogarithmic factors. 
Its post-acceptance Remark~5 announces exact saturation based on \cite{portnoy2024sparse}, but a full code-level proof was not provided.
Ref.~\cite{lin2023geometrically} also announces exact saturation, though a complete proof was not provided. 
$^{c}$The arbitrary-input guarantee is the later extension \cite{li2026almost}.}
\end{table*}

As previously alluded to, geometry places fundamental limits on these quantities \cite{bravyi2009nogo,bravyi2010tradeoffs}. 
Local stabilizer codes in $D$ dimensional Euclidean space obey $kd^{2/(D-1)}=O(n)$, $d=O(n^{(D-1)/D})$, and $\Delta=O(n^{(D-2)/D})$. We refer to these locality bounds collectively as the Bravyi--Poulin--Terhal (BPT) bounds. 
Finding codes that saturate these bounds reduce to satisfying the following (optimal) scaling
\begin{equation}
\label{eq:BPT-bounds}
\begin{gathered}
 k=\Theta(n^{(D-2)/D}),\quad d=\Theta(n^{(D-1)/D}),\\
 \Delta=\Theta(n^{(D-2)/D}).
\end{gathered}
\end{equation}

A variety of approaches have been attempted to saturate these bounds: cellular automata \cite{bravyi2010tradeoffs}, fractals \cite{yoshida2013information, brell2016proposal}, deformation of topological codes \cite{michnicki2014memory}. But the first systemic approach relied on an ingenious code embedding method, where good sparse-but-nonlocal stabilizer codes \cite{panteleev2022asymptotically,leverrier2022quantum,dinur2023good} are embedded in a target dimension $D$, yielding optimal \emph{subsystem} codes\footnote{The existence of good qLDPC codes was unknown at the time, hence why the authors do not claim optimality, nevertheless the implications are straightforward.}. 
As subsystem codes allow for non-commuting Hamiltonians, the method was not immediately transferable to stabilizer codes, but nevertheless remains the first instance of the later very successful embedding approach. 

There currently exists two lines of approach which uses good qLDPC inputs to obtain geometrically local stabilizer codes which saturate the optimal scalings (Table~\ref{tab:comparison}). 
The first builds a geometric structure from the input code, subdivides it, and then uses a random embedding to obtain locality in any fixed dimension $D\geq3$. 
Using a relatively advanced result in manifold construction \cite{freedman2021building}, Portnoy showed that any qLDPC CSS code admitting a \textit{sparse} $\dZ$ \textit{lift} can be embedded in $D$ dimensions, such that good qLDPC inputs result in \textit{almost} optimal embeddings \cite{portnoy2023local}. 
Lin, et. al. \cite{lin2023geometrically,li2026almost} later improved upon this result by removing the somewhat unnatural requirement so that the embedding applies to any qLDPC CSS code.
The proved guarantees retain polylogarithmic losses; Table~\ref{tab:comparison} distinguishes them from announced improvements to exact optimality. 

For these proved almost optimal bounds, the Euclidean embedding step relies on probabilistic arguments, which only establishes the existence of a layout \cite{portnoy2023local,lin2023geometrically,li2026almost}. 
The layout is selected through random choice, rather than specified by a deterministic geometric rule. 
Selecting and verifying a concrete finite-size layout therefore remains an additional step, which can hinder physical implementation and further numerical or analytical studies that depend on the detailed geometry. 
It remains open to provide explicit constructions which also address the remaining polylogarithmic losses.

The second line of approach, three-dimensional Layer Codes, provides an explicit, deterministic geometric construction \cite{williamson2023layer}. 
It assembles surface-code patches into layers and specifies their local couplings along lines through lattice surgery \cite{horsman2012surface}. 
For good qLDPC CSS inputs with a linear energy barrier, these codes attain the optimal scalings in Eq. \eqref{eq:BPT-bounds} without polylogarithmic losses. 
Their prescribed geometry makes individual code instances directly available for finite-size calculations and detailed analysis. 
In particular, it has led to several lines of follow-up research: studies of their partial self-correction properties \cite{baspin2025free,gu2025layer,williamson2025partial}, extensions to embeddings of random CSS codes \cite{gu2025layer,hsieh2025simplified}, and generalizations for quantum weight reduction \cite{yuan2025unified,layerreduction}.


Unfortunately, the Layer Codes do not admit an immediate generalization to higher dimensions, and thus from this perspective, the remaining problem is reduced to extending the approach beyond three dimensions while retaining exact optimality. 
In particular, we provide a resolution in four and five dimensions through \emph{color routing} (see Overview in Section \ref{sec:overview}). 
\subsection{Main Result}
\label{sec:main-result}

Our main result is an explicit, deterministic construction of geometrically local CSS codes in $D=4,5$ from arbitrary qLDPC CSS inputs. The procedure specifies both the local checks and their geometry; suitable input families give optimal codes as quantified below. 
Let $A$ denote the input code, with maximum check weight $\weight$ and total qubit degree $\qubit$. For $n(A)=L^{D-2}$, the output code $C$ has
\begin{equation}
    n(C)=
    \begin{cases}
        O(\weight^2\qubit^2)L^4, & D=4,\\
        O(\weight^6\qubit^4)L^5, & D=5,
    \end{cases}
\end{equation}
and satisfies
\begin{align}
    k(C)&=k(A), \label{eq:main-k}\\
    d(C)&=\Omega\left(\frac{L}{\weight}d(A)\right), \label{eq:main-d}\\
    \Delta(C)&=\Omega\left(\frac{1}{\weight^2\qubit}\Delta(A)\right). \label{eq:main-delta}
\end{align}
For bounded $\weight,\qubit$ and input with $k(A),d(A),\Delta(A)=\Theta(n(A))$, the resulting explicit 4D and 5D codes attain the optimal scalings in Eq.~\eqref{eq:BPT-bounds}, without polylogarithmic losses. 
The output has maximum check weight $3(D-1)$ and total qubit degree $2D$ for $D=4,5$, and each grid edge hosts at most 3 (4) qubit degrees of freedom in 4D (5D). 

Note that the lower-dimensional constructions are special cases of the higher-dimensional ones (Fig.~\ref{fig:nesting}). This common construction motivates the name \emph{Layer Codes} for the whole family. 
The additional ingredients required in each dimension are explained in Section~\ref{sec:constraints}.

\begin{figure}[t]
\centering
\begin{tikzpicture}[x=0.72cm,y=0.72cm]
  \draw[fill=black!6,draw=black,line width=0.4pt] (0,0) ellipse (3.9 and 2.4);
  \draw[fill=black!14,draw=black,line width=0.4pt] (-0.8,-0.35) ellipse (2.5 and 1.55);
  \draw[fill=black!26,draw=black,line width=0.4pt] (-1.75,-0.65) ellipse (1.2 and 0.85);
  \node[align=center] at (-1.75,-0.62) {\small 3D};
  \node[align=center] at (0.35,-0.30) {\small 4D};
  \node[align=center] at (2.55,0.30) {\small 5D};
\end{tikzpicture}
\caption{\label{fig:nesting}A nested family of layer codes. Each input qubit is replaced by a two-dimensional surface-code patch. The patch positions form a line in 3D, a plane in 4D, and a three-dimensional grid in 5D. The lower-dimensional arrangements are special cases of the same construction.}
\end{figure}

Although realizing higher-dimensional connectivity in a three-dimensional device is challenging, the modular structure of Layer Codes offers a route toward such implementations through networks of surface-code modules. 
Higher dimensions also improve the asymptotic scaling of the qubit overhead, while finite-size instances become straightforward to examine due to its deterministic nature, as shown in the next subsection.

\subsection{Finite-Size Overhead}
\label{sec:intro-numerics}

As asymptotics can hide non-negligible constants, we numerically count the output qubits for some explicit embeddings of quantum Tanner inputs \cite{leverrier2022quantum} with up to $n(A)=1024$ qubits. Table~\ref{tab:numerics-intro} reports six input sizes from one family. Each output count is divided by its predicted growth factor, $n(A)^{D/(D-2)}$, to show the finite-size coefficient. Appendix~\ref{sec:numerics} specifies the inputs, gives a broader set of examples, and describes the numerical verification.

\begin{table}[t]
\caption{\label{tab:numerics-intro}Output qubit counts for 3-5D Layer Codes on one family of quantum Tanner inputs \cite{leverrier2022quantum}, normalized by $n^{D/(D-2)}$. Here $n=n(A)$ and $k=k(A)$. Appendix~\ref{sec:numerics} specifies the inputs and verification methods.}
\begin{ruledtabular}
\begin{tabular}{rrrrr}
$n$ & $k$ & $n(3\mathrm{D})/n^{3}$ & $n(4\mathrm{D})/n^{2}$ & $n(5\mathrm{D})/n^{5/3}$ \\
\hline
$32$ & $8$ & $1.423$ & $34.3$ & $72.6$ \\
$64$ & $16$ & $1.461$ & $25.9$ & $86.9$ \\
$128$ & $32$ & $1.481$ & $32.2$ & $209.8$ \\
$256$ & $64$ & $1.490$ & $27.9$ & $305.0$ \\
$512$ & $128$ & $1.495$ & $47.0$ & $137.9$ \\
$1024$ & $256$ & $1.498$ & $28.9$ & $455.2$ \\
\end{tabular}
\end{ruledtabular}
\end{table}

The three columns use different normalizations, so their coefficients should not be compared directly as qubit costs. In 3D, the normalized count approaches $3/2$ for this family. The 4D coefficients lie between approximately $26$ and $47$, while the 5D coefficients are larger and more variable. Comparing the underlying output counts, the 4D embedding uses fewer qubits than 3D for all reported inputs; 5D improves on 4D only for some inputs at these sizes. Its better asymptotic exponent competes with finite-size routing costs, making qubit placement and route selection important optimization targets.


\subsection{Outlook}
\label{sec:outlook}

Extending these explicit optimal constructions to $D\geq6$ remains open. Section~\ref{sec:5Dcolor} identifies the filling of closed paths, rather than routing itself, as the remaining obstacle; Appendix~\ref{sec:anyD} gives the general routing construction for all dimensions.

The thermal properties of the new codes remain open. Their growing energy barriers motivate extending the analyses of thermal protection and memory lifetime for 3D layer codes \cite{baspin2025free,gu2025layer,williamson2025partial}, but those analyses do not follow from the barrier bound alone. Decoding performance and practical connectivity costs likewise require separate study.

Color routing may also be useful beyond this family. It provides shared local connections for commuting checks, an ingredient relevant to other code embeddings and to reducing check weight \cite{sabo2024weight,baspin2024wire,yuan2025unified}. Together with improved finite-size layouts, these directions connect the present scaling result to the broader task of building useful local quantum memories.

\section{Overview of Construction}
\label{sec:overview}

This section explains how the explicit construction attains the guarantees of Section~\ref{sec:main-result}. Starting from the 3D arrangement, we identify the geometric constraints and show how shared routes, colorings, and filling surfaces resolve them in four and five dimensions.
Specifically, each input qubit becomes a two-dimensional surface-code patch extending in two transverse directions, denoted by $\hat{x}$ and $\hat{z}$. 
The positions of these patches form a grid in the remaining
\begin{equation}
\label{eq:grid-dim}
\fD=D-2
\end{equation}
directions, which we call the \emph{qubit grid} with axes $\hat{\bm{q}}$. A grid of side length $L$ holds $L^{\fD}$ input qubits, while the output occupies a $D$-dimensional box whose side lengths grow in proportion to $L$. 

\subsection{The Three-Dimensional Starting Point}
\label{sec:overview-3D}

In a 3D Layer Code \cite{williamson2023layer}, the input qubits occupy positions along a line ($\fD=1$), so that checks made of Pauli $X$ operators occupy layers at fixed $\hat{x}$ positions; checks made of Pauli $Z$ operators occupy layers at fixed $\hat{z}$ positions. Fig.~\ref{fig:layer-code-example-3D} shows an example construction.

\begin{figure}[ht]
\centering
\includegraphics[width=0.8\columnwidth]{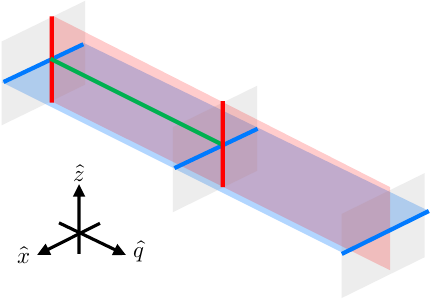}
\caption{\label{fig:layer-code-example-3D}Building a 3D layer code from three input qubits. The input checks are $XXI$ and $ZZZ$, where $I$ denotes no action on a qubit. Grey surfaces replace the input qubits; red and blue surfaces implement the $X$- and $Z$-checks. Green line segments couple the check layers along their intersection. These line defects connect the two qubits shared by the checks and ensure that the resulting local checks commute.}
\end{figure}

A check layer couples only to the qubit patches in that check's support. Where an $X$-check layer meets a $Z$-check layer, additional local couplings are needed to preserve commutation. They form \emph{line defects}: one-dimensional junctions joining pairs of shared qubits. Commuting $X$- and $Z$-checks share an even number of qubits, so those qubits can always be paired. In 3D, the segment connecting each pair is automatically present in both check layers because both extend along the whole input-qubit line.

This layout also controls overlap by giving each check its own transverse position. The underlying network connecting each check's qubits is a line, which has no loops. These convenient properties stop being automatic when the input qubits are arranged on a plane or a three-dimensional grid.

\subsection{Three Constraints and the New Ingredients}
\label{sec:constraints}

To gain the volume advantage of higher dimensions, many check layers must share transverse positions. We begin with a connected \emph{check graph} in the qubit grid $L^{\fD}$: a network of grid edges containing the check's support qubits and their connecting paths. Extending this graph in a transverse direction produces the surfaces of the check layer. Three requirements govern the construction (Table~\ref{tab:jumps}).

\paragraph{Bounded overlap.}
Only a bounded number of layers may use any grid edge or face; otherwise a small region would need to accommodate arbitrarily many qubits or checks. 
We assign different spatial levels\footnote{$\hat{x}$ coordinates for $X$ checks and $\hat{z}$ for $Z$ checks} to check graphs that share edges, while letting nonconflicting graphs reuse a level. This assignment is called \emph{line coloring}: a color represents a level, not a physical property of a qubit. Checks assigned the same level form a \emph{color class}. A related \emph{plane coloring} assigns levels to the filling surfaces introduced below. Both assignments must fit in a transverse extent proportional to the grid side length $L$.

\paragraph{Shared paths.}
Connecting each check's qubits independently can leave an overlapping $X$-check and $Z$-check with incompatible connections. We instead assign a common path $P(q,q')$ to every pair of qubits that occurs together in a check. Every check containing that pair must include the assigned path. Writing $\lambda(x)$ for the graph of an $X$-check $x$ and $\supp x$ for its support, this requirement is
\begin{equation}
\label{eq:oblivious-inclusion}
P(q,q')\subseteq\lambda(x)\quad\text{whenever }q,q'\in\supp x,
\end{equation}
and the same rule applies to $Z$-checks. A path between shared qubits is then already present in the intersection of both layers, where it can support a line defect. We call this \emph{oblivious inclusion}: the construction includes the path before choosing the defect pairing for a particular pair of checks. The common routing assignment may use the whole input code, but it must not change according to which check uses the qubit pair.

\paragraph{Filling loops.}
The collection of shared paths can contain closed loops, also called cycles. Unfilled cycles in a check layer can create logical operators absent from the input code \cite{yuan2025unified}. We therefore add surfaces across these loops, so that every cycle is the boundary of a surface. This is the \emph{filling}, or coning, step. The completed layer must remain connected and contain no cycles that fail to bound a surface. For fixed input check weight and degree, the total number of vertices, edges, and faces used by each check must be $O(L^2)$. The precise algebraic condition, vanishing first homology, is defined in Definition~\ref{def:chain-complex}.

\begin{table*}[t]
\caption{\label{tab:jumps}\label{tab:constraints}What changes with dimension. Two dimensions are reserved for each qubit patch, leaving a line, plane, or three-dimensional qubit grid. Color routing supplies compatible connections while controlling overlap. The filling step removes unwanted logical operators. The row-and-column graphs used in 4D are defined in Section~\ref{sec:4Dcolor}; the intermediate points and star planes used in 5D are defined in Section~\ref{sec:5Dcolor}.}
\begin{ruledtabular}
\begin{tabular}{p{0.15\textwidth}p{0.23\textwidth}p{0.25\textwidth}p{0.25\textwidth}}
Requirement & 3D & 4D & 5D \\
\hline
Bounded overlap & A separate position for each check & Color check graphs and filling planes & Color check graphs and filling surfaces \\[1ex]
Shared paths & The common qubit line & Rows and columns through support qubits & Coordinate lines through support qubits and intermediate points \\[1ex]
Filling loops & No loops to fill & A full qubit plane & Unions of coordinate planes through routing points (star planes) \\
\end{tabular}
\end{ruledtabular}
\end{table*}

The 4D and 5D constructions below realize each requirement with an explicit routing and filling rule.

\subsection{Shared Paths and Color Routing in Four Dimensions}
\label{sec:4Dcolor}

Place the input qubits on a square grid of side length $L$, with coordinate directions $\hat{a}$ and $\hat{b}$. 
Fix an ordering of the qubits, say the lexicographic ordering. For each pair, route from the earlier qubit to the later one, first along $\hat{a}$ and then along $\hat{b}$. We call this an \emph{AB path}; the fixed convention ensures that every check uses the same path for that pair.

For a qubit at $q=(a,b)$, include its full row and column. Their union is the \emph{AB graph},
\begin{equation}
\label{eq:AB-graph}
\AB(q)=(\{a\}\times[L])\cup([L]\times\{b\}),
\end{equation}
where $[L]=\{1,\ldots,L\}$. For a check $x$, take $\lambda(x)=\AB(x)$, where $\AB(x)=\bigcup_{q\in\supp x}\AB(q)$ is the union over its support qubits. The AB path between any two support qubits lies in the union of their two graphs. Consequently every check containing that pair automatically includes the same path, satisfying Eq.~\eqref{eq:oblivious-inclusion}. 

We then assign spatial levels to whole check graphs. Graphs sharing an edge receive different levels; graphs without shared edges can occupy the same level. The assignment preserves every check's connected structure and its shared paths. Only $O(L)$ levels are needed for fixed input weight and degree (Theorem~\ref{thm:line-coloring-4D}). The $X$- and $Z$-checks can be assigned levels independently because their paths have already been fixed in the common qubit grid. Thus color routing preserves the paths needed for line defects while keeping the transverse extent proportional to $L$.

The remaining loops fit within the full qubit plane, which contains only $O(L^2)$ cells. For each check, we add a copy of that plane at a suitable level. Every loop in its extended AB graph can be moved to the added plane and filled. Different filling levels keep the surfaces from accumulating in the same place (Theorem~\ref{thm:plane-coloring-4D}).

\subsection{Example: 4D Embedding of the $\llb 9,1,3\rrb$ Shor Code}
\label{sec:Shor}

In this example, $\eta_X(x)$ and $\eta_Z(z)$ denote the transverse levels assigned to the $X$- and $Z$-check graphs. 
The filling levels are denoted by $\omega_X(x)$ and $\omega_Z(z)$. A repetition code, used below to extend a graph along a transverse direction, enforces equality of neighboring bits along a line (Example~\ref{ex:rep}).

\begin{figure}[ht]
\centering
\includegraphics[width=0.8\columnwidth]{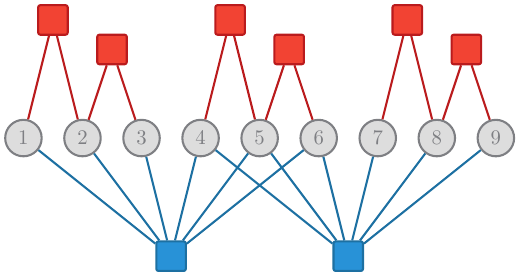}
\caption{The $\llb9,1,3\rrb$ Shor code used as input. Grey circles are qubits, blue boxes are $X$-checks, and red boxes are $Z$-checks. Lines indicate check support.}
\label{fig:Shor-code}
\end{figure}

\begin{figure}[ht]
\centering
\subfloat[\label{fig:Shor-qubits}]{%
    \centering
    \includegraphics[width=0.3\columnwidth]{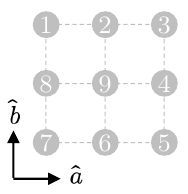}
}
\subfloat[$\eta_X(x)=1\in \hat{x}$\label{fig:X-check-1}]{%
    \centering
    \includegraphics[width=0.3\columnwidth]{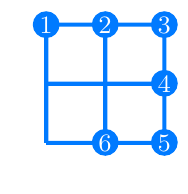}
}
\subfloat[$\eta_X(x)=2\in \hat{x}$\label{fig:X-check-2}]{%
    \centering
    \includegraphics[width=0.3\columnwidth]{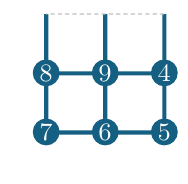}
}
\\
\subfloat[$\eta_Z(z)=1\in \hat{z}$\label{fig:Z-check-1}]{%
    \centering
    \includegraphics[width=0.3\columnwidth]{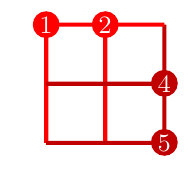}
}
\subfloat[$\eta_Z(z)=2\in \hat{z}$\label{fig:Z-check-2}]{%
    \centering
    \includegraphics[width=0.3\columnwidth]{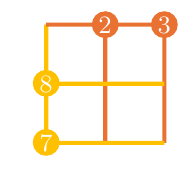}
}
\subfloat[$\eta_Z(z)=3\in \hat{z}$\label{fig:Z-check-3}]{%
    \centering
    \includegraphics[width=0.3\columnwidth]{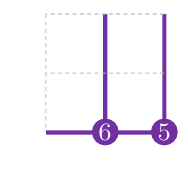}
}
\\
\subfloat[$\eta_Z(z)=4\in \hat{z}$\label{fig:Z-check-4}]{%
    \centering
    \includegraphics[width=0.3\columnwidth]{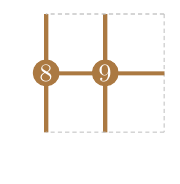}
}
\caption{AB graphs for the Shor-code embedding. (a) Qubit placement on $[3]^2$. (b)--(c) The $X$-check graphs at coordinates $\eta_X\in[2]$. (d)--(g) The $Z$-check graphs grouped at coordinates $\eta_Z\in[4]$. Graphs sharing a coordinate have no common edge; crossings at vertices do not introduce couplings between the corresponding checks. Drawing colors distinguish checks and are separate from the coordinate assignments.}
\label{fig:Shor-check}
\end{figure}

The Shor code in Fig.~\ref{fig:Shor-code} stores one logical qubit in nine physical qubits and has distance three. Place its qubits on $[3]^2$, using the spiral arrangement in Fig.~\ref{fig:Shor-qubits}. The two $X$-checks receive coordinates $\eta_X(x)\in[2]$. The $Z$-checks can be grouped into four color classes, with coordinates $\eta_Z(z)\in[4]$. Figure~\ref{fig:Shor-check} shows the corresponding AB graphs grouped by their assigned levels.

For an $X$-check, extend $\{\eta_X(x)\}\times\AB(x)$ along $\hat{z}$ by taking a product with a repetition code. This gives the vertical surfaces in Fig.~\ref{fig:Shor-contract}, but leaves any cycles of $\AB(x)$ unfilled. We complete the layer with a horizontal copy of $[3]^2$ at $\hat{z}$ coordinate $=\omega_X(x)$, where $\omega_X$ distinguishes checks within each line-color class. A cycle can then be moved along the vertical surfaces to that level and filled in the horizontal plane. Section~\ref{sec:contracting-cycles-4D} gives the formal argument.

\begin{figure}[ht]
\centering
\subfloat[$\eta_X(x)=1\in \hat{x}$\label{fig:Shor-contract-1}]{%
    \centering
    \includegraphics[width=0.45\columnwidth]{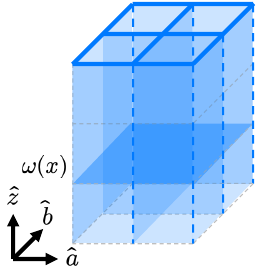}
}
\subfloat[$\eta_X(x)=2\in \hat{x}$\label{fig:Shor-contract-2}]{%
    \centering
    \includegraphics[width=0.45\columnwidth]{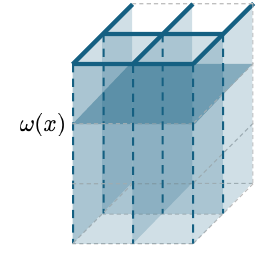}
}
\caption{Filling cycles in the Shor-code $X$-check layers. The vertical surfaces are products of the AB graphs in Figs.~\ref{fig:X-check-1} and~\ref{fig:X-check-2} with the repetition code along $\hat{z}$. A horizontal qubit plane is adjoined at level $\omega_X(x)$, allowing cycles to be moved to that level and filled. Dashed vertical lines mark the support qubit layers.}
\label{fig:Shor-contract}
\end{figure}

There is enough room to assign these filling levels without changing the scaling. By Theorem~\ref{thm:plane-coloring-4D}, a line-color class contains at most $L$ checks, so $L$ values of $\omega_X$ suffice. The pair $(\eta_X(x),\omega_X(x))$ separates their filling planes. The same construction applies to $Z$-checks with the transverse directions exchanged. For this particular placement of the Shor-code qubits, the $Z$-check AB graphs have no cycles, so their additional filling planes are unnecessary.

The qubit layers extend in $\hat{x}\times\hat{z}$. Each $X$-check layer couples to its support qubits along the corresponding $\hat{z}$ lines; each $Z$-check layer couples along $\hat{x}$. For each overlapping $X$--$Z$ pair, we join the paired common qubits by the AB paths already included in both graphs. These become the line defects shown in Fig.~\ref{fig:Shor-defect}. All couplings occur locally along these prescribed intersections.

\begin{figure}[ht]
\centering
\subfloat[$1\times 1\in \hat{x}\times \hat{z}$\label{fig:Shor-defect-11}]{%
    \centering
    \includegraphics[width=0.3\columnwidth]{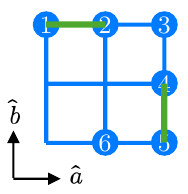}
}
\subfloat[$1\times 2\in \hat{x}\times \hat{z}$\label{fig:Shor-defect-12}]{%
    \centering
    \includegraphics[width=0.3\columnwidth]{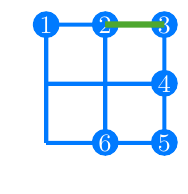}
}
\subfloat[$1\times 3\in \hat{x}\times \hat{z}$\label{fig:Shor-defect-13}]{%
    \centering
    \includegraphics[width=0.3\columnwidth]{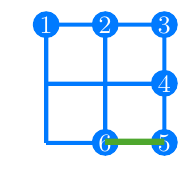}
}
\\
\subfloat[$2\times 1\in \hat{x}\times \hat{z}$\label{fig:Shor-defect-21}]{%
    \centering
    \includegraphics[width=0.3\columnwidth]{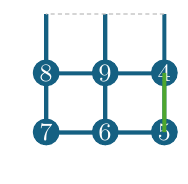}
}
\subfloat[$2\times 2\in \hat{x}\times \hat{z}$\label{fig:Shor-defect-22}]{%
    \centering
    \includegraphics[width=0.3\columnwidth]{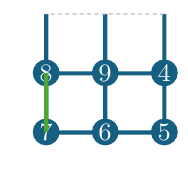}
}
\subfloat[$2\times 3\in \hat{x}\times \hat{z}$\label{fig:Shor-defect-23}]{%
    \centering
    \includegraphics[width=0.3\columnwidth]{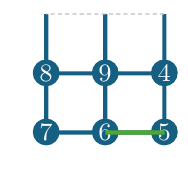}
}
\\
\subfloat[$2\times 4\in \hat{x}\times \hat{z}$\label{fig:Shor-defect-24}]{%
    \centering
    \includegraphics[width=0.3\columnwidth]{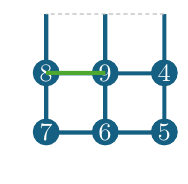}
}
\caption{Line defects for overlapping Shor-code checks. Each panel shows the shared AB paths in the qubit plane at fixed transverse coordinates $(\eta_X(x),\eta_Z(z))$. The green paths lie in both check layers. Compare the corresponding check graphs in Fig.~\ref{fig:Shor-check}.}
\label{fig:Shor-defect}
\end{figure}

\subsection{Intermediate Points and Fillings in Five Dimensions}
\label{sec:5Dcolor}

In 5D, the qubit grid has three dimensions. The counterpart of an AB graph is the \emph{star graph}: the three coordinate lines through a point. Two such graphs need not meet --- the opposite corners of a cube provide an example. 
The endpoint graphs alone therefore no longer connect every support pair.

We resolve this by inserting two intermediate points between each pair of endpoints, as illustrated in Fig.~\ref{fig:color-route-5D}. Consecutive points lie in a common coordinate plane, so the AB rule connects them using their star graphs. The intermediate points are chosen for the entire collection of qubit pairs appearing in any input check. A coloring construction prevents too many routes from concentrating at any one point (Theorem~\ref{thm:color-route-5D}). Every check containing a given pair then includes the same intermediate points and connecting path. This enforces the same oblivious-inclusion rule as in 4D. A separate line coloring places the resulting check graphs at transverse levels; $O(L)$ levels suffice for bounded input check weight and degree (Theorem~\ref{thm:line-coloring-5D}).

Filling the loops requires a second change. A full three-dimensional grid would cost $O(L^3)$ cells per check, exceeding the $O(L^2)$ budget. Instead, take the three coordinate planes through each routing point; their union is its \emph{star plane}. For bounded input check weight, each check uses only a bounded number of routing points, so these planes contribute $O(L^2)$ cells in total. We prove that the union used by a check is connected and fills every cycle (Theorem~\ref{thm:contracting-cycles-5D}). Assigning levels to these surfaces controls their overlap and completes the check layers.

This construction recovers the lower-dimensional cases naturally. In a two-dimensional qubit grid, each pair already shares a coordinate plane, so intermediate points are unnecessary and the filling reduces to the full plane. In a one-dimensional grid, all graphs reduce to the same line and there are no loops to fill.

The distinction between routing and filling also identifies the remaining higher-dimensional obstacle. Routes with controlled overlap exist in any fixed dimension (Appendix~\ref{sec:anyD}). However, coordinate-plane fillings in larger qubit grids can meet in disconnected pieces or fail to meet at all, leaving unwanted cycles. Extending the layer-code construction to $D\geq6$ requires a filling rule that removes these cycles within the same volume and overlap constraints.

\section{Preliminaries}
\label{sec:prelim}

\subsection{General}
\begin{definition}
    We write $[n]=\{1,...,n\}$ and $[n)=\{1,...,n-1\}$.
    We also write $[a,b]$ to only denote the integer values within the interval and similarly for $[a,b),(a,b],(a,b)$.
\end{definition}

\begin{definition}[Coordinates]
    Given $q\in [L]^{\fD}$, let $q|_{i}$ denotes its $i$th-coordinate. 
    Also if $S\subseteq [\fD]$, then let $q|_{S}$ denote the tuple $(q|_{i},i\in S)$. Also write $q|_{\ne i}$ to denote the tuple of all coordinates except for $q|_{i}$.
\end{definition}

\subsection{Complexes}
\begin{definition}
\label{def:chain-complex}
A \textbf{(chain) complex} $C$ (of \textbf{length} $n$) is a sequence of $\dF_2$-vector spaces $C_{i}$ (whose elements are referred as $i$-\textbf{chains}) together with linear $\partial_{i}:C_{i} \to C_{i-1}$, called the \textbf{differentials} of $C$, such that $\partial_{i}\partial_{i+1}=0$ where the subscripts are often omitted. We write
\begin{equation}
    C = C_n\cdots \to  C_{i} \xrightarrow{\partial_{i}} C_{i-1} \to \cdots  C_0
\end{equation}
Note that $\im \partial_{i+1} \subseteq \ker \partial_i$ for all $i$, and thus the \textbf{$i$-homology} of $C$ is defined as
\begin{equation}
    H_i(C)\equiv  \ker \partial_i/\im \partial_{i+1}
\end{equation}
Denote the equivalences class of $i$-chain $\ell$ as $[\ell]\in H_i(C)$ -- conversely, we may also write $[\ell] \in H_i(C)$ without specifying the representation $\ell$.
\end{definition}

\begin{definition}[Basis]
A complex $C$ is equipped with an \textbf{basis} $\sC$ (which we always assume henceforth) if each $C_i = \dF_2^{n_i}$ with basis elements in $\sC_i$ referred as \textbf{$i$-cells}. 
The  \textbf{(Hamming) weight} $|\ell|$ is then defined for any $\ell \in C_{i}$.
We say that cells $c_{i},c_{i-1} $ are \textbf{adjacent}  if $\bra c_{i-1} |\partial c_{i}\ket \ne 0$, and write $c_i \sim c_{i-1}$. 
\end{definition}

\begin{example}[Repetition Code]
    \label{ex:rep}
    The \textbf{repetition code} on $L$ bits is the classical code with parity checks $Z_i Z_{i+1}$ where $i=1,...,L-1$.
    Its corresponding complex is defined as $R \equiv R(L)$ with differential $\partial^{R}$.
    The 1-cells are denoted via $|i^+\ket$ for $i=1,...,L-1$ where $i^\pm =i\pm 1/2$, and 0-cells via $|i\ket$ for $i=1,...,L$ so that
    \begin{equation}
        \partial^{R} |i^+\ket = |i\ket +|i+1\ket
    \end{equation}
    We implicitly assume that $|i\ket, |i^\pm\ket =0$ if the label is not within the previous parameters.
\end{example}

\begin{example}[Graphs]
    Given a graph $\sG=(\sE,\sV)$, there is an associated complex $G=E\to V$ where $E,V$ are the $\dF_2$ vector spaces generated by $\sE,\sV$ and the differential map is given by the adjacency relation. 
    Conversely, a complex $G:E\to V$ with basis can define a graph $\sG=(\sE,\sV)$ if $|\partial e|=2$ for all 1-cells $e\in E$. In either case, we refer to $G$ as a \textbf{graph} complex with associated graph $\sG$.
    For example, the repetition code $R(L)$ is a graph complex. 
    We further call complex $C=C_2 \to C_1 \to C_0$ a \textbf{cell} complex if $C_1 \to C_0$ is a graph complex and $G$ is a \textbf{subgraph} of $C$ if $G$ is a subgraph of $C_1\to C_0$.
\end{example}

\begin{definition}[Cocomplex]
Let $C$ denote a complex.
Since the inner product is nondegenerate, the transpose $\partial_i^\top:C_{i-1} \to C_i$ is well-defined so that that \textbf{cocomplex} is that given by
\begin{equation}
    C^\top = C_n \cdots \leftarrow C_{i} \xleftarrow{\partial_{i}^\top} C_{i-1} \leftarrow \cdots C_0
\end{equation}
The \textbf{$i$-cohomology} of $C$ is $H^i(C) = \ker \partial_{i+1}^\top/\im \partial_i^\top$ which is isomorphic to $H_i(C)$.
\end{definition}

\begin{definition}[Tensor Product]
    Let $A,B$ be two chain complexes; then the \textbf{tensor product} $C = A \otimes B$ is the complex defined as
    \begin{equation}
        C_{k} = \bigoplus_{i+j = k} A_i \otimes B_j
    \end{equation}
    with cells $a\otimes b$ where $a,b$ are cells in $A,B$, respectively, and with differential
    \begin{equation}
        \partial^C (a \otimes b) = \partial^A a \otimes b + a \otimes \partial^{B} b
    \end{equation}
\end{definition}
\begin{remark}
    Note that the tensor product of graph complexes is a cell complex.
\end{remark}
\begin{remark}
    When tensoring a complex $A$ and a cocomplex $B^\top$, we first relabel the cocomplex $B^\top$ as a complex with vector spaces $(B^{\top})_i=B_{n-i}$ where $n$ is the length of $B$.
\end{remark}

\begin{example}[CSS Codes]
    \label{ex:CSS}
    Consider a CSS code as a complex $C=C_2\to C_1 \to C_0$ and convention that $C_2,C_1,C_0$ are the $X$-type checks, qubits, $Z$-type checks, respectively. Let $\weight_X,\weight_Z$ denote the max weight of $X,Z$-checks, respectively, and let $\qubit_X,\qubit_Z$ denote the max number of $X,Z$-checks acting on any qubit, respectively.
    Note that $\weight_X,\qubit_X$ are the max column weights $\|\cdot\|_{\rm{col}}$ of $\partial_2,\partial_1$, respectively, and similarly, $\qubit_Z,\weight_Z$. 
    Hence, we denote the weights via the following diagram
    \begin{equation}
    \label{eq:weight-diagram}
    C_2 \xrightleftharpoons[\qubit_X]{\weight_X} C_1 \xrightleftharpoons[\weight_Z]{\qubit_Z} C_0
    \end{equation}
    If the weights are $O(1)$ for a family of increasing size codes, then we refer to the family as (quantum) \textbf{low-density parity check} or \textbf{(q)LDPC} for short.
    Furthermore, given cells $c_2,c_0$, we shall refer to the \textbf{common qubits} -- denoted as $c_2\wedge c_0$ -- as intersection of supports of $c_2,c_0$, i.e., the collection of cells $c_1$ such that $c_2 \sim c_1 \sim c_0$.
\end{example}

\begin{definition}[Systolic Distance]
    Let $C$ be a complex. Then the $i$\textbf{-systolic distance} $d_i(C)$ is 
    \begin{equation}
        d_i(C) = \min_{\ell:0\ne [\ell]\in H_i(C)} |\ell|
    \end{equation}
    Similarly, define the \textbf{(co)systolic distance} $d^{i}(C)$ via $H^{i}(C)$. Note that when relating to CSS codes, $d_1(C),d^1(C)$ is the $X,Z$-type code distance, respectively.
\end{definition}

\subsection{Graph Coloring}

\begin{definition}[Vertex coloring]
    For a graph $\sG=(\sE,\sV)$, a \textbf{vertex coloring} with $\chi \in \mathbb{N}$ colors is a map $\eta : \sV \to [\chi]$ with the following property: 
    if $v, u \in V$ have the same color $\eta(v) = \eta(u)$ then $(u,v)$ is not an edge in $\sE$.
\end{definition}

\begin{theorem}[Vertex coloring, \cite{diestel2025graph}]
    \label{theorem:vertex-coloring}
    For every graph $\sG$ with maximum degree $\Delta(\sG)$, there exists a vertex coloring for $\sG$ that uses at most $\Delta(\sG) + 1$ colors.
\end{theorem}

\begin{definition}[Edge coloring]
    For a graph $\sG=(\sE,\sV)$, an \textbf{edge coloring} with $\chi \in \mathbb{N}$ colors is a map $\eta : \sE \to [\chi]$ with the following property: 
    if edges $e, e \in \sE$ have the same color $\eta(e) = \eta(e')$ then $e$ and $e'$ do not share a vertex.
\end{definition}

\begin{theorem}[Edge coloring, \cite{vizing1964estimate, diestel2025graph}]
    \label{theorem:edge-coloring}
    For every bipartite graph $\sG$ with maximum degree $\Delta(\sG)$, there exists an edge coloring for $\sG$ that uses at most $\Delta(\sG)$ colors.
\end{theorem}

\subsection{Congestion and Embedding}

\begin{definition}[Congestion]
    Let $\sG = (\sE,\sV)$ be a graph with edges $\sE$ and vertices $\sV$. Let $\Gamma$ be a collection of subgraphs $\gamma$ of $\sG$. Then $\Gamma$ has $c$ \textbf{edge congestion} if for every $e\in E$, there exists at most $c$ subgraphs $\gamma \in \Gamma$ such that $e\in \gamma$.
    If $c=1$, then we say that $\Gamma$ is \textbf{edge decongested}.
\end{definition}

\begin{definition}[Graph Embedding]
    Let $\sG_\infty$ denote an infinite graph of bounded degree and $(\sG_L)$ denote an increasing family of graphs of $O(1)$ max degree. If $\Gamma_L$ are collections of subgraphs of $\sG_L$ with $O(1)$ edge congestion, then $\Gamma_L$ is \textbf{embedded} in $\sG_\infty$.
    We often omit the $L$ subscript for simplicity.
    If $\sG_\infty=\dZ^D$, then we say $\Gamma$ can be embedded in $D$-dimensions for simplicity.
\end{definition}

\begin{remark}
    \label{rem:simplified-embedding-definition}
    Although the previous embedding definition is different from the conventional \cite{li2026almost,portnoy2023local}, it is sufficient for our purposes and simplifies the discussion. Specifically, let $C\gamma$ denote the associated graph complex of $\gamma\in \Gamma$ with 1-cells $|e;\gamma\ket$ and 0-cells $|v;\gamma\ket$ where $e,v$ are edges, vertices of $\gamma$. Here, we refer to $e,v$ as the \textbf{coordinates} of $|e;\gamma\ket$ and $|v;\gamma\ket$, respectively. 
    Let $C\Gamma$ denotes the direct sum (disjoint union). Since $\Gamma$ has $O(1)$-congestion, for any edge $e\in \sG_\infty$, there exists at most $O(1)$ many basis elements in $C\Gamma$ with coordinate equal to $e$.
    Similarly, since $\sG_\infty$ has bounded degree, there exists at most $O(1)$ many basis elements in $C\Gamma$ with coordinate equal to $v$.
\end{remark}

Note that the previous definition and remark naturally extends to complexes. Specifically,

\begin{definition}[Embedding of Complexes]
    \label{def:embedding-of-complexes}
    Let $C$ be a (3-term) complex with basis $\sC$ and with $O(1)$ weights , i.e.,
    \begin{equation}
         C_2 \xrightleftharpoons[O(1)]{O(1)} C_1 \xrightleftharpoons[O(1)]{O(1)} C_0
    \end{equation}
    Let $\gamma\subseteq \sC$ and $C\gamma$ be the (sub-)complex of $C$ generated by $\gamma$.
    Let $\Gamma$ be a collection of $\gamma$. 
    Then $\Gamma$ has $O(1)$ \textbf{edge congestion} in $C$ if for any 1-cell $e$ of $C$, there exists at most $O(1)$ many $\gamma\in \Gamma$ such that $e\in \gamma$; and $\Gamma$ is \textbf{edge decongested} if it has 1 edge congestion.
    Similarly define \textbf{vertex} and \textbf{face congestion}\footnote{Since $C$ has $O(1)$ weights, $O(1)$ edge congestion implies $O(1)$ face (vertex) congestion and vice-versa so that it's not necessary to distinguish between congestions. However, in our Layer code construction, the edge and face congestions can be particularly simple and thus we include the distinction.} for $0$-cells $v$ and $2$-cells $f$ of $C$, respectively.
    Specifically, if $C\Gamma$ is the direct sum of $C\gamma$ for $\gamma \in \Gamma$ with $i$-cells labeled by $|c;\gamma\ket$ where $c$ is an $i$-cell in $\gamma$, then we say that $c$ is the \textbf{coordinate} of $i$-cell $|c;\gamma\ket$.

\end{definition}

\begin{example}[Surface Code]
    \label{ex:surface-code}
    Let $R=R(L)$ denote the repetition code on $L$ vertices so that $R\otimes R^\top$ denote the surface code.  
    Then it's clear that the surface code is embedded in $[L]\times [L+1]$, but for consistency, we say that $R\otimes R^\top$ is embedded in $[L]^2$.
\end{example}

Note that Definition \ref{def:embedding-of-complexes} accounts for the scenario where $C\Gamma$ is the direct sum of $C\gamma,\gamma\in \Gamma$. 
However, if we introduce further adjacency relations between cells in distinct $\gamma,\gamma'$ to obtain an \textit{interacting} complex $C^{\mathrm{int}}\Gamma$, e.g., the Layer Codes, then the definition is slightly insufficient.
Hence, we extend the definition slightly.
\begin{example}[Layer Codes]
    \label{ex:layer-codes}
    In the $D$-dimensional Layer Codes, we say that the output code $C$ is \textbf{embedded} in $D$-dimensions if the $X$-check layers, qubit layers and $Z$-check layers are embedded in $D$-dimensions with $O(1)$ congestion, and that the adjacency relation between cells in distinct type layers (e.g., $X$-type and qubit) is $O(1)$ \textbf{local}, i.e., cells are adjacent only if their coordinates\footnote{For edges, faces, etc, we take the barycenter as its position in $\dR^D$.} are within $O(1)$ $L^1$ distance of each other.  
\end{example}

\subsection{Routing}

\begin{definition}(Routing)
    Let $\sG=(\sE,\sV)$ be a finite graph on $n$ vertices. Let $\sP$ be a set of \textbf{packets}.
    A \textbf{routing problem} is a tuple of maps $\gamma_0,\gamma_\infty:\sP \to \sV$ where $\gamma_0,\gamma_\infty$ are the \textbf{source} and \textbf{destination}, respectively. 
    The routing problem $\gamma_0\to \gamma_\infty$ has \textbf{density} $\rho$ if for any $v\in \sV$ and $t=0,\infty$,
    \begin{equation}
        |\{p\in \sP: \gamma_{t}(p)=v\}| \le \rho
    \end{equation}
    A \textbf{routing algorithm} is a map from $(\gamma_0,\gamma_\infty)$ to a collection of \textbf{routes} $\gamma_t:\sP \to \sV, t\in [0,T]$ such that $t\mapsto \gamma_t(p)$ is a path on $\sG$ and $\gamma_T(p)=\gamma_\infty(p)$ and $T$ is possibly dependent on $p$ and the routing problem.
\end{definition}

\begin{remark}
    In general, a routing algorithm could output non-simple paths $t\mapsto \gamma_t(p)$, e.g., the packet could wait at some vertex for some time, or return to the vertex after some time. 
    Fortunately, this does not occur in this paper.
\end{remark}

\begin{example}[2D Color Routing]
    \label{ex:2D-routing}
    Let $[L]^2$ be a 2D grid with axes $\hat{a}\times \hat{b}$, and $\gamma_0,\gamma_\infty:\sP \to [L]^2$ be injective maps. 
    Then there exists a coloring map $\eta:\sP \to [L]$ (dependent on $\gamma_0,\gamma_\infty$) such that if $\sP_{\eta}$ denotes a partition of $\sP$ with fixed color $\eta(p)=\eta$, then the 2D \textit{greedy} routing algorithm is edge decongested when routing packets $\sP_\eta$ for any fixed $\eta$, i.e., the routes $\gamma(p),p\in \sP_\eta$ are edge decongested.

    Specifically, the 2D greedy routing algorithm is such that it moves each packet $p$ from its source $\gamma_0(p)$ to its destination $\gamma_\infty(p)$ monotonically, first along the rows $\hat{a}$, and then along the columns $\hat{b}$, in parallel and independently of other packets.
\end{example}

\begin{proof}
    Let $r_1,...,r_L$ denote the rows and $c_1,...,c_L$ denote the columns of the 2D grid. 
    Consider a (multi-edge) bipartite graph with vertices consisting of rows and columns, and edges $r_i c_j$ if there exists a packet $p\in \sP$ with source in row $r_i$ and destination in column $c_j$.
    Note we allow for multiple edges with the same endpoints $r_i, c_j$ if there exist multiple packets $p$ from row $r_i$ to column $c_j$.
    Since $\gamma_0,\gamma_\infty$ are injective, we see that the multi-edge graph has a max degree of $L$.
    By the edge coloring theorem of bipartite graphs in Theorem \ref{theorem:edge-coloring}, the edges can be colored by at most $L$ colors such that edges of the same color do not share a vertex.
    Since the edges are in one-to-one correspondence with the packets, we see that $\sP$ can be partitioned into at most $L$ partitions $\sP_\eta$.
    By definition, the packets $p\in \sP$ start from distinct rows and end up in distinct columns.
    Therefore, in the first stage of the 2D greedy routing algorithm of packets $\sP_\eta$, all row edges are traversed at most once and no column edges are traversed yet.
    In the second state, all column edges are traversed at most once, while no row edges are traversed.
    Hence, the routing are edge decongested for packets $\sP_\eta$ for any $\eta$.
\end{proof}

Despite being motivated by the 2D greedy routing algorithm, the proof proves a somewhat stronger statement. Specifically, 
\begin{lemma}[2D Color Routing -- Reformulated]
    \label{lem:2D-routing-reformulated}
    Consider the setup in Lemma \ref{ex:2D-routing}. Then the rows and columns $\row \gamma_0, \col \gamma_{\infty}: \sP_{\eta} \to [L]$ are injective.
    In particular, if 
    \begin{align}
        \Row \gamma_0(p) &=  [L]\times \row \gamma_0(p)\\
        \Col \gamma_\infty(p) &= \col\gamma_\infty (p)\times [L]\\
        \AB(\gamma_0,\gamma_\infty)(p) &= \Row \gamma_0(p) \cup \Col \gamma_\infty(p)
    \end{align}
    Then the collection of $\AB(\gamma_0,\gamma_\infty)(p),p\in \sP_{\eta}$ (regarded as subgraphs) are edge decongested in $[L]^2$.
\end{lemma}

\begin{remark}
    Note that the 2D routes in Example \ref{ex:2D-routing} are included in the collection of AB graphs.
    This reformulation will be the essence of the 4D embedding construction,
    while further extensions (see Theorem \ref{thm:color-route}) will be the essence of higher dimensional embeddings.
\end{remark}

\subsection{Mayer-Vietoris}

\begin{definition}
    Let $A,B$ be complexes with \textbf{possible overlap}, i.e., $\partial^A = \partial^B$ on $A_i \cap B_i$. Then the \textbf{intersection} 
    \begin{equation}
        A\cap B=\cdots A_i\cap B_i \to A_{i-1}\cap B_{i-1}\to \cdots
    \end{equation}
    is a well-defined complex with differential $\partial^{A\cap B}$ defined as the restriction of $\partial^A=\partial^B$ on $A_i\cap B_i$. Similarly, the \textbf{union}
    \begin{equation}
        A+B=\cdots \to A_i+B_i \to A_{i-1}+B_{i-1} \to \cdots
    \end{equation}
    is a well-defined complex with differential $\partial$ defined as follows: given $c\in A_i+B_i$ with representation $c=a+b$ for $a\in A_i,b\in B_i$, let $\partial c= \partial^A a+\partial^B b$, and note that $\partial$ is defined independent of the representation $c=a+b=a'+b'$.
\end{definition}

\begin{remark}
    The usual direct sum $A\oplus B$ can be regarded as the disjoint union of $A,B$.
\end{remark}

\begin{example}[Graphs]
    Let $G^\alpha:E^\alpha \to V^\alpha$ for $\alpha=A,B$ denote graph complexes such that their associated graphs $\sG^{\alpha}$ may overlap $\sG^{A}\cap \sG^{B}=(\sE^A \cap \sE^B, \sV^A\cap \sV^B)$. Then the intersection $G^A\cap G^B$ and union $G^A + G^B$ are the graph complexes associated with $\sG^A\cap \sG^B$ and $\sG^A \cup \sG^B$, respectively.
    In particular, the intersection and union of cell complexes are also a cell complexes.
\end{example}

\begin{lemma}[Mayer-Vietoris, Theorem 6.3 \cite{rotman2009introduction}]
    \label{lem:mayer-vietoris}
    Let $A,B$ be complexes with possible overlap. Then the following is a short exact sequence
    \begin{equation}
        0\to A\cap B \xrightarrow[]{\iota} A\oplus B\xrightarrow[]{\pi} A+B \to 0
    \end{equation}
    where
    \begin{equation}
        \iota(a)=(a,a), \quad \pi(a,b) =a+b
    \end{equation}
    In particular, the following is a (long) exact sequence (for all $s$)
    \begin{equation}
    \begin{tikzpicture}[baseline]
    \matrix(a)[matrix of math nodes, nodes in empty cells, nodes={minimum size=25pt},
    row sep=1.5em, column sep=1.5em,
    text height=1.25ex, text depth=0.25ex]
    {  H_s(A\cap B) & H_s(A\oplus B) & H_s(A+B)   \\
       H_{s-1}(A\cap B) &\cdots & \cdots \\};
    \path[->,font=\scriptsize]
    (a-1-1) edge node[above]{$[\iota]$} (a-1-2)
    (a-1-2) edge node[above]{$[\pi]$} (a-1-3)
    (a-1-3) edge node[below]{$\fd$} (a-2-1)
    (a-2-1) edge (a-2-2)
    (a-2-2) edge (a-2-3);
    \end{tikzpicture}
    \end{equation}
    where $[\iota],[\pi]$ denotes the induced maps of $\iota,\pi$ via passing through quotients, e.g., $[\iota][a]=[\iota(a)]$, and $\fd$ is the connecting map defined as $\fd[a+b]=[\partial^A a]$, and $H_s(A\oplus B)\cong H_s(A) \oplus H_s (B)$
\end{lemma}

\begin{remark}
    One can also prove the previous lemma using mapping cones, since $A+B$ corresponds to \textit{gluing} $A$ to $B$ along their intersection $A\cap B$. However, the notation is much more tedious and asymmetric (e.g., $B$ to $A$ is also a valid choice).
\end{remark}

\begin{example}[Contracting Cycles on Cylinders]
    \label{ex:gauging-cycles}
    Let $G=E\to V$ be the graph complex of a triangle and $R$ be the repetition code on $L$ vertices so that $G\otimes R$ is a \textit{cylinder-like} cell complex. Let $C$ be the cell-complex of the triangle (with a single face generating the boundary triangle) situated at the top of the cylinder, i.e.,
    \begin{equation}
        (G\otimes R) \cap C=G\otimes |0\ket \cong G
    \end{equation}
     Note that $H_1(G\otimes R) \cong \dF_2$ and $H_1 (C) \cong 0$. By Lemma \ref{lem:mayer-vietoris}, $H_1(G\otimes R+C)=0$ and thus the nontrivial cycle in $C$ is now \textit{contractible} due to the additional $T$.
\end{example}

The previous example is formalized as follows.
\begin{lemma}[Contracting Cycles on Cylinders]
    \label{lem:contracting-cycles}
    Let $G=E\to V$ be a connected graph complex, and $R$ be the repetition code on $L$ vertices so that $G\otimes R$ is a (cylinder-like) cell complex. 
    Let $C$ be a 3-term cell complex (with faces, edges and vertices) such that $G$ is a subgraph of $C$ and $H_1(C)=0$. 
    Then 
    \begin{equation}
        H_1(G\otimes R +C\otimes |i\ket) = 0
    \end{equation}
    where $|i\ket$ is a 0-cell of $R$. Moreover, if $H_2(C) =0$, then $H_2(G\otimes R+C\otimes |i\ket)=0$
\end{lemma}
\begin{proof}
    Assign $A = G \otimes R$ and $B = C\otimes |i\ket$ in Lemma \ref{lem:mayer-vietoris}. Note that 
    \begin{align}
        H_1(A\cap B) &= H_1(G\otimes |i\ket) \cong H_1(G) \\
        H_1(A\oplus B) &\cong H_1(A) \cong H_1(G)
    \end{align}
    Hence, $[\iota]$ acts as the identity from $H_1(A\cap B) \to H_1(A\oplus B)$ so that $\ker [\pi ] =\im [\iota] = H_1(A\oplus B)$. 
    Hence, $[\pi]$ must be the zero map $H_1(A\oplus B) \to H_1(A+B)$.
    Again, since the sequence is exact, we see that the connecting map $\fd$ is injective.
    However, since $H_0(A\cap B) \cong H_0(G)$ corresponds to arbitrary vertices in the connected components of $G$, they cannot be in the image of the differential $\partial$ and thus $\fd$ is the zero map.
    Specifically, given $a+b\in \ker \partial^{A+B}$, we see that $\partial^A a= \partial^B b\in A\cap B$.
    Since $a$ is a 1-chain (collection of paths), we see that $\partial^{A} a$ must be an even collection of vertices in $A\cap B$.
    Since $A\cap B =G\otimes |i\ket$ is connected, we see that $[\partial^A a]=0$ and thus the connecting map $\fd=0$.
    Hence, $H_1(A+B) = \ker \fd = \im [\pi] =0$.
    The case is similar for the second homology.
\end{proof}

We will also need the following lemma.
\begin{lemma}[Union]
    \label{lem:union-homology}
    Let $A,B$ be cell complexes such that $A\cap B$ is connected.
    If $A,B$ have trivial first homology, then so does the union $A+B$.
\end{lemma}
\begin{proof}
    By Lemma \ref{lem:mayer-vietoris}, it's sufficient to show that $\fd=0$ for $s=1$.
    Indeed, if so, then by the long exact sequence, $H_1(A+B)=\ker \fd =\im [\pi] =0$ since $H_1(A\oplus B)=0$.
    Note that given 1-chain $a+b\in \ker \partial^{A+B}$, we see that $\partial^A a=\partial^B b \in A\cap B$.
    Since $A$ is a cell complex, we see that $\partial^A a$ must be even in $A\cap B$.
    Since $A\cap B$ is connected, we see that $[\partial^A a]=0$ and thus the connecting map $\fd=0$ so that the statement follows.
\end{proof}

\subsection{Energy Barrier}
\label{sec:energy}
Let $A=X \to Q \to Z$ be a CSS code with the convention that $X,Q,Z$ denote $X$-checks, qubits, and $Z$-check, respectively, and weights
\begin{equation}
    X \xrightleftharpoons[\qubit_X]{\weight_X} Q \xrightleftharpoons[\weight_Z]{\qubit_Z} Z
\end{equation}
\begin{definition}[Pauli Path]
    An ($X$-type) $c$-\textbf{continuous} with \textbf{destination} $e\in Q$ is a path $\gamma:[0,T]\to Q$ such that
    \begin{align}
        \gamma(0)=0, \quad \gamma(T)=e \\
        |\gamma(t) -\gamma(t-1)| \le c, \quad \forall t\in [T]
    \end{align}
\end{definition}
\begin{remark}
    Note that the conventional Pauli path is such that the step size $|\gamma(t)-\gamma(t-1)|\le 1$, while in our case, we allow a step size of constant size $c$.
    In particular, if $c\ge \weight_X$, then it allows us to derive the exact invariance in Lemma \ref{lem:stabilizer-invariance}, instead of invariance up to $\weight_X \qubit_Z$. 
    Also see Remark \ref{rem:step-size}.
\end{remark}
\begin{definition}
    The ($X$-type) \textbf{energy barrier} $\Delta=\Delta_c$ of an ($X$-type) $c$-continuous Pauli path $\gamma$ is defined as
    \begin{equation}
        \Delta(\gamma)=\max_{t} |\partial\gamma(t)|
    \end{equation}
    The ($X$-type) $c$-\textbf{energy barrier} $\Delta_c$ of $e\in Q$ is defined as
    \begin{equation}
        \Delta_c(e) = \min_{\gamma:0 \to e} \Delta(\gamma)
    \end{equation}
    where the minimum is over all $c$-continuous Pauli paths $\gamma$ with destination $e$. The ($X$-type) $c$-\textbf{energy barrier} of code $A$ is then defined as
    \begin{equation}
        \Delta_c(A) = \min_{\ell\in \ker \partial^{A} \backslash \im \partial^A} \Delta_c(\ell)
    \end{equation}
\end{definition}

\begin{remark}[Relation between Step Size]
    \label{rem:step-size}
    Note that $\Delta_{c'} \le \Delta_{c}$ if $c' \ge c$. Conversely, $\Delta_1\le \Delta_c + c\qubit_Z$, since for any $c$-continuous path $\gamma$, we can always construct a corresponding $1$-continuous path $\gamma_1$ by \textit{filling} in the symmetric difference $\gamma(t-1)+\gamma(t)$ one step at a time between time steps $t-1$ and $t$.
\end{remark}


\begin{lemma}[Sub-additivity]
    \label{lem:sub-additivity}
    Let $e,e'\in Q$. Then
    \begin{equation}
        \Delta_c(e+e') \le \Delta_c(e) +\Delta_c(e')
    \end{equation}
\end{lemma}
\begin{proof}
    Let $\gamma,\gamma'$ be Pauli paths with destination $e,e'$, respectively, such that they obtain the energy barriers, i.e., $\Delta_c(\gamma)=\Delta_c(e)$ and $\Delta_c(\gamma')=\Delta_c(e')$.
    Then define the \textit{concatenation} $\gamma\to \gamma'$ as the Pauli path with destination $e+e'$ which is equal to
    \begin{equation}
        \label{eq:concatenated-path}
        \begin{cases}
            \gamma(t) & t\le T\\
            \gamma(T)+\gamma'(t-T) & T\le t\le T+T'
        \end{cases}
    \end{equation}
    Then note that
    \begin{align}
        \Delta_c(\gamma \to \gamma')&\le \max(\Delta_c(\gamma),|\partial\gamma(T)|+\Delta_c(\gamma'))\\
        &\le \Delta_c(\gamma)+\Delta_c(\gamma')
    \end{align}
    Hence, the statement follows.
\end{proof}

\begin{lemma}[Stabilizer Invariance]
    \label{lem:stabilizer-invariance}
    If $c\ge \weight_X$, then $\Delta_c(\partial s)=0$ for any $s\in X$. In particular,
    \begin{equation}
        \Delta_c(e) = \Delta_c(e+\partial s)
    \end{equation}
\end{lemma}
\begin{remark}
    Since the $c\ge \weight_X$-energy barrier is invariant under $\im \partial$, we can write $\Delta_c([e])$ for any $[e]\in Q/\im \partial$. In particular, we can write $\Delta_c([\ell])$ for $[\ell]\in H_1(A)$.
\end{remark}
\begin{proof}
    Wlog, order the $x$ cells in $s$ arbitrarily as $x_1,..,x_T$ and define the Pauli path with destination $\partial s$ as
    \begin{equation}
        \gamma(t) = \sum_{i\le t} \partial x_i
    \end{equation}
    Since $|\partial x|\le \weight_X$ for all $x$, we see that $\gamma(t)$ is well-defined. Hence, it's clear that $\Delta(\gamma)=0$ and thus $\Delta(\partial s)=0$.
    By sub-additivity, we have
    \begin{equation}
        \Delta(e)\le \Delta(e+\partial s) +\Delta(\partial s)=\Delta (e+\partial s)
    \end{equation}
    And similarly,
    \begin{equation}
        \Delta(e+\partial s) \le \Delta(e)+\Delta(\partial s) =\Delta(e)
    \end{equation}
    Hence, the statement follows.
\end{proof}
\section{4D Embedding}
\label{sec:4Dembed}

In this section, we construct an explicit embedding in 4D hypercubes with axes $\hat{x} \times \hat{a}\times \hat{b} \times \hat{z}$.
Let $A=X\to Q \to Z$ be a CSS code written in the form of a complex, with the $X$, $Q$, and $Z$ modules representing respectively the $X$-type Paulis, qubits and $Z$-type Paulis. We assume these spaces have a selected basis $\sX,\sQ,\sZ$; and the max $X,Z$-check weight is $\weight_X,\weight_Z$, and the max $X,Z$-qubit degree is $\qubit_X,\qubit_Z$, i.e.,
\begin{equation}
X \xrightleftharpoons[\qubit_X]{\weight_X} Q \xrightleftharpoons[\weight_Z]{\qubit_Z} Z
\end{equation}

If $A$ has $|\sQ|= L^{2}$ qubits, we identify $\sQ  = [L]^{2}$ in an arbitrary manner; this gives a natural embedding of the qubits on the 2D grid $[L]^2$ in a fixed manner, with axes $\hat{a}\times \smash{\hat{b}}$.
We also define a total ordering on the 2D grid for the qubits, based on the lexicographic order, i.e., if distinct qubits $q_{i}=(a_{i},b_{i}),i=1,2$ are given, then $q_1 \prec q_2$ if $a_1 < a_2$ or if $a_1=a_2,b_1<b_2$.


\subsection{Check Layers}
\label{sec:check-layer-4D}

Since the qubits are arranged on $[L]^2$ in an arbitrary manner, the checks of $A$ may interact with qubits that are far apart. 
As discussed in Section \ref{sec:overview}, the intuition is thus to partition the $|\sX|\le L^2$ checks into $O(\poly(\weight_X,\qubit_X))L$ partitions, each of which corresponding to a distinct coordinate in the $\hat{x}$ axis, and replace the long-range interactions with short range 2D routes, i.e., the AB routes in Lemma \ref{lem:2D-routing-reformulated}.
Each of these routes will be associated with an $x$-check.
Later, tensoring these strings with a repetition code along $\hat{z}$ will give us the analogues of the $X$-check layers from \cite{williamson2023layer}.
Further modifications of the $x$-check layers, similar to Example \ref{ex:gauging-cycles}, will be necessary so that they do not contain internal logicals (nontrivial cycles) as suggested by the framework in \cite{yuan2025unified} and that the layers still remain embedded in 4D.

The case is similar for $z$-checks with $\hat{x} \lr \hat{z}$ switched, and note that the $X$- and $Z$-type check layers are constructed independently of each other.

\subsubsection{Check Routing}

Despite being motivated by the 2D greedy routing algorithm in Lemma \ref{lem:2D-routing-reformulated}, an arbitrary check $x \in \sX$ can have weight $ > 2$ and thus a more general scheme must be considered. Specifically, 
\begin{definition}[AB Graph]
    \label{def:AB-graph}
    For every qubit $q=(a,b)$, let the \textbf{AB graph} $\AB(q)$ denote the union of the row and column that passes through $q$, i.e.,
    \begin{equation}
        \AB(q) =(a\times [L]) \cup  ([L]\times b)
    \end{equation}
    For every $x$ check, let the \textbf{AB graph} $\AB(x)$ be the union of all $\AB(q),q\sim x$ so that $\AB(x)$ has $O(\weight_X L)$ vertices.
    We will also use $\AB(q), \AB(x)$ to refer to the subgraph that includes the horizontal and vertical edges induced from the grid. As a result, we will also frequently use $\AB$ as a $1$-dimensional complex.
\end{definition}

\begin{figure}[ht]
\centering
\includegraphics[width=0.7\columnwidth]{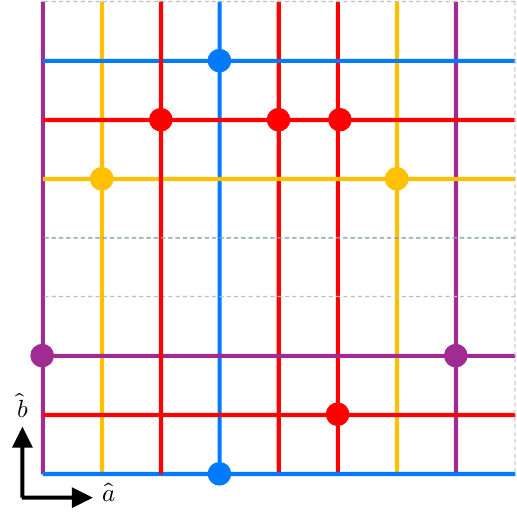}
\caption{Check Routing. Each color denotes the AB route $\AB(x)$ of a distinct $x$ check in a partition $\sX_{\eta}$ of fixed $\eta$ and qubits $q\sim x$ are denoted by dots. 
}
\label{fig:check-routing}
\end{figure}

\begin{theorem}[Line Coloring, Fig. \ref{fig:check-routing}]
    \label{thm:line-coloring-4D}
    There exists a \textbf{(line) coloring} function $\eta=\eta_X:\sX\to [\chi_X]$ where\footnote{Note that it's always possible to boost $\chi_X=\Theta(\weight_X \qubit_X L)$ if necessary.} $\chi_X =O(\weight_X\qubit_X L)$ such that if $\sX_{\eta}$ denotes a partition with fixed color $\eta(x)=\eta$, then the collection of AB routes $\AB(x),x\in \sX_{\eta}$ is edge decongested in $[L]^2$ for every partition $\sX_\eta$.
\end{theorem}

\begin{proof}
    Consider an abstract graph $\sG_X = (\sX, \sE_{\sX})$ with vertices $\sX$ and edges $xx'$ if there exists $q\sim x$ and $q'\sim x'$ such that $q,q'$ share a row or column. 
    Given any $x\in \sX$, we see that there are $\le \weight_X$ qubits $q\sim x$. For each $q\sim x$, there are at most $O(L)$ many $q'$ that share a row or column with $q$. Since the $X$-qubit degree is $\le \qubit_X$, we see that there are at most $O(\weight_X\qubit_X L)$ many $x'$ adjacent to $x$ in this abstract graph, i.e., the max degree of $\sG_X$ is $O(\weight_X \qubit_X L)$.
    By Theorem \ref{theorem:vertex-coloring}, we see that the vertices of $\sG_X$ can be colored using $\chi_X = O(\weight_X \qubit_X L)$ many colors. 
    Let $\eta_X:\sX\to [\chi_X]$ denote the corresponding coloring function and $\sX_{\eta}$ denote the subset of $\sX$ with color $\eta(x)=\eta$.
\end{proof}

By the previous Theorem, for any given $x\in \sX$, let
\begin{equation}
    \eta\AB(x) = \eta(x)\times \AB(x) \subseteq \hat{x} \times \hat{a}\times \hat{b}
\end{equation}
so that $\eta\AB(x),x\in \sX$ is embedded in $[\chi_X]\times [L]^2$. The case is similar for the $Z$-checks with $\chi_Z=O(\weight_Z\qubit_Z L)$ and coloring function $\eta_Z$, so that
\begin{equation}
    \AB\eta(z) = \AB(z) \times \eta(z)\subseteq \hat{a}\times \hat{b} \times \hat{z}
\end{equation}
and that $\AB\eta(z),z\in \sZ$ is embedded in $[L]^2\times [\chi_Z]$.

\begin{remark}[Naive Application]
    Note that if we were to naively apply Lemma \ref{lem:2D-routing-reformulated} for all possible pairs of qubits $\sim x$ in some consisent manner over all $x\in \sX$, then each pair would induce a distinct color. 
    The corresponding routes would thus occupy distinct $\hat{x}$ coordinate, resulting in further difficulties when contracting cycles as shown in Section \ref{sec:contracting-cycles-4D}, and defining the defects in Eq. \eqref{eq:pZX-2-4D}-\eqref{eq:pZX-1-4D}
\end{remark}

\subsubsection{Contracting Cycles}
\label{sec:contracting-cycles-4D}

\begin{figure}[ht]
\centering

\subfloat[$\eta(x)=\eta$\label{fig:check-layer-example-4D}]{%
    \centering
    \includegraphics[width=0.9\columnwidth]{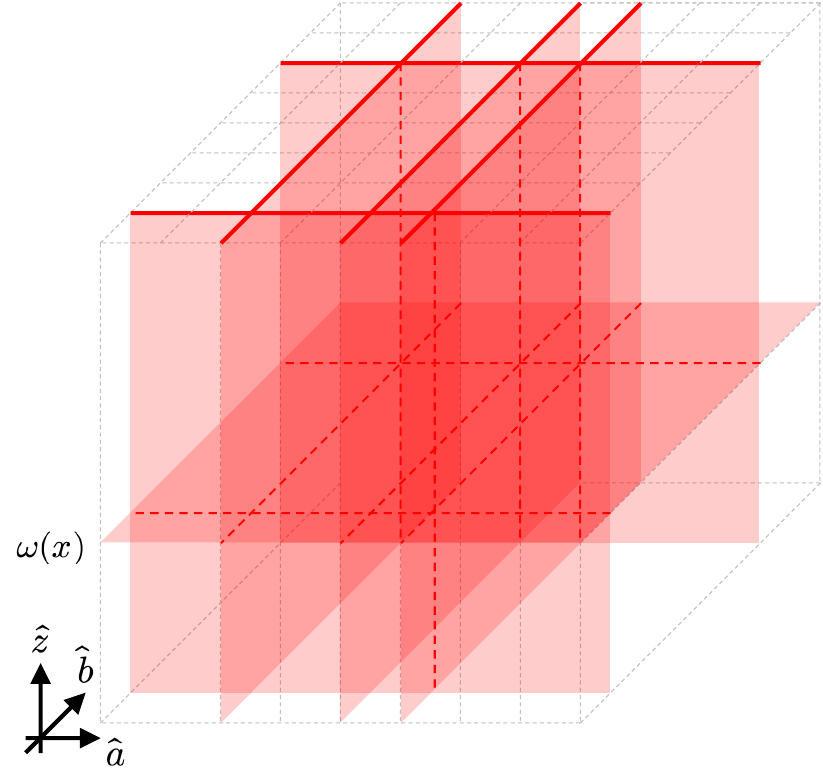}
}
\\
\subfloat[$\eta(x)=\eta$\label{fig:check-layer-congestion-4D}]{%
    \centering
    \includegraphics[width=0.9\columnwidth]{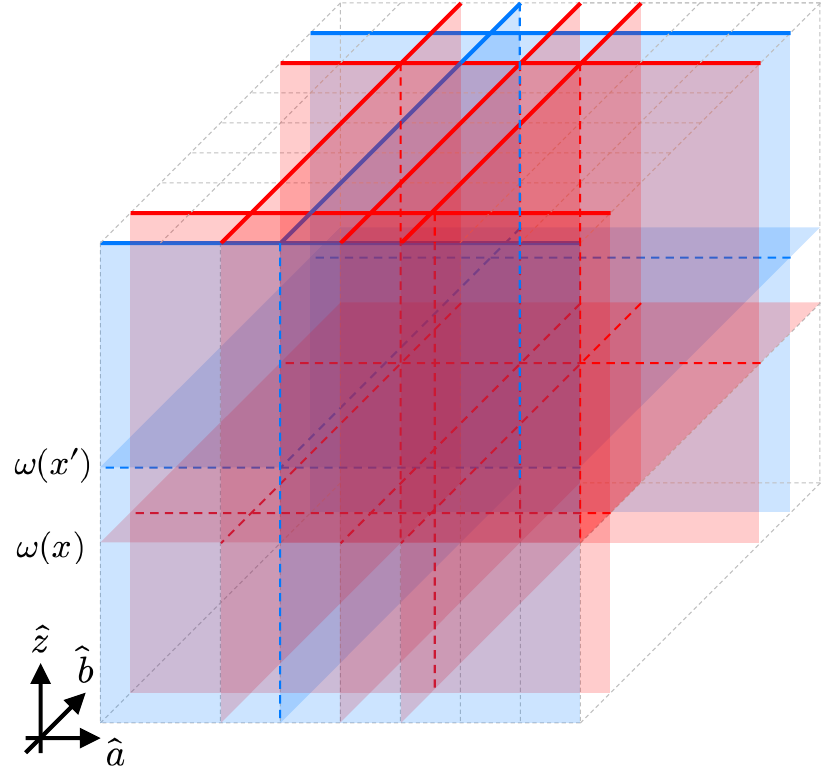}
}
\caption{Check Layer. (a) The vertical $\hat{z}$ layers denote $\eta \AB(x)\otimes R_Z$ where $\AB(x)$ is denoted by red in Fig. \ref{fig:check-routing}, while the vertical dashed lines emphasize positions of the qubits $q\sim x$. The horizontal $\hat{a}\times \hat{b}$ layer denotes $\eta \omega(x)$ as defined in Eq. \eqref{eq:contracting-layer-4D} so that union $C(x)$ has no nontrivial cycles. (b) shows that for distinct $x,x'\in \sX_{\eta}$ (e.g., red and blue checks in Fig. \ref{fig:check-routing}), the contracting layer $\eta\omega(x),\eta\omega(x')$ are distinct. Note that in this example, the blue route $\AB(x')$ does not innately have any cycles, so the contracting layer is technically unnecessary.
}
\label{fig:check-layer-4D}
\end{figure}

Note that $\AB(x)$ may have nontrivial cycles so that $\eta\AB(x)\otimes R_Z$ can have nontrivial internal logicals.
Therefore, in this section, we will attempt to make the cycles contractible by introducing additional faces along some coordinate in the $\hat{z}$ axis, similar to Example \ref{ex:gauging-cycles} and Lemma \ref{lem:contracting-cycles}.
However, we first need the following Lemma
\begin{theorem}[Plane Coloring]
    \label{thm:plane-coloring-4D}
    Let $\sX_{\eta}$ be a partition in Theorem \ref{thm:line-coloring-4D} such that AB graphs $\AB(x)$ do not share the same edge, i.e., edge decongested.
    Then $|\sX_{\eta}| \le L$ for any fixed $\eta$.
    
    In particular, there exists a \textbf{plane coloring} map $\omega=\omega_X:\sX\to [L]$ such that if $\sX_{\eta,\omega}$ is the collection of $x\in \sX_{\eta}$ with fixed \textbf{plane} color $\omega$, then $|\sX_{\eta,\omega}|\le 1$.
    The case is similar for $Z$-check with plane coloring function $\omega=\omega_Z:\sZ \to [L]$
\end{theorem}
\begin{proof}
    Let $\row \AB(x)$ denote the row indices $\AB(x)$ occupies so that $1\le |\row \AB(x)| \le \weight_X$. Since $\AB(x),x\in \sX_{\eta}$ are edge decongested, we see that $\row \AB(x),x\in \sX_{\eta}$ must all be disjoint and thus the statement follows.
\end{proof}

\begin{definition}[Dimension of 4D Grid]
    Let
    \begin{align}
        L_X &= \max(\chi_X,L) = O(\weight_X\qubit_X)L \\
        L_Z &= \max(\chi_Z,L) = O(\weight_Z\qubit_Z)L
    \end{align}
    so that the 4D grid in $\hat{x}\times \hat{a}\times \hat{b}\times \hat{z}$ is given by $[L_X]\times [L]^2\times [L_Z]$.
    Let $R_X,R_Z$ denote the repetition code on $L_X,L_Z$ vertices, respectively.
\end{definition}

\begin{definition}[Contracting Layers]
    \label{eq:contracting-layer-4D}
    For each $x$-check, define the \textbf{contracting layer} as the natural 3-term cell complex defined as
    \begin{equation}
        \eta \omega(x) = \eta(x) \times [L]^2 \times \omega(x) \subseteq \hat{x} \times \hat{a}\times \hat{b} \times \hat{z}
    \end{equation}
    so that the collection $\eta\omega(x),x\in \sX$ is embedded in $[L_X]\times [L]^{2}\times [L_Z]$ with faces decongested\footnote{But also vertex and edge}.
    The case is similar for the $z$-checks.
\end{definition}
\begin{definition}[Check layers, Fig. \ref{fig:check-layer-4D}]
    \label{def:check-layer-4D}
    For a check $x \in \sX$ we define the associated check layer $C(x)$:
    \begin{equation}
        C(x) = \eta \AB(x) \otimes R_Z +\eta \omega(x)
    \end{equation}
    Similarly, define the associated check layer $C(z)$ for $z\in \sZ$
    \begin{equation}
        C(z) = R_X \otimes \AB\eta(z) +\omega\eta(z)
    \end{equation}
\end{definition}
Then by Example \ref{ex:gauging-cycles} or Lemma \ref{lem:contracting-cycles}, we see that $H_1(C(x))=0$ and that $C(x)$ is embedded in $[L_X]\times [L]^2 \times [L_Z]$. The case is similar for $z$ with squares $\omega\eta(z)$ and check layers so that $H_1(C(z))=0$. 

\begin{remark}[Layer Congestion]
    \label{rem:layer-congestion-4D}
    Note that $\eta\AB(x),x\in \sX$ are edge decongested, but after tensoring with $R_Z$, it's possible that $\eta\AB(x)\otimes R_Z$ have $D-2=2$ edge congestion for edges along the $\hat{z}$ axis since at most 2 distinct $\eta\AB(x)$ can share the same vertex.

    Hence, we say that $C(x),x\in \sX$ has 
    \begin{itemize}
        \item 0 edge congestion for edges along the $\hat{x}$ axis
        \item 1 edge congestion for edges within the 2D qubit grid
        \item $D-2=2$ edge congestion for edges along the $\hat{z}$ axis
    \end{itemize}
    The case is similar $Z$-check layers with $\hat{z}\lr \hat{x}$ switched.
\end{remark}

\subsection{Putting Everything Together}

Unlike the check layers, the qubit layers are simply those defined as in the conventional Layer Codes \cite{williamson2023layer}, i.e., a surface code $R_X\otimes R_Z^\top$ for each $q\in \sQ$. 
By Example \ref{ex:surface-code}, the surface code is embedded in $[L_X]\times [L_Z]$ so that the collection of qubit layers $C(q) = R_X\otimes q\otimes R_Z^\top$ is embedded in $[L_X]\times [L]^2 \times [L_Z]$ with edges decongested. 
We can thus follow \cite{yuan2025unified} and define the following
\begin{align}
    \label{eq:CX}
    C^{X} &= \bigoplus_{x\in \sX} C(x)^\top\\
    C^{Q} &= \bigoplus_{q\in \sQ} C(q) \\
    \label{eq:CZ}
    C^{Z} &= \bigoplus_{z\in \sZ} C(z)
\end{align}
Note that $C^{X},C^{Q},C^{Z}$ is embedded in 4D $[L_X]\times [L]^2 \times [L_Z]$.
Hence, the output code $C$ is embedded in 4D, as long as the gluing maps and defect maps are \textit{local}, as we discuss below.

\subsubsection{Gluing Maps}
Similar to the conventional layer codes, we define the gluing maps $g^{QX},g^{ZQ}$ in \cite{yuan2025unified} as follows, which are also depicted by the vertical dashed lines in Fig. \ref{fig:check-layer-4D}.
\begin{align}
    \label{eq:gQX-4D}
    g^{QX} |x;i,q,k^\bullet\ket &=|i, q,k^{\bullet}\ket 1\{i=\eta(x),q\sim x\}\\
    \label{eq:gZQ-4D}
    g^{ZQ} |i^{\bullet},q,k\ket &=\sum_{z\sim q} |i^{\bullet},q,k;z\ket 1\{k=\eta(z)\}
\end{align}
where $i^{\bullet}$ denotes either $i$ or $i^+$ for $i\in [\chi_X]$, and similarly for $k^{\bullet}$, and thus $g^{QX},g^{ZQ}$ act \textit{locally} in 4D.

\subsubsection{Defect Map}

\begin{figure}[ht]
\centering
\subfloat[$\eta(x)\times \eta(z) = \eta_X\times \eta_Z \in \hat{x}\times \hat{z}$\label{fig:defect-pair}]{%
    \centering
    \includegraphics[width=0.7\columnwidth]{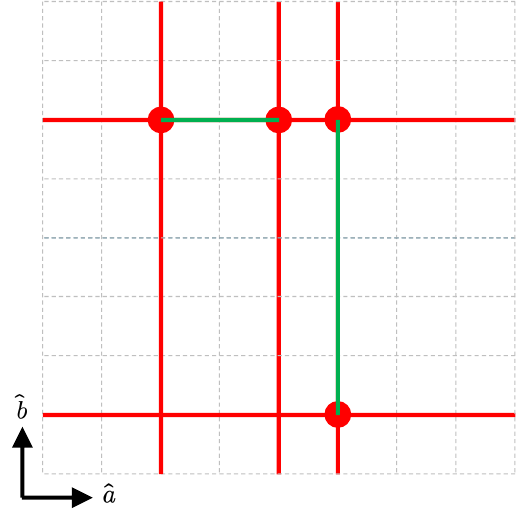}
}
\\
\subfloat[$\eta(x)\times \eta(z) = \eta_X\times \eta_Z \in \hat{x}\times \hat{z}$\label{fig:defect-congestion}]{%
    \centering
    \includegraphics[width=0.7\columnwidth]{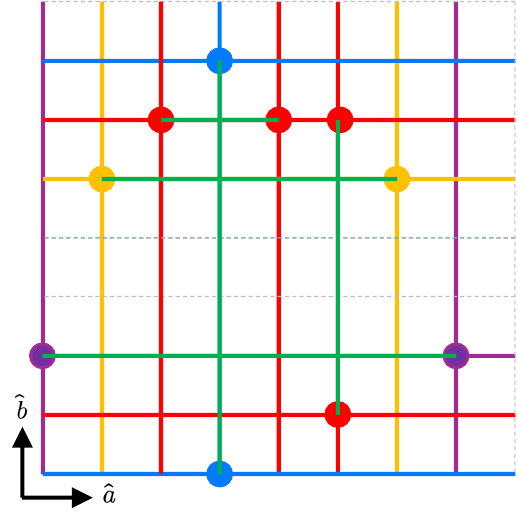}
}
\caption{Defects. (a) The red AB route denotes that of an $x$-check in Fig. \ref{fig:check-routing}. In the case where there exists a $z$-check such that $x\wedge z$ are denote by the red dots, then the green lines denote the defect paths. (b) Although defect lines are vertex disjoint for a fixed adjacent pair $x,z$, they may intersect when considering all adjacent pairs, as indicated by the green lines. However, the edge congestion is bounded by 1.
}
\label{fig:defects}
\end{figure}

Similar to the conventional layer codes, the defect maps $p^{ZX}$ will correspond to strings between pairs of qubits $q_1,q_2\in x\wedge z$ for adjacent $x,z$ checks.
However, due to the 4D hypercube structure (2D grid structure of qubits), the strings are somewhat more complicated. Fortunately, our construction of the check layers exactly accounts for this complication by utilizing the 2D greedy routing algorithm in a novel fashion, i.e., Theorem \ref{thm:line-coloring-4D}.

Given $x\sim z$, we can order the qubits in the overlap $x\wedge z=\{q_1,q_2,...,q_{2M}\}$ based on the lexicographic order. 
Let $\gamma_m(xz)$ denote the AB route in $[L]^2$ from $q_{2m-1} \to q_{2m}$, i.e., we first move in the row direction and then in the column direction. Then it's clear that $\gamma_m(xz) \in \AB(x),\AB(z)$ for all $s$.
Moreover,
\begin{lemma}[Defect Congestion, Fig. \ref{fig:defects}]
    \label{lem:defects-disjoint-4D}
    Fix a pair of adjacent $x,z$ checks. Then $\gamma_m(xz),m=1,...,M$ are (vertex) disjoint. 
    In particular, the collection of \textbf{defect paths}
    \begin{equation}
        \Gamma_m(xz)=\eta(x)\times \gamma_m(xz) \times \eta(z)
    \end{equation}
    over all pairs of adjacent $x,z$ and pairs $m$ must have edges decongested. 
\end{lemma}
\begin{proof}
    Suppose there exists $q=(a,b)\in [L]^2$ which is in $\gamma_i(xz)$ and $\gamma_{j}(xz)$ for $i<j$. 
    If $q_{m}=(a_m,b_m)$ are the qubits in $x\wedge z$, then by definition of the route,
    \begin{align}
        a_{2i-1}\le a\le a_{2i} \\
        a_{2j-1}\le a\le a_{2j}
    \end{align}
    And thus $a=a_{2i}=a_{2j-1}$ by the lexicographic order.
    By definition of $\gamma_m(xz)$ route, 
    \begin{equation}
        b\le b_{2i} < b_{2j-1} =b
    \end{equation}
    and thus we reach a contradiction.
\end{proof}

Parametrize the route $\gamma_m(xz)$ by adding a $t$-dependence, i.e., $t\mapsto \gamma_m(t;xz)$, so that $t=0,T_m\mapsto q_{2m-1}, q_{2m}$, respectively, and $T_m$ is equal to the $L^{1}$ distance between $q_{2m-1},q_{2m}$.
The case is similar for $\Gamma_m(t;xz)$.
In both $\gamma_m,\Gamma_m$, we use half-integers $t^+$ to denote the edges of the route. 
We can then define the defect map as follows
\begin{align}
    \label{eq:pZX-2-4D}
    p^{ZX}|x;iqk\ket &= \sum_{z\sim x} \sum_{m=1}^{|x\wedge z|/2} \sum_{t<T_m} 1\{iqk= \Gamma_m(t;xz)\} \nonumber \\
    &\quad\quad\quad  \times|\Gamma_{m}(t^+;xz);z\ket \\
    \label{eq:pZX-1-4D}
    p^{ZX}|x;i e k\ket &= \sum_{z\sim x}\sum_{m=1}^{|x\wedge z|/2} \sum_{t<T_m} 1\{iek \in \Gamma_m(t^+;xz)\} \nonumber\\
    &\quad\quad\quad \times |\Gamma_{m}(t+1;xz);z\ket
\end{align}
where $e$ denotes an edge (nearest-neighbors) in $[L]^2$. 
Again note that $p^{ZX}$ acts \textit{locally} in 4D, in the sense that, they map cells in $C^X$ to cells in $C^{Z}$, respectively, with the neighboring coordinates, e.g., $\Gamma_m(t^+;xz)$ is an adjacent edge of $\Gamma_m(t;xz)$.
It's then straightforward to check that
\begin{lemma}[Compatibility]
    \label{lem:compatible}
    The defect map $p^{ZX}$ in Eq. \eqref{eq:pZX-2-4D}-\eqref{eq:pZX-1-4D} are compatible with the gluing maps in Eq. \eqref{eq:gQX-4D}-\eqref{eq:gZQ-4D}, i.e.,
    \begin{equation}
        g^{ZQ} g^{QX} = p^{ZX} \partial^{X} + \partial^{Z} p^{ZX}
    \end{equation}
\end{lemma}
\begin{proof}
    Note that 
    \begin{align}
        \partial^{Z} p^{ZX} |x;iqk\ket &= \sum_{z\sim x} \sum_{m=1}^{|x\wedge z|/2} \sum_{t<T_m} 1\{iqk= \Gamma_m(t;xz)\} \nonumber \\
        &\times (|\Gamma_{m}(t;xz);z\ket + |\Gamma_m(t+1;xz)\ket)
    \end{align}
    And that
    \begin{align}
        p^{ZX} \partial^{X} |x;iqk\ket &= \sum_{z\sim x} \sum_{m=1}^{|x\wedge z|/2} \sum_{t<T_m} \sum_{e\sim q} 1\{iek= \Gamma_m(t^+;xz)\} \nonumber\\
        &\quad\quad \times|\Gamma_m(t+1;xz);z\ket
    \end{align}
    Note that if $0<t<T_m$ such that $iqk=\Gamma_m(t;xz)$, then there are exactly two edges $e$ adjacent to $q$ in $[L]^2$ such that $iek=\Gamma_m(s^+;xz)$ for some $s$, i.e., $iek=\Gamma_m(t^-;xz)$ and $=\Gamma_m(t^+;xz)$. Conversely, if $t=0,T_m$, there exists exactly one edge $e$ adjacent to $q$ such that $iek=\Gamma_m(s^+;xz)$ for some $s$, i.e., $iek=\Gamma_m(t^+;xz)$ if $t=0$ and $=\Gamma_m(t^-;xz)$ if $t=T_m$. Hence, $p^{ZX}\partial^{X} + \partial^{Z} p^{ZX}$ maps
    \begin{align}
        |x;iqk\ket &\mapsto \sum_{z\sim x} \sum_{m} 1\{iqk\in \partial\Gamma_m(xz)\} |iqk;z\ket \\
        &= \sum_{z\sim x} 1\{q \in x\wedge z,i=\eta(x),k=\eta(z)\} \\ 
        &\quad\quad  \times |iqk;z\ket
    \end{align}
    where $\partial \Gamma_m(x,z)$ are the boundary points of the string segment in $[L_X]\times [L]^2\times [L_Z]$.
    One can similarly check that the right-hand-side is equal to $g^{ZQ}g^{QX}$ and thus the statement follows.
\end{proof}

\subsection{Main Result}
\begin{theorem}[Logical]
    \label{thm:logical-4D}
    Let $C^{X},C^{Q},C^{Z}$ be defined as in Eq. \eqref{eq:CX}-\eqref{eq:CZ}.
    Let gluing maps $g^{QX},g^{ZQ}$ be defined as in Eq. \eqref{eq:gQX-4D}-\eqref{eq:gZQ-4D}, and defect map $p^{ZX}$ in Eq. \eqref{eq:pZX-2-4D}-\eqref{eq:pZX-1-4D}.
    Then the output, referred as the \textbf{4D Layer Code}, $C$ is a well-defined complex embedded\footnote{See Example \ref{ex:layer-codes}} in the 4D grid $[L_X]\times [L]^2\times [L_Z]$ where each edge hosts at most 3 qubits, and with column code $=A$ so that
    \begin{equation}
        H_s(C) \cong H_s(A), \quad s=2,1,0
    \end{equation}
    And maximum weights ($D=4$)
    \begin{equation}
        C_2 \xrightleftharpoons[2(D-1)]{3(D-1)} C_1 \xrightleftharpoons[3(D-1)]{2(D-1)} C_0
    \end{equation}
    Specifically, the total qubit degree is $2D$.
\end{theorem}

\begin{theorem}[Distance]
    \label{thm:distance-4D}
    Let $C$ be the 4D Layer Code in Theorem \ref{thm:logical-4D} with 1-(co)systolic distance $d_X(C)=d_1(C),d_Z(C)=d^1(C)$, and similarly for the input code $A$. Then
    \begin{align}
        d_X(C) &=\Omega\left(\frac{L_Z}{\weight_X\weight_Z\qubit_Z} \right) d_X(A) \\ 
        d_Z(C) &=\Omega\left(\frac{L_X}{\weight_Z\weight_X\qubit_X} \right) d_Z(A) 
    \end{align}
\end{theorem}

\begin{remark}[Better Bounds]
    \label{rem:better-bounds-4D}
    As discussed in Theorem \ref{thm:line-coloring-4D}, it's always possible to boost $L_Z=\Theta(\weight_Z \qubit_Z L)$ and thus we are guaranteed that
    \begin{equation}
        d_X(C) = \Omega\left(\frac{1}{\weight_X}\right) Ld_X(A)
    \end{equation}
    In fact, in the proof of Lemma \ref{lem:relative-expansion-4D}, we utilized the not necessarily optimal lower bound $|\AB(x)| \ge L$ to account for possible $x$-checks that only act on $\ll \weight_X$ many qubits.
    If all $x$ checks act on $\Omega(\weight)$ qubits and $1\ll \weight_X \ll L^{2/3}$, then by Lemma \ref{lem:AB-size}, we are guaranteed
    \begin{equation}
        d_X(C) = \Omega\left(\frac{1}{\sqrt{\weight_X}}\right) Ld_X(A)
    \end{equation}
    and thus obtain a better scaling compared to the conventional 3D Layer codes \cite{williamson2023layer} which has factor $\Omega(1/\weight_X)$. 
    The case is similar for the $Z$-type (1-cosystolic) distance.
\end{remark}

\begin{theorem}[Energy Barrier]
    \label{thm:energy-4D}
    Let $C$ be the 4D Layer code in Theorem \ref{thm:logical-4D} with $X,Z$-type (1-)energy barrier $\Delta_X(C),\Delta_Z(C)$, and similarly for the input code $A$. Then
    \begin{align}
        \Delta_X(C) =\Omega\left(\frac{1}{\weight_X \qubit_Z \min(\weight_X,\weight_Z)} \right) \Delta_X(A) \\
        \Delta_Z(C) =\Omega\left(\frac{1}{\weight_Z \qubit_X \min(\weight_X,\weight_Z)} \right) \Delta_Z(A)
    \end{align}
\end{theorem}

\section{5D Embedding}
\label{sec:5Dembed}

As discussed previously, the 4D embedding depends on properties intrinsic to the 2D planar grid; mainly any two AB graphs $\AB(q),\AB(q')$ in Definition \ref{def:AB-graph} must intersect and thus contain a route between $q\lr q'$.
However, for higher dimensions, this is not necessarily true (e.g., $q,q'$ situate on opposite corners of a 3D cube).
Hence, an additional idea is necessary to extend the method to higher dimensions.

In this section, we consider the 5D embedding with qubits arranged on the grid $[L]^3$. 
Many essential lemmas in our proof of the 5D embedding generalize in a natural manner to all $D$ dimensions, which we elaborate in Appendix \ref{sec:anyD}. 
Unfortunately, the proof of Theorem \ref{thm:contracting-cycles-5D} utilizes a special property of the 3D grid $[L]^3$, i.e., the union of finite many \textit{star planes} (as in Definition \ref{def:star-planes}) has trivial first homology.
We can neither prove nor disprove that an analogue holds true for higher dimensional grids $[L]^{D-2},D\ge 6$, and thus our proof is uncertain in 6D embeddings and higher.
However, since many essential lemmas generalize to all dimensions, we expect that a more clever proof of Theorem \ref{thm:contracting-cycles-5D} can be constructed in the near future so that the $(D\ge 6)$-dimensional Layer Codes follow immediately from our color routing method.

Similar to the setup for 4D embedding, let $A=X\to Q \to Z$ be a CSS code written in the form of a complex, with the $X$, $Q$, and $Z$ modules representing respectively the $X$-type Paulis, qubits and $Z$-type Paulis. 
We assume these spaces have a selected basis $\sX,\sQ,\sZ$; and the max $X,Z$-check weight is $\weight$, and the max $X,Z$-qubit degree is $\qubit$, i.e.,
\begin{equation}
X \xrightleftharpoons[\qubit]{\weight} Q \xrightleftharpoons[\weight]{\qubit} Z
\end{equation}
If $A$ has $|\sQ|=L^{D-2}$ qubits where $D=5$, we identify $\sQ  = [L]^{D-2}$ in an arbitrary manner; this gives a natural embedding of the qubits on the $(D-2)$-dimensional array $[L]^3$ in a fixed manner.

\subsection{Check Layers}

The construction of the check layers is similar to that of the 4D embedding in Section \ref{sec:check-layer-4D}, except that the defect maps do not \textit{come for free} anymore.
Specifically, since qubits in 4D embedding are arranged on a 2D grid, the AB graphs $\AB(q_1),\AB(q_2)$ must intersect and thus there always exists a path, dependent only on $q_1,q_2$ (and thus exists in both adjacent $x\sim z$ check layers), that connects $q_1\lr q_2$.
However, for the 5D embedding, qubits are arranged on a 3D grid, and thus pairs of qubits $q_1,q_2$ which do not share a face (or equivalently, have the same coordinate along an axes) do not necessarily have intersecting AB graphs.

Therefore, the key observation for 5D embedding is to generalize the AB route by proving the existence of a \textit{color route} as depicted in Fig. \ref{fig:color-route-5D} and formalized in Theorem \ref{thm:color-route-5D}. 
In fact, the existence of a color route holds true for all dimensions as elaborated in Appendix \ref{sec:anyD}, and thus a $(D\ge 6)$-dimensional Layer Code would follow immediately, provided that an analogue of Theorem \ref{thm:plane-coloring-5D} can be proven for $D\ge 6$, i.e., the star planes $\Lambda(x)$ must \textit{remove} all internal logicals in the check layers, so that all cycles in the check layers are contractible (null-homotopic).

\subsubsection{Check Routing}

\begin{definition}[Induced Routing on $Q$]
    \label{def:induced-routing-Q}
    Let $\sG_Q=(\sQ,\sP_Q)$ denote the \textbf{induced graph} on $Q$ with vertices consisting of qubits $\sQ$ and edges (also referred as \textbf{packets}) $p=(q,q')$ if there exists an $x$-check (or $z$-check) such that $q,q'\sim x$ (or $q,q'\sim z$). 
    Note that the max degree is $O(\weight\qubit)$.
    
    Order each packet $p$ arbitrarily.
    Let the \textbf{induced routing problem} on $Q$ be such that $\gamma_{0},\gamma_\infty:\sP_Q \to \sQ$ are the start and end of each (directed) edge $p$.
    By definition, we see that each qubit is the source (destination) of $O(\weight\qubit)$ packets (edges in $\sP_Q$).
\end{definition}

\begin{definition}[Star Graph]
    Let $q\in [L]^3$.
    Then the \textbf{star graph} of $q$ is defined as
    \begin{equation}
        \lambda(q) = \bigcup_{i} q|_{[1,i)} \times [L]\times q|_{(i,3]}
    \end{equation}
    so that $|\lambda(q)| = O(L)$.
    Also define the $S$-\textbf{partial star graph} where $S\subseteq [3]$ as 
    \begin{equation}
        \lambda_{S}(q) = \bigcup_{i \in S} q|_{[1,i)} \times [L]\times q|_{(i,3]}
    \end{equation}
\end{definition}

\begin{figure}[ht]
\centering
\includegraphics[width=0.9\columnwidth]{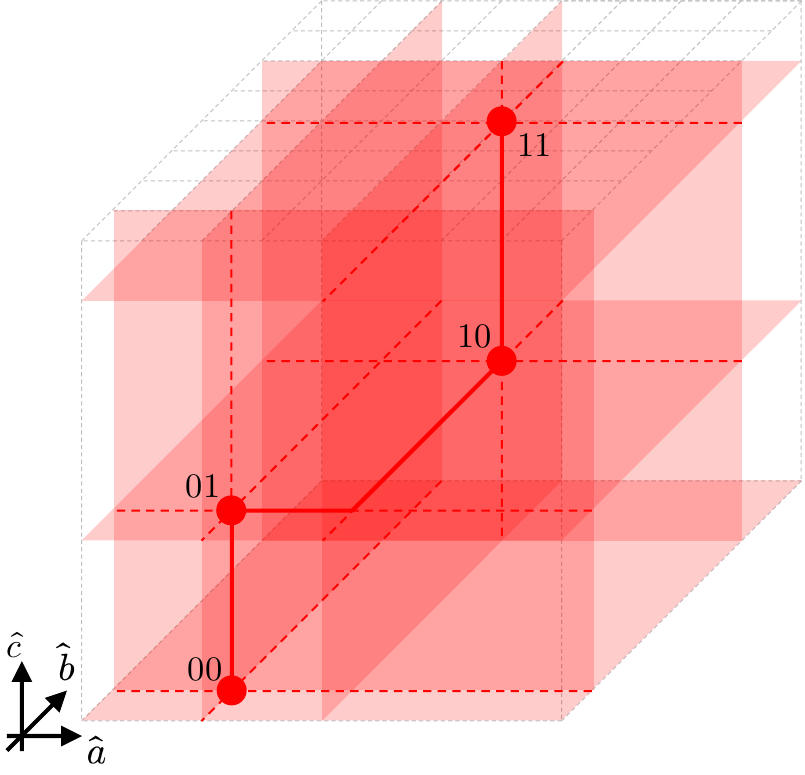}
\caption{Color Route on $[L]^3$. The red dots denote the color route $\gamma_s$ determined by the routing problem $\gamma_0\to \gamma_\infty$. The dashed red lines denote the star graphs $\lambda(\gamma_s(p))$ and the red planes denote the star planes $\Lambda(\gamma_s(p))$. The solid line denotes a route from $\gamma_0\to \gamma_\infty$. Note that by Definition \ref{def:check-layers-5D}, the star planes are included only at $\omega(x)$ along $\hat{z}$ axis (compare with Fig. \ref{fig:check-layer-4D}).
}
\label{fig:color-route-5D}
\end{figure}

\begin{theorem}[Color Route, Fig. \ref{fig:color-route-5D}]
    \label{thm:color-route-5D}
    Consider the induced routing problem $\gamma_0\to \gamma_\infty$ of $\sP_Q$ on $[L]^{3}$. 
    Then for each packet $p\in \sP_Q$, there exists a sequence $\gamma_{s},s\in \dF_2^{2}$, referred as the \textbf{color route}, such that the following properties are satisfied.
    \begin{enumerate}[label=\arabic*)]
        \item (Source and Destination)
        \begin{align}
            \gamma_{00}(p) &= \gamma_0(p)\\
            \gamma_{11}(p) &= \gamma_\infty(p)
        \end{align}
        \item (One Plane at a Time) Continguous points in the sequence share the same plane, i.e., $\gamma_{00},\gamma_{01}$ share the same $[2]$-coordinate, and similarly, for $\gamma_{10},\gamma_{11}$, while $\gamma_{01},\gamma_{10}$ share the same $3$-coordinate.
        \item (Density) For any $s\in \dF_2^{2}$ and $q\in [L]^{3}$, 
        \begin{equation}
            |\{p:\gamma_s(p)=q\}| =O(\weight \qubit)
        \end{equation}
    \end{enumerate}
\end{theorem}
\begin{proof}
    Write 
    \begin{equation}
        \gamma_0(p) = a_0\times b_0\times c_0, \quad \gamma_\infty(p) = a_\infty \times b_\infty \times c_\infty, 
    \end{equation}
    Then the color route is given by 
    \begin{align*}
        \gamma_{00}(p) &= \gamma_0(p) =  a_0 \times b_0 \times c_0 \\
        \gamma_{10}(p) &= a_0 \times b_0 \times c_1 \\
        \gamma_{01}(p) &= a_\infty \times b_\infty \times c_1 \\
        \gamma_{11}(p) &= \gamma_\infty(p)= a_\infty \times b_\infty \times c_\infty \\
    \end{align*}
    For an intermediate position $c_1 \in [L]^2$ dependent on $p$ that we define as follows.

    Consider a (multi-edge) bipartite graph $\sH$ with vertex partitions $\sV_{0},\sV_\infty =[L]^2$, and an edge between $i_0\times j_0\in \sV_0$ and $i_\infty\times j_\infty \in \sV_\infty$ if there exists a packet $p \in \sP_Q$ with source $a_0\times b_0=i_0 \times j_0$ and destination $a_\infty \times b_\infty = i_\infty \times j_\infty$.
    Since each source and destination has $O(\weight \qubit)$ packets, we see that the (multi-edge) bipartite graph has max degree $O(\weight \qubit)L$.
    By Theorem \ref{theorem:edge-coloring}, we see that the edges can be colored with $O(\weight\qubit)L$ many colors so that edges with the same color do not share a vertex.
    Note that there is a one-to-one correspondence between edges in the multi-edge graph and packet $p\in \sP_{Q}$.
    By grouping at most $O(\weight \qubit)$ colors together in an arbitrary manner, we can define a coloring map $p\mapsto c_1(p):\sP_{Q} \to [L]$ such that the following subsets of packets satisfy
    \begin{equation}
        |\{p: (a_t\times b_t  \times c_1)(p)=q\}| = O(\weight \qubit)
    \end{equation}
    For any $q\in \sQ$ and $t=0,\infty$.
    Let $\gamma_{10}(p) = (a_0\times b_0 \times c_1)(p)$ and $\gamma_{01}(p) = (a_\infty\times b_\infty\times c_1)(p)$ as before.
    The statement then follows.
\end{proof}

\begin{definition}[Check Star Graphs]
    For every $x$-check, define the \textbf{star graph} for $x$ as
    \begin{equation}
        \lambda(x) = \bigcup_{s, p\sim x} \lambda(\gamma_s(p))
    \end{equation}
    where the \textbf{adjacency} $p\sim x$ is shorthand for the condition that the endpoints $q,q'$ of $p$ are both contained in the support of $x$.
    In particular,
    \begin{equation}
        |\lambda(x)| = O(\weight^2 )L
    \end{equation}
    The case is similar for $\lambda(z)$ for $z$ checks
\end{definition}

\begin{theorem}[Line Coloring]
    \label{thm:line-coloring-5D}
    There exists a \textbf{line coloring} function $\eta=\eta_X:\sX\to [\chi_X]$ where\footnote{Note that it's always possible to boost $\chi_X$ to saturate its upper bound if necessary.} 
    \begin{equation}
        \label{eq:line-color}
        \chi_X =O(\weight^3 \qubit^2)L
    \end{equation}
    such that if $\sX_{\eta}$ denotes a partition with fixed color $\eta(x)=\eta$, then the collection of star routes $\lambda(x),x\in \sX_{\eta}$ has edges decongested in $[L]^3$ for every partition $\sX_\eta$.
\end{theorem}
\begin{proof}
    Similar to Theorem \ref{thm:line-coloring-4D}, consider an abstract graph $\sG_X = (\sX, \sE_{\sX})$ with vertices $\sX$ and edges $xx'$ if $\lambda(x),\lambda(x')$ share an edge.
    We write $\lambda(x)\sim_e \lambda(x')$ to denote that the two share an edge. Then
    \begin{equation}
        \{x':\lambda(x) \sim_e \lambda(x')\}\subseteq \bigcup_{s,p\sim x} \{x':\lambda(\gamma_s(p))\sim_e \lambda(x')\}
    \end{equation}
    Fix $s,p$ so that $q=\gamma_s(p)$. Our goal is then to bound the size of  $\{x':\lambda(q)\sim_{e} \lambda(x')\}$. Note that
    \begin{align}
        |\{\lambda(q)\sim_{e} \lambda(x')\}|&\le \sum_{x'} 1\{\lambda(q)\sim_{e}\lambda(x')\} \\
        &\le \sum_{x',s',p'} 1\{p'\sim x'\} \\
        &\quad\quad \times 1\{\lambda(q) \sim_{e} \lambda(\gamma_{s'}(p'))\} \\
        &\le \qubit \sum_{s',p'} 1\{\lambda(q) \sim_{e} \lambda(\gamma_{s'}(p'))\}
    \end{align}
    Note that $\lambda(q)$ and $\lambda(\gamma_{s'}(p'))$ share an edge only if $\gamma_{s'}(p')=q$ except possibly at the $i$-coordinate for some $i\in [3]$.
    Hence,
    \begin{align}
        \sum_{p'}1\{\lambda(q) \sim_{e} \lambda(\gamma_{s'}(p'))\} &\le \sum_{i,p'} 1\{\gamma_{s'}(p')|_{\ne i}=q|_{\ne i}\} \nonumber \\
        &= \sum_{i,p'} 1\{\gamma_{s'}(p')|_{\ne i}=q|_{\ne i}\} \nonumber\\
        &\times \sum_{c\in [L]} 1\{\gamma_{s'}(p')|_{i} = c\}\\
        &=\sum_{i} \sum_{c\in [L]} O(\weight \qubit) \\
        &=O(\weight \qubit L)
    \end{align}
    where we utilized property (3) of the color routes.
    Hence,
    \begin{align}
        |\{x':\lambda(q)\sim_{e} \lambda(x')\}| &= O(\weight \qubit^2)L \\
        |\{x':\lambda(x) \sim_e \lambda(x')\}| &= O(\weight^3 \qubit^2)L
    \end{align}
    By Theorem \ref{theorem:vertex-coloring}, we see that the vertices of $\sG_X$ can be colored using $\chi_X$ many colors so that adjacent vertices do not share the same color.
    Let $\eta_X:\sX\to [\chi_X]$ denote the corresponding coloring function and $\sX_{\eta}$ denote the subset of $\sX$ with color $\eta(x)=\eta$.
    Then the statement follows.
\end{proof}

Similar to the 4D scenario, by the previous Theorem, for any given $x\in \sX$, define
\begin{equation}
    \eta\lambda(x) =\eta(x)\times \lambda (x) \subseteq [\chi_X]\times [L]^{3}
\end{equation}
Then it's clear that $\eta\lambda(x),x\in \sX$ is embedded in $\hat{x} \times \hat{q}$ with edges decongested.
The case is similarly defined for the $Z$-checks with line coloring function $\eta=\eta_Z:\sZ \to [\chi_Z]$ so that
\begin{equation}
    \eta\lambda(z) =\lambda(z) \times \eta(z) \subseteq  [L]^{3} \times [\chi_Z]
\end{equation}
over $z\in \sZ$ is embedded in $\hat{q}\times \hat{z}$ with edges decongested.

\subsubsection{Contracting Cycles}

Let us denote the planes induced by filling in the star graphs $\lambda(q)$ via $\Lambda(q)$, which are naturally 3-term cell complexes.
The filled planes $\Lambda(x)$ are similarly defined.
Specifically,
\begin{definition}[Star Planes]
    \label{def:star-planes}
    Let $q\in \sQ=[L]^{3}$. Then define the \textbf{star plane} of $q$ as
    \begin{equation}
        \Lambda(q) = \bigcup_{i<j} q|_{[1,i)}\times [L] \times q|_{(i,j)} \times [L] \times q|_{(j,3]}
    \end{equation}
    so that $|\Lambda(q)|=O(L^2)$ in terms of vertices.
    Also define $S$-\textbf{partial star plane} where $S\subseteq [3]$ as
    \begin{equation}
        \Lambda_{S}(q) = \bigcup_{i,j\in S:i<j} q|_{[1,i)}\times [L] \times q|_{(i,j)} \times [L] \times q|_{(j,3]}
    \end{equation}
\end{definition}

\begin{definition}[Star Planes of Checks]
    Let $\gamma_s(p)\in [L]^{3},s\in \dF_2^{2}$ denote the color route given by the induced routing problem in Theorem \ref{thm:color-route-5D}. For each $x$-check, further define the \textbf{star plane} as
    \begin{equation}
        \Lambda(x) = \bigcup_{s,p\sim x} \Lambda (\gamma_s(p))
    \end{equation}
    And similarly for the $z$ checks.
\end{definition}
\begin{theorem}[Plane Coloring]
    \label{thm:plane-coloring-5D}
    Let $\sX_{\eta}$ be a partition in Theorem \ref{thm:line-coloring-5D} such that star graphs $\lambda(x)$ do not share the same edge, i.e., are edge decongested.
    Then there exists a \textbf{plane coloring} map $\omega=\omega_X:\sX\to [\zeta_X]$ where
    \begin{equation}
        \zeta_X = O(\weight^2)L
    \end{equation}
    such that if $\sX_{\eta,\omega}$ is the collection of $x\in \sX_{\eta}$ with fixed \textbf{plane} color $\omega$, then the collection $\Lambda(x),x\in \sX_{\eta,\gamma}$ is face decongested\footnote{Compare with the 4D case in Theorem \ref{thm:plane-coloring-4D}}.
    The case is similar for $Z$-check with plane coloring function $\omega=\omega_Z:\sZ \to [\zeta_Z]$
\end{theorem}
\begin{proof}
    Consider the abstract graph $\sG_{X,\eta}=(\sX_{\eta},\sE_{X,\eta})$ with edges $(x,x')$ if $\Lambda(x),\Lambda(x')$ share a face, which we write as $\Lambda(x) \sim_{f} \Lambda(x')$.
    Note that the number of $x'$ adjacent to $x$ is upper bounded by the following.
    \begin{equation}
        \sum_{s,p\sim x} \underbrace{\sum_{x'} 1\{\Lambda(\gamma_s(p))\sim_{f} \Lambda(x')\}}
    \end{equation}
    Hence, our goal is to upper bound the underbraced term where we fix $q=\gamma_s(p)$.
    
    Note that if $\Lambda(q),\Lambda(x')$ share a face, then there exists $s'$ and $p'\sim x'$ such that $\Lambda(q),\Lambda(\gamma_{s'}(p'))$ share a face.
    Hence, there exists $i<j$ such that $q|_{\ne ij}=\gamma_{s'}(p')|_{\ne ij}$ and thus
    \begin{align}
        \{\Lambda(q)\sim_{f} \Lambda(x')\} &\subseteq \bigcup_{s'} \bigcup_{i<j} \\
        &\quad \bigcup_{p'\sim x'} \{q|_{\ne i,j}=\gamma_{s'}(p')|_{\ne i,j}\}
    \end{align}
    Let $I_{s,ij}(x') \subseteq [L]$ denote the collection of coordinates $\gamma_{s'}(p')|_{i}\in [L]$ over $p'\sim x'$ with $\gamma_{s'}(p')|_{\ne ij}=q|_{\ne ij}$. Then note that
    \begin{equation}
        \bigcup_{p'\sim x'} \{q|_{\ne i,j}=\gamma_{s'}(p')|_{\ne i,j}\} \subseteq \{I_{s,ij}(x') \ne \varnothing\}
    \end{equation}
    By definition, we see that for $x'$ of the same line color, we must have $I_{s',ij}(x')$ be disjoint (this is similar to Theorem \ref{thm:plane-coloring-4D}) and thus for a given $s',i,j'$, there can at most be $L$ many $x'$ of the same line color with $I_{s',ij}(x')\ne \varnothing$. 
    Hence, the max degree of $\sG_{X,\eta}$ is upper bounded by
    \begin{equation}
        \sum_{s,p\sim x} \sum_{s'} \sum_{i<j} L =O(\weight^2)L
    \end{equation}
    The statement then follows by applying Theorem \ref{theorem:vertex-coloring} to the abstract graph $\sG_{X,\eta}$.
    
\end{proof}

\begin{definition}[Dimension of 5D Grid]
    Let $\chi_X,\chi_Z$ be that defined in Theorem \ref{thm:line-coloring-5D} and $\zeta_X,\zeta_Z$ be that in Theorem \ref{thm:plane-coloring-5D}. Define
    \begin{align}
        L_X &= \max(\chi_X,\zeta_Z) = O(\weight^3 \qubit^2)L \\
        L_Z &= \max(\chi_Z,\zeta_X) = O(\weight^3 \qubit^2)L
    \end{align}
    so that the 5D grid is given by $[L_X]\times [L]^3\times [L_Z]$.
\end{definition}

\begin{definition}[Contracting Layers]
    For each $x$-check, define the \textbf{contracting layer} as
    \begin{equation}
        \eta\Lambda\omega(x)=\eta(x)\times \Lambda(x) \times \omega(x)
    \end{equation}
    so that the collection $\eta\Lambda\omega(x),x\in \sX$ is embedded in $[L_X]\times [L]^{3}\times [L_Z]$ with faces decongested.
    The case is similar for the $z$-checks.
    Let $R_X,R_Z$ denote the repetition codes on $L_X,L_Z$ vertices, respectively.
\end{definition}

\begin{definition}[Check Layers]
    \label{def:check-layers-5D}
    Define the $x$-check layer as
    \begin{equation}
        C(x) = \eta\lambda(x)\otimes R_Z +\eta\Lambda\omega(x)
    \end{equation}
    Similarly, define the $z$-check layer as
    \begin{equation}
        C(z)=R_X\otimes \lambda\eta (z) +\omega\Lambda\eta(z)
    \end{equation}
\end{definition}

\begin{remark}[Layer Congestion]
    \label{rem:layer-congestion-5D}
    Similar to the 4D case in Remark \ref{rem:layer-congestion-4D}, note that $\eta\lambda(x),x\in \sX$ has edge decongested, but after tensoring with $R_Z$, it's possible that $\eta\lambda(x)\otimes R_Z$ have $D-2=3$ edge congestion for edges along the $\hat{z}$ axis since at most 3 distinct $\eta\AB(x)$ can share the same vertex.

    In contrast to the 4D case, the contracting layers are not vertex disjoint, and instead only face decongested.
    Hence, any edge in the 3D qubit grid can be shared by $D-3=2$ many contracting layers.
    Hence, $C(x),x\in \sX$ has
    \begin{itemize}
        \item 0 edge congestion for edges along the $\hat{x}$ axis
        \item $\max(1,D-3)=2$ edge congestion for edges in the $(D-2)$-dimensional qubit grid
        \item $D-2=3$ edge congestion for edges along $\hat{z}$ axes
    \end{itemize}
    The case is similar for $z$-check layers with $\hat{x}\lr \hat{z}$ switched.
\end{remark}


\begin{figure}[ht]
\centering
\includegraphics[width=0.9\columnwidth]{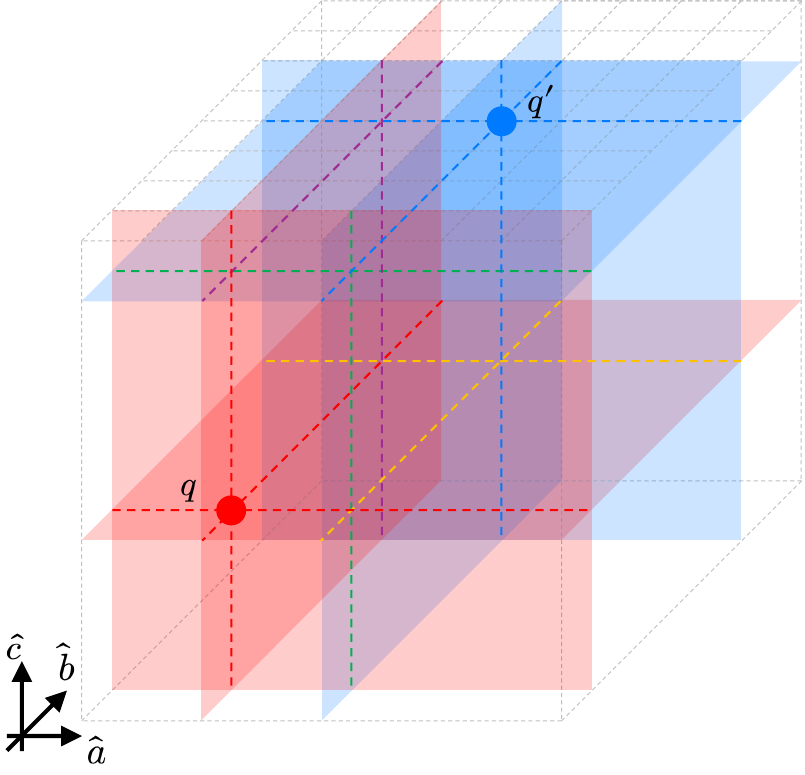}
\caption{Contracting Cycles on $[L]^3$. The red planes (lines) denote the star plane (graph) of $q$, while the blue planes (lines) denote that of $q'$. The purple, yellow and green dashed lines denote the intersection $\Lambda(q)\cap \Lambda(q')$.
}
\label{fig:contracting-cycles-5D}
\end{figure}

Similar to the 4D embedding scenario, the existence of contracting layers $\eta\Lambda\omega(x)$ further promotes 
\begin{theorem}[Contracting Cycles, Fig. \ref{fig:color-route-5D}, \ref{fig:contracting-cycles-5D}]
    \label{thm:contracting-cycles-5D}
    For any $x$- or $z$-check, 
    \begin{equation}
        H_1(C(x))=H_1(C(z))=0
    \end{equation}
\end{theorem}
\begin{proof}
    Note that the construction is symmetric in $X$- and $Z$-checks and thus we shall restrict our attention to $X$-checks.
    By Example \ref{ex:gauging-cycles} or Lemma \ref{lem:contracting-cycles}, we see that it's sufficient to show that $H_1(\Lambda(x))=0$.
    Note that $\Lambda(x)$ is the union of $\Lambda(q)$ over certain qubits, and thus it's sufficient to show that the finite union of any $\Lambda(q)$ has trivial first homology in $[L]^3$.

    We first \textbf{claim} (whose proof is postponed to the end) that for any qubits $q,q'\in [L]^3$, the intersection $\Lambda(q)\cap \Lambda(q')$ is connected and contains 
    \begin{equation}
        \Lambda(q)\cap \Lambda(q') \supseteq  \bigcup_{k} \lambda_{ij}(q'|_{ij}\times q|_{k})
    \end{equation}
    where $\lambda_{ij}$ is the partial star graph and the union is over all $k\in [3]$ and $i,j$ are the other indices in $[3]$. 
    Using this claim, consider $\Lambda(q)\cap \Lambda(q'')$ for another qubit $q''$ so that it also contains
    \begin{equation}
        \lambda_{12}(q'|_{12}\times q|_{3})
    \end{equation}
    Since $\lambda_{12}(q'|_{12}\times q|_{3})$ and $\lambda_{12}(q''|_{12}\times q|_{3})$ are within the same face, i.e., the 3-coordinate is fixed $=q|_3$, we see that they must intersect, and thus $\Lambda(q)\cap (\Lambda(q')\cup\Lambda(q''))$ is connected.
    In particular, given any collection of qubits $q_1,...,q_m$, we see that $\Lambda(q)\cap \Lambda(q_i)$ is connected and that any pair $\Lambda(q)\cap \Lambda(q_i),\Lambda(q)\cap \Lambda(q_j)$ has nonempty intersection.
    Hence, $\Lambda(q)\cap (\Lambda(q_1)\cup \cdots \cup \Lambda(q_m))$ is connected.

    By induction, we can assume that $\Lambda(q_1)\cup \cdots \cup \Lambda(q_m)$ has trivial first homology, where $m=1$ is the base case.
    Then by Lemma \ref{lem:union-homology}, we see that $\Lambda(q)\cup \Lambda(q_1)\cup \cdots \cup \Lambda(q_m)$ has trivial first homology, and thus the statement follows.

    What remains is then to prove the claim.
    Note that if $q,q'$ do not share a face, i.e., $q|_{i}\ne q'|_{i}$ for any $i$-coordinate, then (as depicted in Fig. \ref{fig:contracting-cycles-5D})
    \begin{equation}
        \Lambda(q)\cap \Lambda(q')=\bigcup_{i,j} q|_{i} \times [L]\times q'|_{j}
    \end{equation}
    where the union is over all $i,j\le 3$. In particular, we can write
    \begin{equation}
        \Lambda(q)\cap \Lambda(q') = \bigcup_{k} \lambda_{ij}(q'|_{ij}\times q|_{k})
    \end{equation}
    where the union is over all $k$ and $i,j$ are the remaining indices in $[3]$. Note that $\lambda_{12}(q'|_{12}\times q|_3)$ contains $q'|_{1}\times q|_{23}$ and similarly, $\lambda_{13}(q'|_{13}\times q|_2)$. Hence, $\lambda_{12}(q'|_{12}\times q|_3)$ and $\lambda_{13}(q'|_{13}\times q|_2)$ are connected. By permuting the indices, it's straightforward to check that $\Lambda(q)\cap \Lambda(q')$ is connected.
    The other possibilities can also be similarly derived. 
    Specifically, let $q,q'$ share a face, but not a line, i.e., $q|_{k}=q'|_{k}$ for some index $k$, say $k=1$, and not equal for the remaining indices $ij=23$. Then
    \begin{align}
        \Lambda(q)\cap \Lambda(q') &= \Lambda_{23}(q) \\
        & \cup \lambda_{12}(q'|_{12}\times q|_3) \cup \lambda_{23}(q'|_{13}\times q|_{2})
    \end{align}
    Again note that all three terms in the union are individually connected and contain $q'|_{1}\times q|_{23}=q$. Hence, $\Lambda(q)\cap \Lambda(q')$ is connected.
    Now let $q,q'$ share a line, but not the same point, so that $q|_{ik}=q'|_{ik}$ some indices say $ik=13$, but not for 2-coordinate. Then
    \begin{equation}
        \Lambda(q)\cap \Lambda(q') = \Lambda_{23}(q) \cup \Lambda_{12}(q)
    \end{equation}
    which is connected since the terms in the union both contain $q$. Finally, if $q=q'$, then $\Lambda(q)\cap \Lambda(q')=\Lambda(q)$ is trivially connected.
\end{proof}

\subsection{Putting Everything Together}
\label{sec:together-5D}
Unlike the check layers, the qubit layers are simply those defined as in the conventional Layer Codes \cite{williamson2023layer}, i.e., a surface code $R_X\otimes R_Z^\top$ for each $q\in \sQ$ where $R_X,R_Z$ are repetition codes on $L_X,L_Z$. 
By Example \ref{ex:surface-code}, the surface code is embedded in $[L_X]\times [L_Z]$ so that the collection of qubit layers $C(q) = R_X\otimes q\otimes R_Z^\top$ is embedded in $[L_X]\times [L]^3 \times [L_Z]$ with edge decongested. We can thus define the following
\begin{align}
    \label{eq:CX}
    C^{X} &= \bigoplus_{x\in \sX} C(x)^\top\\
    C^{Q} &= \bigoplus_{q\in \sQ} C(q) \\
    \label{eq:CZ}
    C^{Z} &= \bigoplus_{z\in \sZ} C(z)
\end{align}
Note that $C^{X},C^{Q},C^{Z}$ is embedded in the 5D grid.
Hence, the output code $C$ is embedded in the 5D hypercube, as long as the gluing maps and defect maps are \textit{local} in 5D, as we discuss below.

\subsubsection{Gluing Maps}
Similar to the conventional layer codes, we define the gluing maps $g^{QX},g^{ZQ}$ as follows (analogue of the vertical dashed lines in Fig. \ref{fig:check-layer-4D}).
\begin{align}
    \label{eq:gQX-5D}
    g^{QX} |x;i,q,k^\bullet\ket &=|i, q,k^{\bullet}\ket 1\{i=\eta(x),q\sim x\}\\
    \label{eq:gZQ-5D}
    g^{ZQ} |i^{\bullet},q,k\ket &=\sum_{z\sim q} |i^{\bullet},q,k;z\ket 1\{k=\eta(z)\}
\end{align}
where $i^{\bullet}$ denotes either $i$ or $i^+$ for $i\in [L_X]$, and similarly for $k^{\bullet}$, so that $g^{QX},g^{ZQ}$ act \textit{locally} in 5D.

\subsubsection{Defect Map}

\begin{figure}[ht]
\centering
\includegraphics[width=0.9\columnwidth]{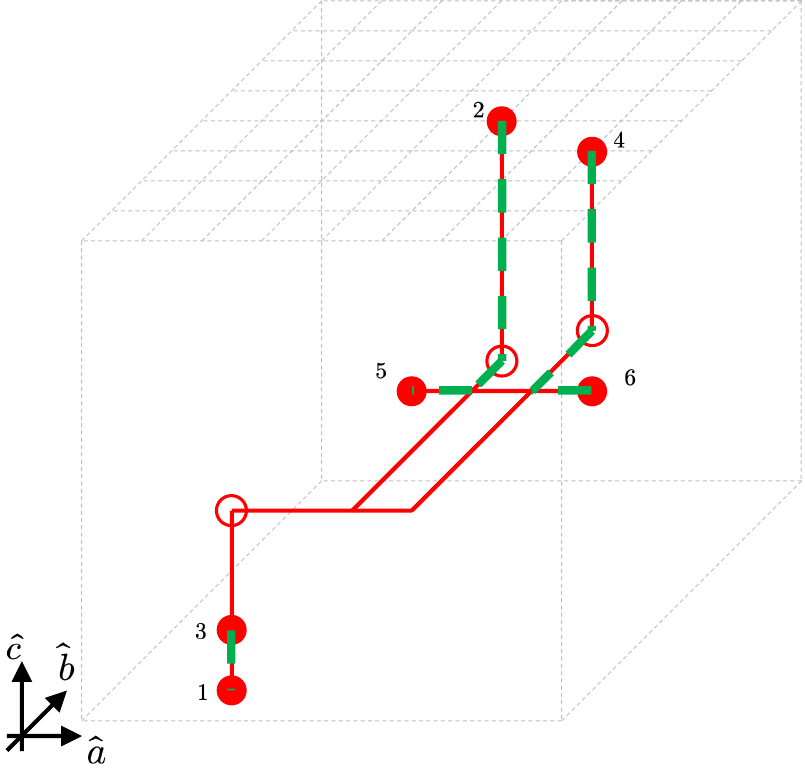}
\caption{Defects on $[L]^3$. The red dots denote qubits $q_1,...,q_6$ within the overlap $x\wedge z$ of some $x$ and $z$ check pair, while hollow red dots denote the intermediate qubits determined by the color route in Theorem \ref{thm:color-route-5D}. The red lines denote routes $\rho_1,\rho_2,\rho_3$ connecting $1\lr 2,3\lr 4$ and $5\lr 6$, respectively.
The green dashed lines denote the defects $\gamma_1,\gamma_2,\gamma_3$ obtained from Corollary \ref{cor:defect-congestion-5D}.
}
\label{fig:defect-pair-5D} 
\end{figure}

Similar to the 3D Layer Codes, the defect maps $p^{ZX}$ will correspond to strings between pairs of qubits $q_1,q_2\in x\wedge z$ for adjacent $x,z$ checks.
However, due to the 5D hypercube structure (3D grid structure of qubits), the strings are somewhat more complicated, even more so than the 4D case.
Fortunately, our construction of the check layers exactly accounts for this complication by utilizing the color routes in Theorem \ref{thm:color-route-5D}, which were constructed independently of $X$- and $Z$-checks.

Given $x\sim z$, we can order the qubits in the overlap $x\wedge z=\{q_1,q_2,...,q_{2M}\}$ arbitrarily. 
Let $p_m=q_{2m-1}q_{2m}$ and $\gamma_{\bullet}(p_m)$ denote the corresponding color route. 
Then denote $\rho_m(xz)$ the concatenation of the following routes: the AB route\footnote{First traverse along the 1-coordinate, then along the 2-coordinate} $\gamma_{00}\to \gamma_{10}$ along the fixed 3-coordinate, the line route from $\gamma_{10}\to \gamma_{01}$, and the AB route $\gamma_{01}\to \gamma_{11}$.
Then $\rho_m(xz)$ is contained in $\lambda(x),\lambda(z)$ for all $m$ and has boundaries $=q_{2m-1},q_{2m}$.
Unlike the 4D case in Lemma \ref{lem:defects-disjoint-4D}, the routes $\rho_m(xz),m=1,...,M$ are not necessarily vertex disjoint. In fact, they can overlap arbitrarily since we generally have no control over the intermediate points $\gamma_{01}(p),\gamma_{10}(p)$ of a packet $p$, i.e., the intermediate points do not depend solely on the endpoints of the packet $p$.
Fortunately, we can utilize the symmetric difference to construct simple paths $\gamma_{m}(xz),m=1,...,M$ which are pair-wise edge disjoint (overlap possibly at vertices) and that the boundaries of $\bigcup_{m}\partial\gamma_m(xz) = x\wedge z$.
Specifically,
\begin{lemma}
    \label{lem:disjoint-paths}
    Let $\rho_m,m=1,...,M$ denote a collection of paths in a finite graph $\sG$ with distinct boundaries and let $\rho =\rho_1+\cdots +\rho_m$ where $\rho_M$ are regarded as 1-chains in the associated graph complex $G$.
    Then there exists simple paths $\gamma_1,...,\gamma_M$ in $\sG$ which are edge-disjoint and that for $\gamma=\gamma_1+\cdots +\gamma_M$ then
    \begin{equation}
        \partial \gamma=\partial \rho
    \end{equation}
\end{lemma}
\begin{proof}
    Note that $\rho$ is a subgraph of $\sG$. We can then arbitrarily remove cycles without affecting the boundary $\partial\rho$ so that $\gamma$ denotes the forest which is a subgraph of $\rho$ and $\partial \gamma =\partial \rho$.
    Note that the edges of any forest with $2M$ boundary points can be partitioned into $M$ pairwise edge-disjoint simple paths $\gamma_m,m=1,...,M$ whose endpoints are exactly the the boundary points \cite{proofwiki:odd-vertices-edge-disjoint-trails}, and thus the statement follows.
\end{proof}

The following corollary is then immediate due to the definition of the line coloring maps $\eta(x),\eta(z)$ (see Theorem \ref{thm:line-coloring-5D}).
\begin{corollary}[Defect Congestion, Fig. \ref{fig:defect-pair-5D}]
    \label{cor:defect-congestion-5D}
    Let $\gamma_m(xz),m=1,...,M$ be that constructed by Lemma \ref{lem:disjoint-paths} so that $\gamma_m(xz),m=1,...,M$ are edge disjoint and contained in $\lambda(x),\lambda(z)$.
    Then the collection of \textbf{defect paths}
    \begin{equation}
        \Gamma_m(xz)=\eta(x)\times \gamma_m(xz)\times \eta(z)
    \end{equation}
    over all pairs of adjacent $x,z$ checks and $m$ must be edge decongested.
\end{corollary}

Parametrize the simple paths $\gamma_m(xz)$ by adding a $t$-dependence, i.e., $t\mapsto \gamma_m(t;xz)$, so that $t=0,T_m\mapsto$ the boundaries of $\gamma_m(t;xz)$, respectively.
The case is similar for $\Gamma_m(t;xz)$.
In both $\gamma_m,\Gamma_m$, we use half-integers $t^+$ to denote the edges of the route. 
We can then define the defect map as follows
\begin{align}
    \label{eq:pZX-2-5D}
    p^{ZX}|x;iqk\ket &= \sum_{z\sim x} \sum_{m=1}^{|x\wedge z|/2} \sum_{t<T_m} 1\{iqk= \Gamma_m(t;xz)\} \nonumber \\
    &\quad\quad\quad  \times|\Gamma_{m}(t^+;xz);z\ket \\
    \label{eq:pZX-1-5D}
    p^{ZX}|x;i e k\ket &= \sum_{z\sim x}\sum_{m=1}^{|x\wedge z|/2} \sum_{t<T_m} 1\{iek \in \Gamma_m(t^+;xz)\} \nonumber\\
    &\quad\quad\quad \times |\Gamma_{m}(t+1;xz);z\ket
\end{align}
where $e$ denotes an edge (nearest-neighbors) in $[L]^3$. 
Again note that $p^{ZX}$ acts \textit{locally} in 4D, in the sense that, they map cells in $C^X$ to cells in $C^{Z}$, respectively, with the neighboring coordinates, e.g., $\Gamma_m(t^+;xz)$ is an adjacent edge of $\Gamma_m(t;xz)$.
It's then straightforward to check that
\begin{lemma}[Compatibility]
    \label{lem:compatible-5D}
    The defect map $p^{ZX}$ in Eq. \eqref{eq:pZX-2-5D}-\eqref{eq:pZX-1-5D} are compatible with the gluing maps in Eq. \eqref{eq:gQX-5D}-\eqref{eq:gZQ-5D}, i.e.,
    \begin{equation}
        g^{ZQ} g^{QX} = p^{ZX} \partial^{X} + \partial^{Z} p^{ZX}
    \end{equation}
\end{lemma}
\begin{proof}
    Note that $\partial^Z p^{ZX}$ maps
    \begin{align}
        |x;iqk\ket &\mapsto \sum_{z\sim x} \sum_{m,t<T_m} 1\{iqk= \Gamma_m(t;xz)\} \\
        &\times (|\Gamma_{m}(t;xz);z\ket + |\Gamma_m(t+1;xz)\ket)
    \end{align}
    And that $p^{ZX}\partial^X$ maps
    \begin{align}
        |x;iqk\ket &\mapsto \sum_{z\sim x} \sum_{m,t<T_m} \sum_{e\sim q} 1\{iek= \Gamma_m(t^+;xz)\}\\
        &\quad\quad \times|\Gamma_m(t+1;xz);z\ket
    \end{align}
    Note that if $0<t<T_m$ such that $iqk=\Gamma_m(t;xz)$, then there are exactly two edges $e$ adjacent to $q$ in $[L]^2$ such that $iek=\Gamma_m(s^+;xz)$ for some $s$, i.e., $iek=\Gamma_m(t^-;xz)$ and $=\Gamma_m(t^+;xz)$. Conversely, if $t=0,T_m$, there exists exactly one edge $e$ adjacent to $q$ such that $iek=\Gamma_m(s^+;xz)$ for some $s$, i.e., $iek=\Gamma_m(t^+;xz)$ if $t=0$ and $=\Gamma_m(t^-;xz)$ if $t=T_m$.
    
    Hence, $p^{ZX}\partial^{X} + \partial^{Z} p^{ZX}$ maps
    \begin{align}
        |x;iqk\ket &\mapsto \sum_{z\sim x} \sum_{m} 1\{iqk\in \partial\Gamma_m(xz)\} |iqk;z\ket \\
        &= \sum_{z\sim x} 1\{q \in x\wedge z,i=\eta(x),k=\eta(z)\} \\
        &\quad\quad\times |iqk;z\ket
    \end{align}
    where $\partial \Gamma_m(xz)$ are the boundary points of the string segment in $[L_X]\times [L]^2\times [L_Z]$.
    One can similarly check that the right-hand-side is equal to $g^{ZQ}g^{QX}$ and thus the statement follows.
\end{proof}

\subsection{Main Result}
\label{sec:main-result-5D}
\begin{theorem}[Logical]
    \label{thm:logical-5D}
    Let $D=5$.
    Let $C^{X},C^{Q},C^{Z}$ be defined as in Eq. \eqref{eq:CX}-\eqref{eq:CZ}.
    Let gluing maps $g^{QX},g^{ZQ}$ be defined as in Eq. \eqref{eq:gQX-5D}-\eqref{eq:gZQ-5D}, and defect map $p^{ZX}$ in Eq. \eqref{eq:pZX-2-5D}-\eqref{eq:pZX-1-5D}.
    Then the output, referred as the \textbf{5D Layer Code}, $C$ is a well-defined complex embedded\footnote{See Example \ref{ex:layer-codes}} in 5D with at most 
    \begin{equation}
        \max(2, 2(D-3), D-1) = 4
    \end{equation}
    qubits at each edge and with column code $=A$, and
    \begin{equation}
        H_1(C) \cong H_1(A)
    \end{equation}
    And maximum weights  
    \begin{equation}
        C_2 \xrightleftharpoons[2(D-1)]{3(D-1)} C_1 \xrightleftharpoons[3(D-1)]{2(D-1)} C_0
    \end{equation}
    Specifically, the total qubit degree is $2D$.
\end{theorem}

\begin{theorem}[Distance]
    \label{thm:distance-5D}
    Let $C$ be the 5D Layer Code in Theorem \ref{thm:logical-5D} with 1-(co)systolic distance $d_X(C)=d_1(C),d_Z(C)=d^1(C)$, and similarly for the input code $A$. Then
    \begin{align}
        d_X(C) &=\Omega\left(\frac{L_Z}{\weight^4 \qubit^2} \right) d_X(A) \\ 
        d_Z(C) &=\Omega\left(\frac{L_X}{\weight^4 \qubit^2} \right) d_Z(A) 
    \end{align}
\end{theorem}

\begin{remark}[Better Bounds]
    \label{rem:better-bounds-5D}
    Similar to the 4D case in Remark \ref{rem:better-bounds-4D}, it's always possible to boost $L_Z=\Theta(\weight^3\qubit^2 L)$ and thus we are guaranteed that
    \begin{equation}
        d_X(C) = \Omega\left(\frac{1}{\weight}\right) Ld_X(A)
    \end{equation}
    The case is similar for the $Z$-type (1-cosystolic) distance.
\end{remark}

\begin{theorem}[Energy Barrier]
    \label{thm:energy-5D}
    Let $C$ be the 5D Layer code in Theorem \ref{thm:logical-5D} with $X,Z$-type (1-)energy barrier $\Delta_X(C),\Delta_Z(C)$, and similarly for the input code $A$. Then
    \begin{align}
        \Delta_X(C) =\Omega\left(\frac{1}{\weight^2 \qubit} \right) \Delta_X(A) \\
        \Delta_Z(C) =\Omega\left(\frac{1}{\weight^2 \qubit} \right) \Delta_Z(A)
    \end{align}
\end{theorem}

\section*{Acknowledgements}
ACY and NB are employed by Iceberg Quantum. 
ACY thanks Dominic J. Williamson, Aleksander Kubica, Shouzhen Gu for helpful discussions.

All definitions, theorems and proofs presented in this work were obtained by the authors. AI tools were used only for editorial purposes and for generating the numerical calculations reported in Appendix~\ref{sec:numerics}.

\appendix
\section{3D Layer Code Energy Barrier}
\label{sec:3D-layer-code}
In \cite{yuan2025unified}, the logical subspace and code distance of the 3D Layer Codes were derived using the language of homology, while the proof of the energy barrier was omitted.
However, since we are concerned with the energy barrier of the 4D and 5D embedding in the main text, we provide a detailed derivation in this section for comparison.

\subsection{Review}
Following \cite{yuan2025unified}, though with somewhat different notation, let us define the conventional 3D layer codes as follows. 
Let $A=X\to Q\to Z$ be the input CSS code with $X,Q,Z$ denoting the $X$-checks, qubits, and $Z$-checks.
Let $\sX,\sQ,\sZ$ denote the collection of corresponding basis elements $x,q,z$, all of which are arranged on a 1D linear array so that $\sX =[|\sX|]$ and simialrly, for $\sQ,\sZ$.
Let $A$ have weights
\begin{equation}
X \xrightleftharpoons[\qubit_X]{\weight_X} Q \xrightleftharpoons[\weight_Z]{\qubit_Z} Z
\end{equation}
Let $R_X,R_Q,R_Z$ be the repetition code on $|\sX|,|\sQ|,|\sZ|$ vertices, and define
\begin{align}
    \label{eq:CX-3D}
    C^X &= X\otimes R_Q^\top \otimes R_Z^\top \\
    C^Q &= R_X\otimes Q\otimes R_Z^\top \\
    \label{eq:CZ-3D}
    C^Z &= R_X\otimes R_Q\otimes Z
\end{align}
so that the basis elements of $C^X,C^Q,C^Z$ can be labeled via $|\bar{x},q,z\ket,|x,\bar{q},z\ket,|x,q,\bar{z}\ket$ (and so on), respectively.
Denote $\pi^X$ by the projections from $C\to C^{X}$ and $\pi^x$ to the $x$-check layer. Similarly, define $\pi^{Q},\pi^q,\pi^{Z},\pi^{z}$.
Write $\partial^x = \pi^{x}\partial^X$ and similarly for $\partial^{q},\partial^z$. 
Also write $g^{qX}=\pi^{q}g^{QX}$ and similarly for $p^{zX},g^{zQ}$.
Note that $\partial^q \ne \pi^{q} \partial$ in general; rather, $\pi^{q}\partial = \partial^q + g^{qX}$.

Define the gluing maps
\begin{align}
    \label{eq:gQX-3D}
    g^{QX}|\bar{x},q,z^{\bullet}\ket &= |x,\bar{q},z^{\bullet}\ket 1\{q\sim x\} \\
    \label{eq:gZQ-3D}
    g^{ZQ} |x^{\bullet},\bar{q},z\ket &= |x^{\bullet},q,\bar{z}\ket 1\{z\sim q\}
\end{align}
where $z^{\bullet}$ denotes either integer $z$ or half-integer $z^+$, and similarly, for $x^{\bullet}$.

The defect map is defined as follows. Given pair of adjacent $x,z$, order the common qubits $x\wedge z$ based on the ordering of the 1D array of qubits $\sQ$ and label via $q_1<q_2<...<q_{2M}$. 
Let $t\mapsto \gamma_m(t;xz): [0,T_m] \to \sQ$ denote the path from $q_{2m-1} \to q_{2m}$ along the 1D array of qubits and $T_m=q_{2m}-q_{2m-1}$.
Write half-integer $t^+$ to denote the edges of $\gamma_m$.
Then define the defect map as
\begin{align}
    \label{eq:pZX-2-3D}
    p^{ZX} |\bar{x}qz\ket &= 1\{z\sim x\} |xq^+\bar{z}\ket \\
    &\quad \times \sum_{m} \sum_{t<T_m} \ket1\{q =\gamma_m(t;xz)\}\\
    \label{eq:pZX-1-3D}
    p^{ZX} |\bar{x}q^+ z\ket &= 1\{z\sim x\} |x,q+1,\bar{z}\ket \\
    &\quad\times\sum_{m} \sum_{t<T_m} 1\{q^+ =\gamma_m(t^+;xz)\} 
\end{align}

\subsection{Result}
\begin{theorem}[Energy Barrier]
    \label{thm:energy-3D}
    Let $C$ denote the 3D layer code as defined in Eq. \eqref{eq:CX-3D}-\eqref{eq:pZX-1-3D} with $X$- and $Z$-type (1-)energy barrier $\Delta_X(C),\Delta_Z(C)$, and similarly for the input code $A$.
    Then
    \begin{align}
        \Delta_X(C) &= \Omega\left(\frac{1}{\weight_X \qubit_Z}\right) \Delta_X(A) \\
        \Delta_Z(C) &= \Omega\left(\frac{1}{\weight_Z \qubit_X}\right) \Delta_Z(A)
    \end{align}
\end{theorem}
\begin{proof}
    The statement follows from Theorem \ref{thm:path-reduction-3D} and Remark \ref{rem:step-size}.
\end{proof}
\subsection{Proof of Energy Barrier}

Since the construction is symmetric in $X,Z$, we shall be concerned with the $X$-type energy barrier $\Delta(C)$, while the $Z$-type is similarly derived.

\begin{lemma}[$X$ layer error]
    \label{lem:X-layer-err-3D}
    Let $\sigma^x\in x\otimes R_Q^\top \otimes R_Z^\top$ be a 0-chain (syndrome). Then there exists a map $\sigma^x\mapsto \hat{e}^x$ such that the 1-chain $\hat{e}^x$ is such that $\partial^x\hat{e}^x =\sigma^x$ and
    \begin{equation}
        \Delta_1(\hat{e}^x)\le (1+\weight_X)|\sigma^x|
    \end{equation}
    and
    \begin{align}
        |g^{qX} \hat{e}^x| &\le 1\{q\sim x\} |\sigma^x| \\
        |p^{zX} \hat{e}^x| &= 0
    \end{align}
    Moreover, if $\sigma^x=0$, then $\hat{e}^x=0$.
\end{lemma}
\begin{proof}
    Define $\hat{e}^x$ as the collection of strings with endpoints in $\partial^x e^x$, i.e., 
    \begin{equation}
        \hat{e}^x = \sum_{q^+ z^+} 1\{|\bar{x}q^+z^+\ket\in \sigma^x\}\left(\underbrace{\sum_{j\le q} |\bar{x}j z^+\ket }\right)
    \end{equation}
    Then it's clear that $\partial^x \hat{e}^x =\sigma^x$.
    Also note that the string (underbraced) can be generated by the following string-like Pauli path
    \begin{equation}
        \gamma^x(t)=t\mapsto \sum_{j\le t}|\bar{x} j z^+\ket
    \end{equation}
    It's then clear that 
    \begin{equation}
        \left|\partial^X \gamma^x(t) \right| \le 1,\quad  \forall t
    \end{equation}
    Also note that for any $t\le q$, we have
    \begin{align}
        |g^{qX} \gamma^x(t)| &\le \sum_{j\le t} 1\{j=q\sim x\} \\
        &\le 1\{q\sim x\} 
    \end{align}
    Similarly, by definition of the defect map, we see that $p^{zX} \gamma^x(t)=0$. Hence,
    \begin{align}
        |\partial \gamma^{x} (t)| &\le (1+\weight_X), \quad \forall t \\
        \Delta_1(\gamma^x) &\le (1+\weight_X) \\
        \Delta_1(\hat{e}^x) &\le (1+\weight_X) |\sigma^x|
    \end{align}
    where we utilized the sub-additivity of the energy barrier.
\end{proof}

We will also require analogous statement for the qubit and $Z$-check layers, as we collect below.
\begin{lemma}[$Q$ layer error]
    \label{lem:Q-layer-err-3D}
    Let $\sigma^q\in R_X\otimes q\otimes R_Z^\top$ be a 0-chain. Then there exists map $\sigma^{q}\mapsto \hat{e}^q$ such that the 1-chain $\hat{e}^q$ satisfies $\partial^q \hat{e}^q=\sigma^q$ and
    \begin{equation}
        \Delta_1(\hat{e}^{q}) \le (1+\qubit_Z)|\sigma^q|
    \end{equation}
    Moreover, if $\sigma^q=0$, then $\hat{e}^q=0$.
\end{lemma}
\begin{proof}
    Define $\hat{e}^q$ as the collection of strings with endpoints in $\partial^q e^q$, i.e.,
    \begin{equation}
        \hat{e}^q = \sum_{ x z^+} 1\{|x \bar{q} z^+\ket\in \sigma^q\}\left(\underbrace{\sum_{k\le z} |x\bar{q} k\ket }\right)
    \end{equation}
    Then it's clear that $\partial^q \hat{e}^q =\sigma^q$.
    Note that the string (underbraced) can be generated by the following string-like Pauli path
    \begin{equation}
        \gamma^{q}(t) = t\mapsto \sum_{k\le t}|x\bar{q} k\ket
    \end{equation}
    It's then clear that 
    \begin{equation}
        |\partial^Q \gamma^q(t)| \le 1, \quad \forall t
    \end{equation}
    Also note that
    \begin{align}
        |g^{zQ} \gamma^{q}(t)| &= \sum_{k\le z'} 1\{z=k\sim q\}  \\
        &\le 1\{z\sim q\} 
    \end{align}
    Hence,
    \begin{align}
        |\partial \gamma^{q}(t)| &\le (1+\qubit_Z), \quad \forall t\\
        \Delta_1(\gamma^{q}) &\le (1+\qubit_Z) \\
        \Delta_1(\hat{e}^q) &\le (1+\qubit_Z)|\sigma^{q}|
    \end{align}
    where we utilized the sub-additivity of the energy barrier.
\end{proof}

After cleaning syndromes off of $X$ and $Q$ layers, there only remains $Z$ layers. However, a given $C^Z(z)$ might contain multiple point-like syndromes. In fact, if there is an even number of them, they can brought together and cancel each other out at a low energy cost. After this step, each $Z$ layer contains either no syndrome, or a single point-like syndrome.

\begin{lemma}[$Z$ layer error]
    \label{lem:Z-layer-err-3D}
    Let $\sigma^{z}\in R_X\otimes R_Q\otimes z$ be a 0-chain. Then there exists map $\sigma^z \mapsto \hat{e}^z$ such that the 1-chain $\hat{e}^z$ saitsfies
    \begin{equation}
        \partial^z \hat{e}^z + \sigma^z = |00\bar{z}\ket 1\{[\sigma_z]\ne 0\}
    \end{equation}
    where $[\sigma^z]\in H_0(C^Z)$, and
    \begin{equation}
        \Delta_1(\hat{e}^{z}) \le 2
    \end{equation}
    Moreover, if $\sigma^z=0$, we can choose $\hat{e}^z=0$.
\end{lemma}
\begin{proof}
    Note that $\sigma^z \in \im \partial^z$ iff $|\sigma^z|$ is even.
    Order $|xq\bar{z}\ket\in \sigma^z$ arbitrarily, i.e., $|x_1q_1\bar{z}\ket,|x_2q_2\bar{z}\ket...$.
    Let $\gamma^z_m(t)$ be a simple path in $R_X\otimes R_Q \otimes z$ connecting pairs $|x_{2m-1} q_{2m}\bar{z}\ket \to |x_{2m}q_{2m}$. 
    If $|\sigma^z|=2M+1$ is odd, let $\gamma_{M+1}^z(t)$ denote the simple path from $|x_{2M+1}q_{2M+1} \bar{z}\ket \to |00\bar{z}\ket$.
    It's then clear that $|\gamma^z_m(t)|\le 1$ for all $t$.
    Consider the concatenated path $\Gamma^z=\gamma_1^z\to \gamma_2^z \to \cdots $, defined similarly to Eq. \eqref{eq:concatenated-path}.
    Then it's clear that $\Delta_1(\Gamma^z)\le 1$.
    Let $\hat{e}^z$ denote the destination of Pauli path $\Gamma^z$. 
    Then we see that the statement follows.
\end{proof}

\begin{theorem}
    \label{thm:path-reduction-3D}
    Let $\sL$ be a nontrivial logical of the 3D Layer code $C$ so that $[\sL]\in H_1(C)$. If $\phi:H_1(C) \to H_1(A)$ denotes the isomorphism in Theorem 1 of \cite{yuan2025unified}, then
    \begin{equation}
        \Delta_{\weight_X}(\phi [\sL]) =O(\weight_X \qubit_Z) \Delta_1(\sL)
    \end{equation}
\end{theorem}
\begin{proof}
    Let $t\mapsto \Gamma(t)$ be a 1-continuous path in $C$ with destination $\sL$ such that $\Delta(\sL)=\Delta(\Gamma)$ and let $e =\Gamma(t)$ for an arbitrary $t$.
    
    By Lemma \ref{lem:X-layer-err-3D}, we can obtain $\hat{e}^x$ based on the syndrome $\pi^x\partial e$ where we take $\hat{e}^x=0$ if $\pi^x \partial e=0$. 
    Let $\hat{e}^X$ be the direct sum of $\hat{e}^x$ over $x\in \sX$ so that
    \begin{align}
        \Delta(\hat{e}^X) &\le(1+\weight_X)|\pi^X \partial e| =O(\weight_X) |\partial e|\\
        |g^{QX} \hat{e}^X| &\le \weight_X |\pi^X\partial e|  \\
        |p^{ZX} \hat{e}^X| & =0
    \end{align}
    By Lemma \ref{lem:Q-layer-err-3D}, we can obtain $\hat{e}^q$ based on the syndrome 
    \begin{equation}
        \pi^q\partial (e+\hat{e}^X) = \pi^q\partial e +g^{qX}\hat{e}^X
    \end{equation}
    where we take $\hat{e}^q=0$ if the syndrome $=0$.
    Let $\hat{e}^Q$ be the direct sum of $\hat{e}^q$ over $q\in \sQ$ so that
    \begin{align}
        \Delta(\hat{e}^Q) &\le (1+\qubit_Z)\left(|\pi^{Q} \partial e| +\weight_X |\pi^X \partial e|\right) \\
        &= O(\weight_X \qubit_Z) |\partial e|
    \end{align}
    By Lemma \ref{lem:Z-layer-err-3D}, we can obtain $\hat{e}^z$ based on the syndrome
    \begin{equation}
        \pi^z \partial (e+\hat{e}^X +\hat{e}^Q) = \pi^z\partial e +g^{zQ} \hat{e}^Q
    \end{equation}
    where we take $\hat{e}^z=0$ if the syndrome $=0$.
    Note that we used the fact that $p^{zX}\hat{e}^X =0$ for all $z$.
    Let $\hat{e}^Z$ denote the direct sum of $\hat{e}^z$ over $z\in \sZ$ so that $\Delta(\hat{e}^Z) \le 2$ and that if $\hat{e} =\hat{e}^X +\hat{e}^Q+\hat{e}^Z$, then
    \begin{equation}
        \partial (e+\hat{e}) = \sum_{z} |00\bar{z}\ket 1\{[\pi^z\partial e +g^{zQ}\hat{e}^Q]\ne 0\}
    \end{equation}
    
    Note that $e^X+\hat{e}^X \in \ker \partial^X = \im \partial^X$ and thus $=\partial^X s^X$ for some 2-chain $s^X\in C^X$.
    Let $s^X$ be obtained in an arbitrary but predetermined manner via a map $e^X+\hat{e}^X\mapsto s^X$.
    Hence,
    \begin{equation}
        \pi^{Q} \partial(e+\hat{e}^X) = \partial^Q(e^Q+g^{QX} s^X)
    \end{equation}
    And thus $\ell^{Q} \equiv e^Q+\hat{e}^Q +g^{QX} s^X \in \ker \partial^Q$ so that $\ell^{A} = [\ell^Q] \in H_1(C^Q)=Q$ is well-defined as a 1-chain in $A$.
    Finally, note that
    \begin{align}
        \pi^z \partial e+g^{zQ} \hat{e}^Q &= g^{zQ} \ell^Q + \partial^z (e^z +p^{zX} s^X) \\
        [\pi^z\partial e+ g^{zQ} \hat{e}^Q] &= [g^{zQ}] [\ell^{Q}] \\
        [\pi^z\partial e+ g^{zQ} \hat{e}^Q] &= \bra z|\partial^A \ell^A\ket
    \end{align}
    Hence, we see that
    \begin{equation}
        |\partial (e+\hat{e})| = |\partial^A \ell^A|
    \end{equation}
    Note that the previous construction defines a map $e\mapsto \ell^A$ and applies to any $e=\Gamma(t)$ and thus we obtain a path $t\mapsto \ell^A(t)$ with respect to $A$. 
    We further claim that the induced path $t\mapsto \ell^A(t)$ is $\weight_X$-continuous, which we postpone to the end of the proof.
    Note that if $e=\sL$, i.e., at the destination of the Pauli path $\Gamma(t)$, then $\hat{e}=0$ and that $\ell^A \in \phi[\sL]$. 
    Hence, the destination of Pauli path $\ell^A(t)$ is a nontrivial logical in the coset $\phi[\sL]\in H_1(A)$.
    Note that for any $t$, we have shown that
    \begin{equation}
        |\partial^A\ell^A(t)| = O(\weight_X \qubit_Z)|\partial \Gamma(t)|
    \end{equation}
    Hence,
    \begin{equation}
        \Delta_{\weight_X}(\phi[\sL]) =O(\weight_X \qubit_Z) \Delta_1(\sL)
    \end{equation}

    Let us now prove that $t\mapsto \ell^A(t)$ is $\weight_X$-continuous.
    Indeed, since $\Gamma$ is 1-continuous, we see that $e'=\Gamma(t+1)$ must differ from $e=\Gamma(t)$ by a single 1-cell in $C$, and thus there are three possibilities: the difference occurs in $e^x$ for some $x\in \sX$, or $e^q$ for some $q\in \sQ$, or $e^z$ for some $z\in \sZ$. 
    In the first case, we have $e'+e=f^x$ for some $x\in \sX$ and thus $|g^{qX}(s^X +s'^X)|\ne 0$ only if $q\sim x$ and thus changes at most $\weight_X$ qubit layers. Hence, $\ell'^q\ne \ell^q$ only if $q\sim x$ and thus $|\ell'^{A} -\ell^A| \le \weight_X$.
    In the second case, we have $e'+e=f^q$ for some $q\in \sQ$ and thus only $\ell'^q,\ell^q$ may differ over all qubits. Hence, $|\ell'^{A} -\ell^{A}| \le 1$.
    In the third case, it's clear that $\ell^A$ is independent of $e^z$ and thus $\ell'^A =\ell^A$.
    Hence, $t\mapsto \ell^{A}(t)$ is $\weight_X$-continuous.
    


\end{proof}
\section{Proof of 4D Embedding}
\label{sec:4Dproof}
In this section, we provided proofs for the main results collected in Theorem \ref{thm:logical-4D}, \ref{thm:distance-4D} and \ref{thm:energy-4D}.
The proofs follow the similar reasoning as \cite{williamson2023layer,layerreduction,yuan2025unified}, but included for completeness.

\subsection{Logicals Preserved}
\begin{proof}[Proof of Theorem \ref{thm:logical-4D}]
    By Lemma \ref{lem:compatible}, we see that the gluing maps and the defect map are compatible and thus by the framework in \cite{yuan2025unified}, the output code $C$ is well-defined.

    Let us now check that the embedded column code is $=A$.
    Since $C(x)$ is a cell-complex with a (unique) connected component $\sC_0(x)$, we see that $H^0(C(x))\cong \dF_2$ has basis $[\|x\ket]$ where we regard $\|x\ket= \sC_0(x)$ as the sum over all vertices in $C(x)$. 
    In particular, $H_2(C^X) \cong X$ with basis $[\sC_0(x)],x\in \sX$.
    Similarly, $H_0(C(z))\cong \dF_2$ has basis $[\|z\ket]$ where $\|z\ket = |v;z\ket$ and $v$ is any vertex in $C(z)$, and thus $H_0(C^Z)\cong Z$ with basis $[\|z \ket],z\in \sZ$. 
    Since $C(q)$ is the surface code, we see that it has logical $H_1(C(q))\cong \dF_2$ given by basis $[\|q\ket]$ where
    \begin{equation}
        \|q\ket = \sum_{k=1}^{L_Z} |iqk\ket
    \end{equation}
    for any arbitrary $i\in [\chi_X]$ index.
    Further note that
    \begin{align}
        [g^{QX}][\|x\ket] &= \sum_{q\in \AB(x),k} [g^{QX}|x;\eta(x)qk\ket] \\
        &= \sum_{q\sim x} \sum_{k} [|\eta(x),q,k\ket] \\
        &= \sum_{q\sim x} [\|q\ket]
    \end{align}
    where vertices in $(\eta\omega)(x)$ (but not in $\eta\AB(x)\otimes R_Z$) are omitted, since $g^{QX}$ acts trivially on them.
    Also note that
    \begin{align}
        [g^{ZQ}][\|q\ket] &= \sum_{k} [g^{ZQ} |i,q,k\ket] \\
        &= \sum_{z\sim q} [|i,q,\eta(z);z\ket] \\
        &=\sum_{z\sim q} [\|z\ket]
    \end{align}
    Hence, the embedded column code $=A$. In fact, since $C^{X},C^{Z}$ has no internal logicals, we see that $H_1(C)\cong H_1(A)$. Moreover, by Lemma \ref{lem:contracting-cycles}, $H_2(C^X),H_1(C^Q),H_0(C^Z)$ are the only nontrivial homology for $C^X,C^Q,C^Z$, respectively.
    Hence, by \cite{yuan2025unified},
    \begin{equation}
        H_s(C)\cong H_s(A), \quad s=2,1,0
    \end{equation}

    Note that the number of qubits of $C$ per edge follows from Remark \ref{rem:layer-congestion-4D} and the fact that an edge is either along the $\hat{x},\hat{z}$ axis or within the qubit grid.

    The weights of the 4D Layer Code $C$ is determined via the following diagram
    \begin{equation}
    \label{eq:weight-diagram-4D}
    \begin{tikzpicture}[baseline]
    \matrix(a)[matrix of math nodes, nodes in empty cells, nodes={minimum size=25pt},
    row sep=2em, column sep=2em,
    text height=1.25ex, text depth=0.25ex]
    {&& \red{C^{X}_{2}}  & \red{C^{X}_{1}} & \red{C^{X}_{0}}\\
    & \gray{C^{Q}_{2}}  & \gray{C^{Q}_{1}}  & \gray{C^{Q}_{0}} &\\
    \blue{C^{Z}_{2}} & \blue{C^{Z}_{1}} & \blue{C^{Z}_{0}} &&\\};
    \path[-left to,font=\scriptsize,transform canvas={yshift=0.2ex}]
    (a-1-3) edge node[above]{$6$}  (a-1-4)
    (a-1-4) edge node[above]{$4$}  (a-1-5)
    (a-2-2) edge node[above]{$4$}  (a-2-3)
    (a-2-3) edge node[above]{$2$}  (a-2-4)
    (a-3-1) edge node[above]{$4$}  (a-3-2)
    (a-3-2) edge node[above]{$2$}  (a-3-3);
    \path[left to-,font=\scriptsize,transform canvas={yshift=-0.2ex}]
    (a-1-3) edge node[below]{$2$}  (a-1-4)
    (a-1-4) edge node[below]{$4$}  (a-1-5)
    (a-2-2) edge node[below]{$2$}  (a-2-3)
    (a-2-3) edge node[below]{$4$}  (a-2-4)
    (a-3-1) edge node[below]{$4$}  (a-3-2)
    (a-3-2) edge node[below]{$6$}  (a-3-3);
    \path[-left to,font=\scriptsize,transform canvas={xshift=0.2ex}]
    (a-1-3) edge node[right]{$1$}  (a-2-3)
    (a-1-4) edge node[right]{$1$}  (a-2-4)
    (a-2-2) edge node[right]{$1$}  (a-3-2)
    (a-2-3) edge node[right]{$1$}  (a-3-3);
    \path[left to-,font=\scriptsize,transform canvas={xshift=-0.2ex}]
    (a-1-3) edge node[left]{$1$}  (a-2-3)
    (a-1-4) edge node[left]{$1$}  (a-2-4)
    (a-2-2) edge node[left]{$1$}  (a-3-2)
    (a-2-3) edge node[left]{$1$}  (a-3-3);
    \path[-left to,font=\scriptsize]
    (a-1-3) edge[bend right=80, dashed] node[right]{$2$} (a-3-2)
    (a-1-4) edge[bend left=80, dashed] node[right]{$1$} (a-3-3);
    \path[left to-,font=\scriptsize]
    (a-1-3) edge[bend right=80, dashed] node[left]{$1$} (a-3-2)
    (a-1-4) edge[bend left=80, dashed] node[left]{$2$} (a-3-3);
    \end{tikzpicture}
    \end{equation}
    Specifically, the weights of the row $C^{Q}$ are straightforwardly obtained since each qubit layer is a surface code.
    The weights of $C^X,C^Z$ can also be obtained straightforwardly since the cells of each $x,z$-check layer correspond to faces, edges and vertices in a 3D sub-grid of the 4D hypercube. For example, as shown in Fig. \ref{fig:check-layer-4D}, $\eta\AB(x) \otimes R_Z +(\eta \omega)(x)$ has fixed $\hat{x}$ coordinate $\eta(x)$ and thus is a \textit{layer-like} (or, \textit{cylinder-like}) structure in the 3D grid with axes $\hat{a}\times \hat{b} \times \hat{z}$.
    Note that in this case, a vertex can have at most 6 adjacent edges and an edge can have at most 4 adjacent faces.
    
    Finally, let us focus on the weights of $g^{QX},g^{ZQ}$ and $p^{ZX}$.
    By Eq. \eqref{eq:gQX-4D}, it's clear that $g^{QX}$ has max column weight $=1$. 
    By Eq. \eqref{eq:gZQ-4D}, we see that $g^{ZQ}$ has max column weight $=1$ since
    \begin{equation}
        |g^{ZQ}|i^{\bullet},q,k\ket| \le \sum_{z\sim q} 1\{\eta(z)=k\} \le 1
    \end{equation}
    where we used the fact that if $q\sim z,z'$, then $\eta(z)\ne \eta(z')$, which follows from Theorem \ref{thm:line-coloring-4D}.
    Since the construction is symmetric, it's straightforward to check that $g^{QX},g^{ZQ}$ have max row weights $=1$.
    Now consider the defect maps as defined in Eq. \eqref{eq:pZX-2-4D}-\eqref{eq:pZX-1-4D}.
    Note that
    \begin{align}
        |p^{ZX} |x;iqk\ket|&\le 1\{i=\eta(x)\}\sum_{z\sim x} 1\{k=\eta(z)\} \\
        &\quad \times \sum_{m} 1\{q\in \AB(q_{2m-1},q_{2m})\} \\
        &\le 1\{i=\eta(x)\}\sum_{z\sim x} 1\{k=\eta(z)\} \\
        &\quad\quad \times 1\{q\in \AB(q_{2m-1},q_{2m}),\exists m\}
    \end{align}
    where we used Lemma \ref{lem:defects-disjoint-4D}.
    Further note that if $z,z'\sim x$ are such that $k=\eta(z)=\eta(z')$, then by definition, $\AB(z),\AB(z')$ cannot share the same column or row.
    However, since $q\in \Gamma_m(xz)$, we see that it must either be in a row of $\AB(z)$ or a column of $\AB(z)$. In either case, say a row of $\AB(z)$, there is at most one more $z'\sim x$ such that $q\in \Gamma_m(xz')$ and is in a column of $\AB(z')$. See, for example, Fig. \ref{fig:defect-congestion} for a qubit $q$ which lies in the intersection between two green lines. Hence,
    \begin{equation}
        |p^{ZX} |x;iqk\ket|\le 2
    \end{equation}
    Similarly, note that
    \begin{align}
        |p^{ZX} |x;iek\ket | &\le 1\{i=\eta(x)\} \sum_{z\sim x}  1\{k=\eta(z)\} \\
        &\quad \times \sum_{m} 1\{e \in \AB(q_{2m-1},q_{2m})\} \\
        &\le 1\{i=\eta(x)\} \sum_{z\sim x}  1\{k=\eta(z)\} \\
        &\quad \times 1\{e\in \AB(q_{2m-1},q_{2m}), \exists m\} \\
        &\le 1
    \end{align}
    where we used Lemma \ref{lem:defects-disjoint-4D}, and the fact that the collection of $z$-check layers has edges decongested (or equivalently, $\AB(z),\AB(z')$ do not share a row or column for $\eta(z)=\eta(z')=k$). The case is similar for the max row weight, and thus the check weight follows.
\end{proof}

\subsection{Distance Preserved}

To prove that the distance of the 4D Layer Codes is preserved in some manner, the procedure is similar to that of the $\infty$-dimensional Layer Codes \cite{layerreduction} (which also applies to the conventional 3D Layer Codes \cite{williamson2023layer,yuan2025unified}). 
Specifically, we will utilize Appendix (A) of \cite{layerreduction} to compute the relative expansion coefficient of the check layers $C^{X}$ ($C^{Z}$) with respect to the gluing map $g^{QX}$ ($g^{ZQ}$), and then apply the Cleaning Lemma in \cite{yuan2025unified} to prove a lower bound on the 4D Layer Codes.
In particular, since the 4D Layer Codes are symmetric in $X,Z$-checks, we shall restrict our attention to the $X$-type distance.

\begin{lemma}[AB (co)expansion]
    For any $x$-check, $\AB(x)$ is $2/|\AB(x)|$-coexpanding, and $1/|\AB(x)|$-coexpanding relative to the projection $\pi$ onto $w\le \weight_X$ qubits $q\sim x$.
\end{lemma}
\begin{proof}
    The proof is exactly the same as Example A.2 of \cite{layerreduction} and thus omitted.
\end{proof}

\begin{lemma}[AB size]
    \label{lem:AB-size}
    Fix an $x$-check and let there be $w\le \weight_X$ many qubits $q\sim x$. Then
    \begin{equation}
        |\AB(x)| \ge 2L\sqrt{w} -w^2
    \end{equation}
\end{lemma}
\begin{proof}
    Note that $\AB(x)$ must consist of $r$ rows and $c$ columns and that $q\sim x$ must be positioned at the intersection of the rows and columns (see, e.g., Fig. \ref{fig:check-routing}). Hence, $w\le rc$.
    Conversely, we also must have $r,c\le w$ since if there exists a row that does not contain any $q\sim x$, then it cannot be included in $\AB(x)$.
    Note that $|\AB(x)|= (r+c)L-rc$ and thus
    \begin{equation}
        |\AB(x)| \ge 2L\sqrt{w} -w^2
    \end{equation}
    where we used the fact that $2\sqrt{rc} \le r+c$.
\end{proof}

\begin{lemma}[Relative Expansion]
    \label{lem:relative-expansion-4D}
    Fix any $x$ check. Let $\tilde{\delta}^{x}$ be the codifferential of $\eta\AB(x)\otimes R_Z$. Then for every subset of vertices $s\in \eta\AB(x)\otimes R_Z$, there exists $\hat{s}$ with $\delta^{x} s=\delta^{x} \hat{s}$ such that
    \begin{equation}
        |\tilde{\delta}^x s| = \Omega\left(\frac{1}{\weight^2\qubit}\right) |g^{QX} \hat{s}|
    \end{equation}
    
    In particular, if $\delta^{x}$ is the codifferential of $C(x)$, then for every subset of vertices $s\in C(x)$, there exists $\hat{s}$ with $\tilde{\delta}^{x}s = \tilde{\delta}^{x}\hat{s}$ such that
    \begin{equation}
        |\delta^x s| = \Omega\left(\frac{1}{\weight^2\qubit}\right) |g^{QX} \hat{s}|
    \end{equation}
\end{lemma}
\begin{proof}
    Note that $R_Z$ is $2/L_Z$-coexpanding. Hence, by Lemma (A.2) of \cite{layerreduction}, we see that $\eta\AB(x)\otimes R_Z$ is $c^{QX}$-coexpanding relative to $g^{QX}$ where
    \begin{equation}
        c^{QX} =\min\left(\frac{1}{|\AB(x)|},\frac{2}{L_Z}\right)
    \end{equation}
    Hence, by Remark (5) of \cite{layerreduction},
    \begin{align}
        |\tilde{\delta}^x s| &\ge c^{\pi} \frac{|\AB(x)|L_Z}{2\weight_X L_Z} |g^{QX} \hat{s}| \\
        &\ge \frac{1}{2\weight_X} \min\left(1,\frac{2|\AB(x)|}{L_Z}\right) |g^{QX} \hat{s}| \\
        &\ge \Omega\left(\frac{1}{\weight_X\weight_Z \qubit_Z}\right)|g^{QX}\hat{s}|
    \end{align}
    where we used the fact that $|\AB(x)|\ge L$. 

    Now consider the case where $s\in C(x)$. Let $\tau:C(x) \to \eta\AB(x) \otimes R_Z$ denote the natural projection from $C(x)= \eta \AB(x) \otimes R_Z +\eta \omega(x)$ to $\eta \AB(x) \otimes R_Z$. Then
    \begin{align}
        |\delta^x s| &\ge |\tau \delta^x s|\\
        &= |\tilde{\delta}^x \tau s| \\
        &=\Omega\left(\frac{1}{\weight_X\weight_Z\qubit_Z}\right) \\
        &\quad\quad\times \min(|g^{QX} \tau s|, |g^{QX} (\tau s+ \sV)|)
    \end{align}
    where $\sV(x)$ denotes the collection of all vertices in $\eta\AB(x)\otimes R_Z$. Note that $\tau\sC_0(x) = \sV(x)$ where $\sC_0(x)$ is the collection of all vertices in $C(x)$ and thus
    \begin{align}
        |\delta^x s| &=\Omega\left(\frac{1}{\weight_X\weight_Z\qubit_Z}\right) \\
        &\quad\quad\times \min(|g^{QX}  s|, |g^{QX} (s+ \sC_0)|)
    \end{align}
    Where we used the fact that $g^{QX}=g^{QX}\tau$.
\end{proof}

\begin{proof}[Proof of Theorem \ref{thm:distance-4D}]
    By Lemma \ref{lem:relative-expansion-4D} and the Cleaning Lemma in \cite{yuan2025unified}, the 1-systolic distance must satisfy
    \begin{equation}
        d_1(C) =\Omega\left(\frac{L_Z}{\weight_X\weight_Z\qubit_Z}\right) d_1(A)
    \end{equation}
    where the extra factor $L_Z$ utilizes the fact that each qubit layer $R_X\otimes R_Z^\top$ has 1-systolic distance $\ge L_Z$. 
    Since the construction is symmetric in $X\lr Z$, the statement also holds for the 1-cosystolic distance.
\end{proof}

\subsection{Energy Barrier Preserved}

To prove that the energy barrier of the 4D Layer Codes is preserved in some manner, the procedure is similar to that of the conventional 3D Layer codes \cite{williamson2023layer} (see Appendix \ref{sec:3D-layer-code} for a detailed exposition), though additional care must be taken due to the contracting layers, e.g., $\eta\omega (x)$ in Fig. \ref{fig:check-layer-congestion-4D}.
In particular, since the algorithm Layer Codes are symmetric in $X,Z$-checks, we shall restrict our attention to the $X$-type energy barrier.

\begin{definition}[Notation]
    Let $\pi^X,\pi^Q,\pi^Z$ denote the projection chain maps $C\to C^X,C^Q,C^Z$, respectively.
    Let $\pi^x,\pi^q,\pi^z$ denote the projection chain maps onto $x$-check, qubit $q$, $z$-check layer, respectively.
    Let $\partial^x=\pi^x\partial^X$ and similarly define $\partial^q,\partial^z$.
    Let $g^{qX}=\pi^q g^{QX}$ and similarly define $g^{zQ}$ and $p^{zX}$.
\end{definition}

\begin{lemma}[$X$ layer error]
    \label{lem:X-layer-err-4D}
    Let $\sigma^x\in C(x)^\top$ be a 0-chain (syndrome). Then there exists a map $\sigma^x\mapsto \hat{e}^x$ such that the 1-chain $\hat{e}^x$ is such that $\partial^x\hat{e}^x =\sigma^x$ and
    \begin{equation}
        \Delta_1(\hat{e}^x)\le (1+\weight_X \qubit_Z \min(\weight_X,\weight_Z))|\sigma^x|
    \end{equation}
    and
    \begin{align}
        |g^{qX} \hat{e}^x| &\le 1\{q\sim x\} |\sigma^x| \\
        |p^{zX} \hat{e}^x| &\le 1\{z\sim x\} \min(\weight_X,\weight_Z) |\sigma^x|
    \end{align}
    Moreover, if $\sigma^x=0$, then $\hat{e}^x=0$.
\end{lemma}
\begin{proof}
    For every 0-cell $v^x\in \sigma^x$ (a face in $C(x)$), we see that either $v^x\in \eta\AB(x)\otimes R_Z$ or $\in (\eta\omega)(x)$.
    
    If $v^x\in \eta\AB(x)\otimes R_Z$ as depicted in Fig. \ref{fig:syndrome-cleaning-vertical-4D}, we see that $v^x= |x;\eta(x),e,k^+\ket$ for some $k\in [L_Z]$ and edge $e\in \AB(x)$.
    Note that $e$ is either horizontal along that $\hat{a}$ axis or vertical along the $\hat{b}$ axis. In either case, let $(a,b)$ denote the smaller endpoint of the edge $e$ (relative to the lexicographic order).
    If $e$ is horizontal, define the following string
    \begin{equation}
        [0,a]\ni t \mapsto \gamma_t(v^x) = \sum_{a'\le t} |x;\eta(x),a'b,k^+\ket
    \end{equation}
    If $e$ is vertical, define the following string as
    \begin{equation}
        [0,b]\ni t \mapsto \gamma_t(v^x) = \sum_{b'\le t} |x;\eta(x),a b',k^+\ket
    \end{equation}
    By definition of $\AB(x)$, we see that $t\mapsto \gamma_t (v^x)$ is a 1-continuous Pauli path in $\eta \AB(x)\otimes R_Z$ and that
    \begin{align}
        |\partial^x \gamma_t (v^x)| &\le 1, \quad \forall t\\
        \partial^x \gamma_T (v^x) &= \sigma^x
    \end{align}
    where $T=a,b$ if $e$ is horizontal, vertical, respectively. 
    Also note that
    \begin{align}
        |g^{qX} \gamma_t(v^x)| &\le \sum_{a' \le t} 1\{q=a'b\sim x\} \\
        &\le 1\{q\sim x\}
    \end{align}
    And by definition of the defect maps in Eq. \eqref{eq:pZX-1-4D}, $p^{zX}\gamma_t (v^x)=0$ for all $z\in \sZ$ and $t$; since $k^+$ is a half integer, while $p^{zX}$ only acts on whole integer coordinates.
    Hence, we see that
    \begin{equation}
        \Delta_1(\gamma_T(v^x))\le \Delta(\gamma(v^x))\le 1+\weight_X
    \end{equation}

    \begin{figure}[ht]
    \centering
    \subfloat[$\eta(x)=\eta \in \hat{x}$\label{fig:syndrome-cleaning-vertical-4D}]{%
        \centering
        \includegraphics[width=0.9\columnwidth]{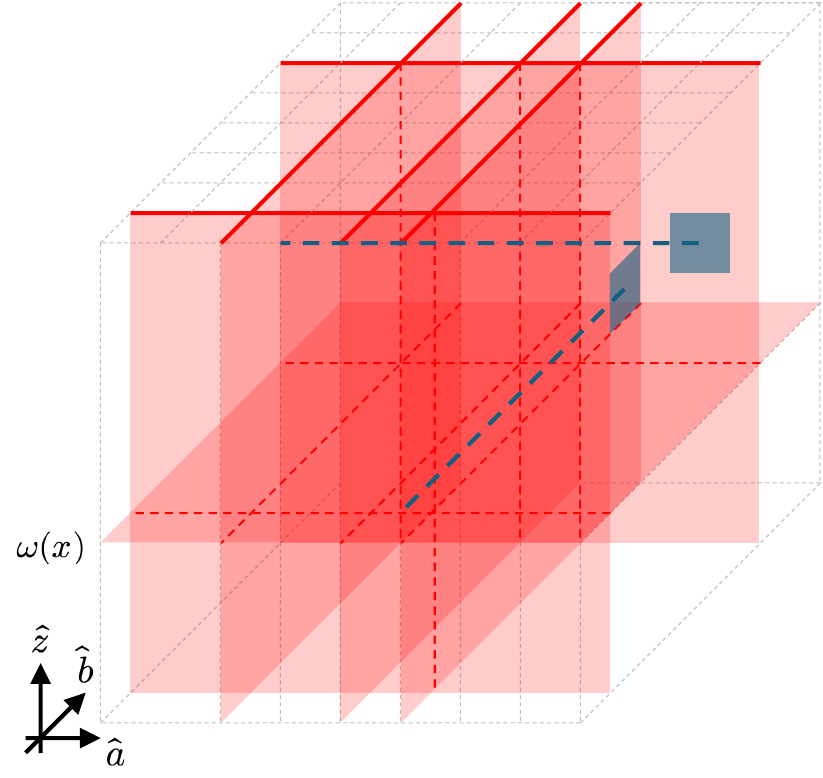}
    }
    \\
    \subfloat[$\eta\times \omega \in \hat{x} \times \hat{z}$\label{fig:syndrome-cleaning-horizontal-4D}]{%
        \centering
        \includegraphics[width=0.7\columnwidth]{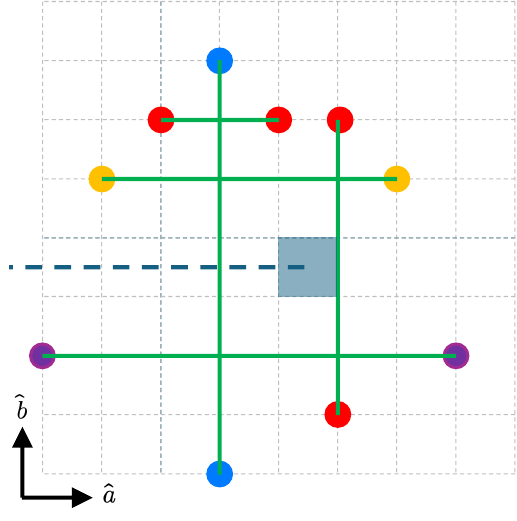}
    }
    \caption{Syndrome Cleaning. (a) is adopted from Fig. \ref{fig:check-layer-example-4D}. When cleaning a syndrome (shaded face) via the dashed lines, additional syndromes may occur on the qubit layers due to intersecting the qubit lines (vertical dashed). (b) is adopted from Fig. \ref{fig:defect-congestion}, with AB routes omitted for clarity. For a fixed $\eta\times \omega \in \hat{x}\times \hat{z}$, it's possible that the syndrome (shaded face) is enclosed by defect lines (green) so that, in particular, $\omega=\omega(x)=\eta(z)$ for some adjacent pairs $z\sim x$. This occurs due to the fact that $\omega(x)$ and $\eta(z)$ are independent.
    In this case, cleaning the syndrome via the dashed line induces additional excitation of the correspond $z$-check layers.
    }
    \label{fig:syndrome-cleaning}
    \end{figure}
    If $v^x \in (\eta \omega)(x)$, then $v^x =|x;\eta(x),f,\omega(x)\ket$ for some face $f\in [L]^2$ in $\hat{a}\times \hat{b}$ plane, as shown by the shaded region in Fig. \ref{fig:syndrome-cleaning-horizontal-4D}.
    Let $(a,b)$ denote the smallest vertex adjacent to $f$ in $[L]^2$ (relative to the lexicographic order).
    Define the horizontal string (dashed line in Fig. \ref{fig:syndrome-cleaning})
    \begin{equation}
        [0,a]\ni t \mapsto \gamma_t (v^x) = \sum_{a'\le t} |x;\eta(x),a'b^+,\omega(x)\ket
    \end{equation}
    Then $t\mapsto \gamma_t (v^x)$ is a 1-continuous Pauli path in $(\eta\omega)(x)$ and that
    \begin{align}
        |\partial^x \gamma_t (v^x)|&\le 1, \quad \forall t \\
        \partial^x \gamma_T(v^x) &=\sigma^x
    \end{align}
    Also note that $g^{qX} \gamma_t(v^x)=0$ for all qubits $q$ and $t$, and that
    \begin{equation}
        |p^{zX} \gamma_t(v^x)| \le  1\{z\sim x\} \min(\weight_X,\weight_Z)
    \end{equation}
    Hence, we see that
    \begin{equation}
        \Delta_1(\gamma_T(v^x)) \le \Delta(\gamma(v^x)) \le 1+ \weight_X \qubit_Z \min(\weight_X,\weight_Z)
    \end{equation}
    
    Define
    \begin{equation}
        \hat{e}^x =\sum_{v^x} 1\{v^x\in \sigma^x\} \gamma_T(v^x)
    \end{equation}
    Then by sub-additivity, we have
    \begin{equation}
        \Delta_1(\hat{e}^x) \le (1+\weight_X \qubit_Z \min(\weight_X,\weight_Z)) |\sigma^x|
    \end{equation}
    Also note that
    \begin{align}
        |g^{qX} \hat{e}^x| &\le 1\{q\sim x\}|\sigma^x| \\
        |p^{zX} \hat{e}^x| &\le 1\{z\sim x\}\min(\weight_X,\weight_Z)|\sigma^x|
    \end{align}
\end{proof}
\begin{lemma}[$Q$ layer error]
    \label{lem:Q-layer-err-4D}
    Let $\sigma^q\in C(q)$ be a 0-chain. Then there exists map $\sigma^{q}\mapsto \hat{e}^q$ such that the 1-chain $\hat{e}^q$ satisfies $\partial^q \hat{e}^q=\sigma^q$ and
    \begin{equation}
        \Delta_1(\hat{e}^{q}) \le (1+\qubit_Z)|\sigma^q|
    \end{equation}
    Moreover, if $\sigma^q=0$, then $\hat{e}^q=0$.
\end{lemma}
\begin{proof}
    Note that every 0-cell $v^q\in \sigma^q$ is of the form $|iqk^+\ket$. Define the string as
    \begin{equation}
        [0,T]\ni t \mapsto \gamma_t(v^q) = \sum_{k'\le t} |iqk'\ket
    \end{equation}
    where $T=k$ so that $\gamma_t(v^q)$ is a 1-continuous Pauli path in $C(q)$.
    Then it's clear that 
    \begin{align}
        |\partial^q \gamma_t(v^q)| &\le 1, \quad \forall t\\
        \partial^q \gamma_T(v^q) &= |iqk^+\ket
    \end{align}
    Also note $|g^{zQ} \gamma_t(v^q)| \le 1\{z\sim q\}$ and thus
    \begin{equation}
        \Delta_1(\gamma_T(v^q)) \le \Delta(\gamma(v^q)) \le 1+\qubit_Z
    \end{equation}
    Define
    \begin{equation}
        \hat{e}^q = \sum_{v^q} 1\{v^q\in \sigma^q\} \gamma_T(v^q)
    \end{equation}
    Then $\partial^q \hat{e}^q= \sigma^q$ and by sub-additivity, the statement follows.
\end{proof}
\begin{lemma}[$Z$ layer error]
    \label{lem:Z-layer-err-4D}
    Let $\sigma^{z}\in C(z)$ be a 0-chain. Then there exists map $\sigma^z \mapsto \hat{e}^z$ such that the 1-chain $\hat{e}^z$ satisfies 
    \begin{equation}
        \partial^z \hat{e}^z + \sigma^z = |00\bar{z}\ket 1\{[\sigma_z]\ne 0\}
    \end{equation}
    where $[\sigma^z]\in H_0(C^Z)$, and  $\Delta_1(\hat{e}^{z}) \le 2$.
    Moreover, if $\sigma^z=0$, we can choose $\hat{e}^z=0$.
\end{lemma}
\begin{proof}
    The proof is similar to Lemma \ref{lem:Z-layer-err-3D} and thus omitted.
\end{proof}
\begin{theorem}[Path Reduction]
    \label{thm:path-reduction-4D}
    Let $\sL$ be a nontrivial logical of the 4D Layer Code $C$ so that $[\sL]\in H_1(C)$. If $\phi:H_1(C) \to H_1(A)$ denotes the isomorphism in Theorem 1 of \cite{yuan2025unified}, then
    \begin{equation}
        \Delta_{\weight_X}(\phi [\sL]) =O(\weight_X \qubit_Z \min(\weight_X,\weight_Z)) \Delta_1(\sL)
    \end{equation}
\end{theorem}

\begin{proof}
    Let $t\mapsto \Gamma(t)$ be a 1-continuous path in $C$ with destination $\sL$ such that $\Delta(\sL)=\Delta(\Gamma)$ and let $e =\Gamma(t)$ for an arbitrary $t$, and note that the output code $C$ has differential
    \begin{equation}
        \partial = 
        \begin{pmatrix}
            \partial^X &&\\
            g^{QX} & \partial^{Q} &\\
            p^{ZX} & g^{ZQ} & \partial^{Z}
        \end{pmatrix}
    \end{equation}
    
    By Lemma \ref{lem:X-layer-err-4D}, we can obtain $\hat{e}^x$ based on the syndrome $\pi^x\partial e$ where we take $\hat{e}^x=0$ if $\pi^x \partial e=0$. 
    Let $\hat{e}^X$ be the direct sum of $\hat{e}^x$ over $x\in \sX$ so that
    \begin{align}
        \Delta(\hat{e}^X) &= O(\weight_X \qubit_Z \min(\weight_X,\weight_Z)) |\partial e|\\
        |g^{QX} \hat{e}^X| &\le \weight_X |\pi^X\partial e|  \\
        |p^{ZX} \hat{e}^X| &\le \weight_X\qubit_Z \min(\weight_X,\weight_Z) |\pi^X \partial e|
    \end{align}
    By Lemma \ref{lem:Q-layer-err-4D}, we can obtain $\hat{e}^q$ based on the syndrome 
    \begin{equation}
        \pi^q\partial (e+\hat{e}^X) = \pi^q\partial e +g^{qX}\hat{e}^X
    \end{equation}
    where we take $\hat{e}^q=0$ if the syndrome $=0$.
    Let $\hat{e}^Q$ be the direct sum of $\hat{e}^q$ over $q\in \sQ$ so that
    \begin{align}
        \Delta(\hat{e}^Q) &\le (1+\qubit_Z)\left(|\pi^{Q} \partial e| +\weight_X |\pi^X \partial e|\right) \\
        &= O(\weight_X \qubit_Z) |\partial e|
    \end{align}
    By Lemma \ref{lem:Z-layer-err-4D}, we can obtain $\hat{e}^z$ based on the syndrome
    \begin{equation}
        \pi^z \partial (e+\hat{e}^X +\hat{e}^Q) = \pi^z\partial e +p^{zX} \hat{e}^X +g^{zQ} \hat{e}^Q
    \end{equation}
    where we take $\hat{e}^z=0$ if the syndrome $=0$.
    Let $\hat{e}^Z$ denote the direct sum of $\hat{e}^z$ over $z\in \sZ$ so that $\Delta(\hat{e}^Z) \le 2$ and that if $\hat{e} =\hat{e}^X +\hat{e}^Q+\hat{e}^Z$, then
    \begin{equation}
        \partial (e+\hat{e}) = \sum_{z} |00\bar{z}\ket 1\{[\pi^z(\partial e +\hat{e}^X +\hat{e}^Q)]\ne 0\}
    \end{equation}
    
    Note that $e^X+\hat{e}^X \in \ker \partial^X = \im \partial^X$ and thus $=\partial^X s^X$ for some 2-chain $s^X\in C^X$.
    Let $s^X$ be obtained in an arbitrary but predetermined manner via a map $e^X+\hat{e}^X\mapsto s^X$.
    Hence,
    \begin{equation}
        \pi^{Q} \partial(e+\hat{e}^X) = \partial^Q(e^Q+g^{QX} s^X)
    \end{equation}
    And thus $\ell^{Q} \equiv e^Q+\hat{e}^Q +g^{QX} s^X \in \ker \partial^Q$ so that $\ell^{A} = [\ell^Q] \in H_1(C^Q)=Q$ is well-defined as a 1-chain in $A$.
    Finally, note that
    \begin{align}
        \pi^z \partial (e+\hat{e}^X+\hat{e}^Q) &= g^{zQ} \ell^Q + \partial^z (e^z +p^{zX} s^X) \\
        [\pi^z \partial (e+\hat{e}^X+\hat{e}^Q)] &= [g^{zQ}] [\ell^{Q}] \\
        [\pi^z \partial (e+\hat{e}^X+\hat{e}^Q)] &= \bra z|\partial^A \ell^A\ket
    \end{align}
    Hence, we see that
    \begin{equation}
        |\partial (e+\hat{e})| = |\partial^A \ell^A|
    \end{equation}
    Note that the previous construction defines a map $e\mapsto \ell^A$ and applies to any $e=\Gamma(t)$ and thus we obtain a path $t\mapsto \ell^A(t)$ with respect to $A$. 
    We further claim that the induced path $t\mapsto \ell^A(t)$ is $\weight_X$-continuous, which we postpone to the end of the proof.
    Note that if $e=\sL$, i.e., at the destination of the Pauli path $\Gamma(t)$, then $\hat{e}=0$ and that $\ell^A \in \phi[\sL]$. 
    Hence, the destination of Pauli path $\ell^A(t)$ is a nontrivial logical in the coset $\phi[\sL]\in H_1(A)$.
    Note that for any $t$, we have shown that
    \begin{equation}
        |\partial^A\ell^A(t)| = O(\weight_X\qubit_Z \min(\weight_X,\weight_Z))|\partial \Gamma(t)|
    \end{equation}
    Hence,
    \begin{equation}
        \Delta_{\weight_X}(\phi[\sL]) =O(\qubit_Z \min(\weight_X,\weight_Z)) \Delta_1(\sL)
    \end{equation}
    
    Let us now prove that $t\mapsto \ell^A(t)$ is $\weight_X$-continuous.
    Indeed, since $\Gamma$ is 1-continuous, we see that $e'=\Gamma(t+1)$ must differ from $e=\Gamma(t)$ by a single 1-cell in $C$, and thus there are three possibilities: the difference occurs in $e^x$ for some $x\in \sX$, or $e^q$ for some $q\in \sQ$, or $e^z$ for some $z\in \sZ$. 
    In the first case, we have $e'+e=f^x$ for some $x\in \sX$ and thus $|g^{qX}(s^X +s'^X)|\ne 0$ only if $q\sim x$ and thus changes at most $\weight_X$ qubit layers. Hence, $\ell'^q\ne \ell^q$ only if $q\sim x$ and thus $|\ell'^{A} -\ell^A| \le \weight_X$.
    In the second case, we have $e'+e=f^q$ for some $q\in \sQ$ and thus only $\ell'^q,\ell^q$ may differ. Hence, $|\ell'^{A} -\ell^{A}| \le 1$.
    In the third case, it's clear that $\ell^A$ is independent of $e^z$ and thus $\ell'^A =\ell^A$.
    Hence, $t\mapsto \ell^{A}(t)$ is $\weight_X$-continuous.
\end{proof}

\begin{proof}[Proof of Theorem \ref{thm:energy-4D}]
    The statement follows from Theorem \ref{thm:path-reduction-4D} and Remark \ref{rem:step-size}.
\end{proof}

\section{Proof of 5D Embedding}
\label{sec:5Dproof}
In this section, we provided proofs for the main results collected in Theorem \ref{thm:logical-5D}, \ref{thm:distance-5D} and \ref{thm:energy-5D}.
The proofs follow the similar reasoning as \cite{williamson2023layer,layerreduction,yuan2025unified}, but included for completeness.

\subsection{Logicals Preserved}
\begin{proof}[Proof of Theorem \ref{thm:logical-5D}]
    The proof is similar to that of Theorem \ref{thm:logical-4D} and thus we will only discuss the differences.
    By Theorem \ref{thm:contracting-cycles-5D} and \cite{yuan2025unified}, the output code and input code has isomorphic logicals, i.e.,
    \begin{equation}
        H_1(C) \cong H_1(A)
    \end{equation}
    
    Note that the number of qubits of $C$ per edge follows from Remark \ref{rem:layer-congestion-5D} and the fact that an edge is either along the $\hat{x},\hat{z}$ axis or within the qubit grid.

    The weights of the 5D Layer Code $C$ is determined via the following diagram (where $D=5$)
    \begin{equation}
    \label{eq:weight-diagram-5D}
    \begin{tikzpicture}[baseline]
    \matrix(a)[matrix of math nodes, nodes in empty cells, nodes={minimum size=25pt},
    row sep=2em, column sep=2em,
    text height=1.25ex, text depth=0.25ex]
    {&& \red{C^{X}_{2}}  & \red{C^{X}_{1}} & \red{C^{X}_{0}}\\
    & \gray{C^{Q}_{2}}  & \gray{C^{Q}_{1}}  & \gray{C^{Q}_{0}} &\\
    \blue{C^{Z}_{2}} & \blue{C^{Z}_{1}} & \blue{C^{Z}_{0}} &&\\};
    \path[-left to,font=\scriptsize,transform canvas={yshift=0.2ex}]
    (a-1-3) edge node[above]{$2(D-1)$}  (a-1-4)
    (a-1-4) edge node[above]{$2(D-2)$}  (a-1-5)
    (a-2-2) edge node[above]{$4$}  (a-2-3)
    (a-2-3) edge node[above]{$2$}  (a-2-4)
    (a-3-1) edge node[above]{$4$}  (a-3-2)
    (a-3-2) edge node[above]{$2$}  (a-3-3);
    \path[left to-,font=\scriptsize,transform canvas={yshift=-0.2ex}]
    (a-1-3) edge node[below]{$2$}  (a-1-4)
    (a-1-4) edge node[below]{$4$}  (a-1-5)
    (a-2-2) edge node[below]{$2$}  (a-2-3)
    (a-2-3) edge node[below]{$4$}  (a-2-4)
    (a-3-1) edge node[below]{$2(D-2)$}  (a-3-2)
    (a-3-2) edge node[below]{$2(D-1)$}  (a-3-3);
    \path[-left to,font=\scriptsize,transform canvas={xshift=0.2ex}]
    (a-1-3) edge node[right]{$1$}  (a-2-3)
    (a-1-4) edge node[right]{$1$}  (a-2-4)
    (a-2-2) edge node[right]{$1$}  (a-3-2)
    (a-2-3) edge node[right]{$1$}  (a-3-3);
    \path[left to-,font=\scriptsize,transform canvas={xshift=-0.2ex}]
    (a-1-3) edge node[left]{$1$}  (a-2-3)
    (a-1-4) edge node[left]{$1$}  (a-2-4)
    (a-2-2) edge node[left]{$1$}  (a-3-2)
    (a-2-3) edge node[left]{$1$}  (a-3-3);
    \path[-left to,font=\scriptsize]
    (a-1-3) edge[bend right=80, dashed] node[right]{$D-2$} (a-3-2)
    (a-1-4) edge[bend left=80, dashed] node[right]{$1$} (a-3-3);
    \path[left to-,font=\scriptsize]
    (a-1-3) edge[bend right=80, dashed] node[left]{$1$} (a-3-2)
    (a-1-4) edge[bend left=80, dashed] node[left]{$D-2$} (a-3-3);
    \end{tikzpicture}
    \end{equation}
    Specifically, the weights of the row $C^{Q}$ are straightforwardly obtained since each qubit layer is a surface code.
    The weights of $C^X,C^Z$ can also be obtained straightforwardly since the cells of each $x,z$-check layer correspond to faces, edges and vertices in a 4D sub-grid of the 5D hypercube. 
    Note that in this case, a vertex can have at most $2(D-1)$ adjacent edges and an edge can have at most $2(D-2)$ adjacent faces.
    
    Finally, let us focus on the weights of $g^{QX},g^{ZQ}$ and $p^{ZX}$.
    By Eq. \eqref{eq:gQX-5D}, it's clear that $g^{QX}$ has max column weight $=1$. 
    By Eq. \eqref{eq:gZQ-5D}, we see that $g^{ZQ}$ has max column weight $=1$ since
    \begin{equation}
        |g^{ZQ}|i^{\bullet},q,k\ket| \le \sum_{z\sim q} 1\{\eta(z)=k\} \le 1
    \end{equation}
    where we used the fact that if $q\sim z,z'$, then $\eta(z)\ne \eta(z')$, which follows from Theorem \ref{thm:line-coloring-5D}.
    Since the construction is symmetric, it's straightforward to check that $g^{QX},g^{ZQ}$ have max row weights $=1$.
    Now consider the defect maps as defined in Eq. \eqref{eq:pZX-2-5D}-\eqref{eq:pZX-1-5D}.
    Note that
    \begin{align}
        |p^{ZX} |x;iqk\ket|&\le 1\{i=\eta(x)\}\sum_{z\sim x} 1\{k=\eta(z)\} \\
        &\quad \times \sum_{m} 1\{q\in \gamma_m(xz)\} \\
        &\le 1\{i=\eta(x)\}\sum_{z\sim x} 1\{k=\eta(z)\} \\
        &\quad\quad \times |\{m:q\in \gamma_m(xz)\}|
    \end{align}
    Let $m_1(z),...,m_{\ell}(z)$ denote the possible $\gamma_m(xz)\ni q$. 
    Note that $\lambda(z),\lambda(z')$ do not share edges for $z,z'$ with the same color and that by Corollary \ref{cor:defect-congestion-5D}, $\gamma_m(xz),m\in [M]$ are all simple edge-disjoint paths (for fixed $xz$).
    Since there are $2(D-2)$ edges adjacent to $q$, we see that the collection of $m_1(z),...,m_{\ell}(z),\eta(z)=k,z\sim x$ must be $\le (D-2)$.
    Hence,
    \begin{equation}
        |p^{ZX} |x;iqk\ket|\le D-2
    \end{equation}
    Similarly, note that
    \begin{align}
        |p^{ZX} |x;iek\ket | &\le 1\{i=\eta(x)\} \sum_{z\sim x}  1\{k=\eta(z)\} \\
        &\quad \times \sum_{m} 1\{e \in \gamma_m(xz)\} \\
        &\le 1\{i=\eta(x)\} \sum_{z\sim x}  1\{k=\eta(z)\} \\
        &\quad \times 1\{e\in \gamma_m(xz), \exists m\} \\
        &\le 1
    \end{align}
    where we used Corollary \ref{cor:defect-congestion-5D}, and the fact that $\lambda(z),\lambda(z')$ do not share an edge if $z,z'$ have the same line color. The case is similar for the max row weight, and thus the check weight follows.
\end{proof}

\subsection{Distance Preserved}

To prove that the distance of the 5D Layer Codes is preserved in some manner, the procedure is similar to that of the $\infty$-dimensional Layer Codes \cite{layerreduction} (which also applies to the 3D case in\cite{williamson2023layer,yuan2025unified} and 4D case in Theorem \ref{thm:distance-4D}.). 
Specifically, we will utilize Appendix (A) of \cite{layerreduction} to compute the relative expansion coefficient of the check layers $C^{X}$ ($C^{Z}$) with respect to the gluing map $g^{QX}$ ($g^{ZQ}$), and then apply the Cleaning Lemma in \cite{yuan2025unified} to prove a lower bound on the 4D Layer Codes.
In particular, since the 4D Layer Codes are symmetric in $X,Z$-checks, we shall restrict our attention to the $X$-type distance.

\begin{lemma}[$\lambda$ (co)expansion]
    For any $x$-check, $\lambda(x)$ is $2/|\lambda(x)|$-coexpanding, and $1/|\lambda(x)|$-coexpanding relative to the projection onto $w\le \weight$ qubits $q\sim x$.
\end{lemma}
\begin{proof}
    The proof is exactly the same as Example (A.2) of \cite{layerreduction} and thus omitted.
\end{proof}


\begin{lemma}[Relative Expansion]
    \label{lem:relative-expansion-5D}
    Fix any $x$ check. Let $\tilde{\delta}^{x}$ be the codifferential of $\eta\lambda(x)\otimes R_Z$. Then for every subset of vertices $s\in \eta\lambda(x)\otimes R_Z$, there exists $\hat{s}$ with $\delta^{x} s=\delta^{x} \hat{s}$ such that
    \begin{equation}
        |\tilde{\delta}^x s| = \Omega\left(\frac{1}{\weight^4\qubit^2}\right) |g^{QX} \hat{s}|
    \end{equation}
    
    In particular, if $\delta^{x}$ is the codifferential of $C(x)$, then for every subset of vertices $s\in C(x)$, there exists $\hat{s}$ with $\tilde{\delta}^{x}s = \tilde{\delta}^{x}\hat{s}$ such that
    \begin{equation}
        |\delta^x s| = \Omega\left(\frac{1}{\weight^4\qubit^2}\right) |g^{QX} \hat{s}|
    \end{equation}
\end{lemma}
\begin{proof}
    The proof is exactly the same as that of Lemma \eqref{lem:relative-expansion-4D} and thus omitted.
\end{proof}

\begin{proof}[Proof of Theorem \ref{thm:distance-5D}]
    By Lemma \ref{lem:relative-expansion-5D} and the Cleaning Lemma in \cite{yuan2025unified}, the 1-systolic distance must satisfy
    \begin{equation}
        d_1(C) =\Omega\left(\frac{L_Z}{\weight^4\qubit^2}\right) d_1(A)
    \end{equation}
    where the extra factor $L_Z$ utilizes the fact that each qubit layer $R_X\otimes R_Z^\top$ has 1-systolic distance $\ge L_Z$. 
    Since the construction is symmetric in $X\lr Z$, the statement also holds for the 1-cosystolic distance.
\end{proof}

\subsection{Energy Barrier Preserved}

To prove that the energy barrier of the 5D Layer Codes is preserved in some manner, the procedure is similar to that of the 3D case \cite{williamson2023layer} (see Appendix \ref{sec:3D-layer-code} for a detailed exposition) and 4D case in Theorem \ref{thm:energy-4D}, though similar to the 4D case, additional care must be taken due to the contracting layers $\eta\Lambda\omega (x)$.
In particular, since the 5D Layer Codes are symmetric in $X,Z$-checks, we shall restrict our attention to the $X$-type energy barrier.

\begin{definition}[Notation]
    Let $\pi^X,\pi^Q,\pi^Z$ denote the projection chain maps $C\to C^X,C^Q,C^Z$, respectively.
    Let $\pi^x,\pi^q,\pi^z$ denote the projection chain maps onto $x$-check, qubit $q$, $z$-check layer, respectively.
    Let $\partial^x=\pi^x\partial^X$ and similarly define $\partial^q,\partial^z$.
    Let $g^{qX}=\pi^q g^{QX}$ and similarly define $g^{zQ}$ and $p^{zX}$.
\end{definition}

\begin{lemma}[$X$ layer error]
    \label{lem:X-layer-err-5D}
    Let $\sigma^x\in C(x)^\top$ be a 0-chain (syndrome). Then there exists a map $\sigma^x\mapsto \hat{e}^x$ such that the 1-chain $\hat{e}^x$ is such that $\partial^x\hat{e}^x =\sigma^x$ and
    \begin{equation}
        \Delta_1(\hat{e}^x)\le (1+\weight^2 \qubit)|\sigma^x|
    \end{equation}
    and
    \begin{align}
        |g^{qX} \hat{e}^x| &\le 1\{q\sim x\} |\sigma^x| \\
        |p^{zX} \hat{e}^x| &\le 1\{z\sim x\} \weight |\sigma^x|
    \end{align}
    Moreover, if $\sigma^x=0$, then $\hat{e}^x=0$.
\end{lemma}
\begin{proof}
    For every 0-cell $v^x\in \sigma^x$ (a face in $C(x)$), we see that either $v^x\in \eta\lambda(x)\otimes R_Z$ or $\in (\eta\Lambda\omega)(x)$.
    
    If $v^x\in \eta\lambda(x)\otimes R_Z$, then similar to the 4D case in Lemma \ref{lem:X-layer-err-4D}, we see that $v^x= |x;\eta(x),e,k^+\ket$ for some $k\in [L_Z]$ and edge $e\in \lambda(x)$.
    Note that $e$ is along the $i$-coordinate axis for some $i\in [3]$ and thus it's coordinate $e|_{\ne i}$ are fixed and well-defined.
    Let $e|_{i}$ denote the smaller value of the corresponding endpoints, i.e., if $e=qq'$ with $q|_{i}<q'|_{i}$, then $e|_i=q|_i$.
    Define the string
    \begin{equation}
        [0,e|_{i}]\ni t \mapsto \gamma_t(v^x) = \sum_{a\le t} |x;\eta(x),a\times e|_{\ne i},k^+\ket
    \end{equation}
    By definition of $\lambda(x)$, we see that $t\mapsto \gamma_t (v^x)$ is a 1-continuous Pauli path in $\eta \lambda(x)\otimes R_Z$ and that
    \begin{align}
        |\partial^x \gamma_t (v^x)| &\le 1, \quad \forall t\\
        \partial^x \gamma_T (v^x) &= \sigma^x
    \end{align}
    where $T=e|_{i}$.
    Also note that
    \begin{align}
        |g^{qX} \gamma_t(v^x)| &\le \sum_{a \le t} 1\{q=a\times e|_{\ne i} \sim x\} \\
        &\le 1\{q\sim x\}
    \end{align}
    And by definition of the defect maps in Eq. \eqref{eq:pZX-1-5D}, $p^{zX}\gamma_t (v^x)=0$ for all $z\in \sZ$ and $t$; since $k^+$ is a half integer, while $p^{zX}$ only acts on whole integer coordinates.
    Hence, we see that
    \begin{equation}
        \Delta_1(\gamma_T(v^x))\le \Delta(\gamma(v^x))\le 1+\weight
    \end{equation}
    If $v^x \in (\eta \Lambda \omega)(x)$, then similar to the 4D case in Lemma \ref{lem:X-layer-err-4D}, $v^x =|x;\eta(x),f,\omega(x)\ket$ for some face $f\in \Lambda(x)$.
    The proof then proceeds similarly and thus omitted.
\end{proof}
\begin{lemma}[$Q$ layer error]
    \label{lem:Q-layer-err-5D}
    Let $\sigma^q\in C(q)$ be a 0-chain. Then there exists map $\sigma^{q}\mapsto \hat{e}^q$ such that the 1-chain $\hat{e}^q$ satisfies $\partial^q \hat{e}^q=\sigma^q$ and
    \begin{equation}
        \Delta_1(\hat{e}^{q}) \le (1+\qubit)|\sigma^q|
    \end{equation}
    Moreover, if $\sigma^q=0$, then $\hat{e}^q=0$.
\end{lemma}
\begin{proof}
    The proof is exactly the same as Lemma \ref{lem:Q-layer-err-4D} and thus omitted.
\end{proof}
\begin{lemma}[$Z$ layer error]
    \label{lem:Z-layer-err-5D}
    Let $\sigma^{z}\in C(z)$ be a 0-chain. Then there exists map $\sigma^z \mapsto \hat{e}^z$ such that the 1-chain $\hat{e}^z$ satisfies 
    \begin{equation}
        \partial^z \hat{e}^z + \sigma^z = |00\bar{z}\ket 1\{[\sigma_z]\ne 0\}
    \end{equation}
    where $[\sigma^z]\in H_0(C^Z)$, and  $\Delta_1(\hat{e}^{z}) \le 2$.
    Moreover, if $\sigma^z=0$, we can choose $\hat{e}^z=0$.
\end{lemma}
\begin{proof}
    The proof is similar to the 3D case in Lemma \ref{lem:Z-layer-err-3D} and 4D case in Lemma \ref{lem:Z-layer-err-4D} and thus omitted.
\end{proof}
\begin{theorem}[Path Reduction]
    \label{thm:path-reduction-5D}
    Let $\sL$ be a nontrivial logical of the 5D Layer Code $C$ so that $[\sL]\in H_1(C)$. If $\phi:H_1(C) \to H_1(A)$ denotes the isomorphism in Theorem 1 of \cite{yuan2025unified}, then
    \begin{equation}
        \Delta_{\weight}(\phi [\sL]) =O(\weight^2 \qubit) \Delta_1(\sL)
    \end{equation}
\end{theorem}

\begin{proof}
    The proof is exactly the same as the 4D case in Theorem \ref{thm:path-reduction-4D} and thus omitted.
\end{proof}

\begin{proof}[Proof of Theorem \ref{thm:energy-5D}]
    The statement follows from Theorem \ref{thm:path-reduction-5D} and Remark \ref{rem:step-size}.
\end{proof}

\section{Color Routing for Arbitrary Dimensions}
\label{sec:anyD}

Note that the 4D and 5D embedding both utilize color routing (Theorem \ref{thm:line-coloring-4D}, \ref{thm:plane-coloring-4D}, \ref{thm:line-coloring-5D}, \ref{thm:plane-coloring-5D}), and thus it's natural to suspect that higher dimensional Layer Code constructions also utilize analogous color routing schemes.

In this section, we generalize the color routing scheme to all $\fD=D-2$ qubit grids $[L]^{\fD}$ and compare with the 4D and 5D cases, so that if an analogue of Theorem \ref{thm:contracting-cycles-5D} can be proven for $\fD$ dimensions, then the $D$-dimensional Layer Code construction immediately follows.

\subsection{Check Layers}


\subsubsection{Check Routing}

\begin{definition}[Induced Routing on $Q$]
    \label{def:induced-graph-Q}
    Let $\sG_Q=(\sQ,\sP_Q)$ denote the \textbf{induced graph} on $Q$ with vertices consisting of qubits $\sQ$ and edges (also referred as \textbf{packets}) $p=qq'$ if there exists an $x$-check (or $z$-check) such that $q,q'\sim x$ (or $q,q'\sim z$). 
    Note that the max degree is $O(\weight\qubit)$.
    
    Order each packet $p$ based on the lexicographic order of $\sQ=[L]^{\fD}$.
    Let the \textbf{induced routing problem} on $Q$ be such that $\gamma_{0},\gamma_\infty:\sP_Q \to \sQ$ are the start and end of each (directed) edge.
    By definition, we see that each qubit is the source (destination) of $O(\weight\qubit)$ packets (edges in $\sP_Q$).
    We refer to $O(\weight \qubit)$ as the \textbf{density} of the routing problem.
\end{definition}

\begin{definition}[Bit Strings]
    Let $s\in \dF_2^{\lceil \fD/2\rceil }$. 
    Write $\bar{0} =0\cdots 0$ and similarly for $\bar{1}$.
    Define $\delta(s) \le \fD$ to be the last bit index which is distinct from the remaining bit indices, e.g.,
    \begin{equation}
        s = s|_{[\delta)}0\underbrace{1\cdots 1}_{\fD-\delta}, \quad  s|_{[\delta)}1\underbrace{0\cdots 0}_{\fD-\delta}
    \end{equation}
    where we take the convention $\delta(\bar{0})=\delta(\bar{1})=\lceil\fD/2\rceil$.
    Also note that if $\delta(s)=\delta$ for $s\ne \bar{0},\bar{1}$, then there are only two possibilities for $s|_{[\delta,\fD]}$, i.e., either $1\cdots 10$ or $0\cdots 01$.
\end{definition}

\begin{theorem}[Color Route]
    \label{thm:color-route}
    Consider the induced routing problem $\gamma_0\to \gamma_\infty$ of $\sP_Q$ on $[L]^{\fD}$. 
    Then for each packet $p\in \sP_Q$, there exists a sequence $\gamma_s(p)\in[L]^{\fD} ,s\in \dF_2^{\lceil \fD/2\rceil}$, referred as the \textbf{color route}, such that the following properties are satisfied.
    \begin{enumerate}[label=\arabic*)]
        \item (Source and Destination)
        \begin{align}
            \gamma_{\bar{0}}(p) &= \gamma_0(p)\\
            \gamma_{\bar{1}}(p) &= \gamma_\infty(p)
        \end{align}
        \item (One Plane at a Time) If $s|_{[i]}=s'|_{[i]}$, then 
        \begin{equation}
            \gamma_s(p)|_{[2i]}=\gamma_{s'}(p)|_{[2i]}
        \end{equation}
        If, in addition, $\delta(s)=\delta(s')=i+1$, then $\gamma_s(p),\gamma_{s'}(p)$ have the same coordinates except possibly for the $2i+1$- and $2i+2$- coordinate.
        \item (Density) For any $s\in \dF_2^{\fD}$ and $q\in [L]^{\fD}$, 
        \begin{equation}
            |\{p:\gamma_s(p)=q\}| =O(\weight \qubit)
        \end{equation}
    \end{enumerate}
\end{theorem}
\begin{proof}
    Let us consider $\fD=1$ or $\fD=2$ as the base case and proceed by induction on the dimension $\fD$, i.e., suppose that the statement holds for $\fD-2$ dimensions for $\fD\ge 3$.   
    
    Consider a (multi-edge) bipartite graph $\sH$ with vertex partitions $\sV_{0},\sV_\infty =[L]^2$, and an edge between $i_t\times j_t \in \sV_t,t=0,\infty$ if there exists a packet $p$ with source $\gamma_0(p)|_{[2]}= i_0\times j_0$ and destination $\gamma_\infty(p)|_{[2]}=i_\infty \times j_\infty$.
    Since each source and destination has $O(\weight \qubit)$ packets, we see that the (multi-edge) bipartite graph has max degree $O(\weight \qubit)L^{\fD-2}$.
    By the edge-coloring theorem, we see that the edges can be colored with $O(\weight\qubit)L^{\fD-2}$ many colors so that edges with the same color do not share a vertex.
    Note that there is a one-to-one correspondence between edges in the multi-edge graph and packet $p\in \sP_{Q}$.
    By grouping at most $O(\weight \qubit)$ colors together in an arbitrary manner, we can define a coloring map $\xi:\sP_{Q} \to [L]^{\fD-2}$ such that the following subsets of packets satisfy
    \begin{equation}
        |\{p: (\gamma_t|_{[2]} \times \xi)(p)=q\}| = O(\weight \qubit)
    \end{equation}
    For any $q\in \sQ$ and $t=0,\infty$.
    Let 
    \begin{align}
        \gamma_{01\cdots 1}(p) &= ( \gamma_{0}|_{[2]} \times \xi)(p)\\
        \gamma_{10\cdots 0}(p) &= ( \gamma_\infty|_{[2]} \times \xi)(p)
    \end{align}


    Fix $i_0\times j_0\in [L]^2$ and consider the routing problem in $[L]^{\fD-2}\cong i_0\times j_0\times [L]^{\fD-2}$ from $\gamma_0\to \gamma_{01\cdots 1}$ for all packets with fixed $\gamma_0(p)|_{[2]} = i_0\times j_0$.
    By induction, there exists color route $\gamma_{0s},s\in \dF_2^{\lceil \fD/2\rceil -1}$ from $\gamma_0(p) \to \gamma_{01\cdots 1}(p)$ in $i_0\times j_0\times [L]^{\fD-2}$.
    The case is similar for the routing problem in $i_\infty\times j_\infty\times [L]^{\fD-2}$ from $\gamma_{10\cdots 0}\to \gamma_\infty$ for all packets with fixed $\gamma_\infty(p)|_{[2]}=i_\infty\times j_\infty$ so that we obtain color route $\gamma_{1s},s\in \dF_2^{\lceil \fD/2\rceil -1}$.
    Then it's straightforward to check that $\gamma_s,s\in \dF_2^{\lceil \fD/2 \rceil}$ satisfies property (1)-(3).
\end{proof}


\begin{definition}[Star Graph]
    Let $q\in [L]^\fD$.
    Then the \textbf{star graph} of $q$ is defined as
    \begin{equation}
        \lambda(q) = \bigcup_{i} q|_{[1,i)} \times [L]\times q|_{(i,\fD]}
    \end{equation}
    so that $|\lambda(q)| = O(\fD)L$.
\end{definition}

\begin{definition}[Check Star Graphs]
    For every $x$-check, define the \textbf{star graph} for $x$ as
    \begin{equation}
        \lambda(x) = \bigcup_{s, p\sim x} \lambda(\gamma_s(p))
    \end{equation}
    where the \textbf{adjacency} $p\sim x$ is shorthand for the condition that the endpoints $q,q'$ of $p$ are both contained in the support of $x$.
    In particular,
    \begin{equation}
        |\lambda(x)| = O(2^{\fD/2}\fD \weight^2 )L
    \end{equation}
    The case is similar for $\lambda(z)$ for $z$ checks
\end{definition}

\begin{theorem}[Line Coloring]
    \label{thm:line-coloring-anyD}
    There exists a \textbf{line coloring} function $\eta=\eta_X:\sX\to [\chi_X]$ where 
    \begin{equation}
        \label{eq:line-color}
        \chi_X =O(2^{\fD} \fD\weight^3 \qubit^2)L
    \end{equation}
    such that if $\sX_{\eta}$ denotes a partition with fixed color $\eta(x)=\eta$, then the collection of star routes $\lambda(x),x\in \sX_{\eta}$ has $1$-congestion in $[L]^3$ for every partition $\sX_\eta$.
\end{theorem}
\begin{proof}
    The proof is exactly the same as that of Theorem \ref{thm:line-coloring-5D} and thus omitted.
    The only difference is that the summation over bit strings $s,s'\in \dF_2^{\lceil \fD/2\rceil}$ and indices $i\in [\fD]$ induces a dependence of $\chi_X$ on the qubit grid dimension $\fD$.
\end{proof}

Similar to the 4D and 5D scenario, by the previous Theorem, for any given $x\in \sX$, define
\begin{equation}
    \eta\lambda(x) =\eta(x)\times \lambda (x) \subseteq [\chi_X]\times [L]^{\fD}
\end{equation}
Then it's clear that $\eta\lambda(x),x\in \sX$ is embedded in $\hat{x} \times \hat{q}$ with $1$ edge congestion.
The case is similarly defined for the $Z$-checks with line coloring function $\eta=\eta_Z:\sZ \to [\chi_Z]$ so that
\begin{equation}
    \eta\lambda(z) =\lambda(z) \times \eta(z) \subseteq  [L]^{\fD} \times [\chi_Z]
\end{equation}
over $z\in \sZ$ is embedded in $\hat{q}\times \hat{z}$ with $1$ edge congestion. 


\subsubsection{Contracting Cycles}

Let us denote the planes induced by filling in the star routes $\lambda(q)$ via $\Lambda(q)$, which are naturally 3-term cell complexes.
The filled planes $\Lambda(p),\Lambda(x)$ are similarly defined.
Specifically
\begin{definition}[Star Planes]
    Let $q\in \sQ=[L]^{\fD}$. Then define the \textbf{star plane} of $q$ as
    \begin{equation}
        \Lambda(q) = \bigcup_{i<j} q|_{[1,i)}\times [L] \times q|_{(i,j)} \times [L] \times q|_{(j,\fD]}
    \end{equation}
    so that $|\Lambda(q)|=O(\fD^2)L^2$ in terms of vertices.
\end{definition}

\begin{definition}[Check Star Planes]
    Let $\gamma_{\bullet}(p)\in [L]^{\fD}$ denote the color route given by the induced routing problem in Theorem \ref{thm:color-route}. 
    For each $x$-check, further define the \textbf{star plane} as
    \begin{equation}
        \Lambda(x) = \bigcup_{s,p\sim x} \Lambda(\gamma_s(p))
    \end{equation}
    And similarly for the $z$ checks.
\end{definition}
\begin{theorem}[Plane Coloring]
    \label{thm:plane-coloring-anyD}
    Let $\sX_{\eta}$ be a partition in Theorem \ref{thm:line-coloring-anyD} such that star graphs $\lambda(x)$ do not share the same edge, i.e., has 1 edge congestion.
    Then there exists a \textbf{plane coloring} map $\omega=\omega_X:\sX\to [\zeta_X]$ where
    \begin{equation}
        \zeta_X = O(2^{\fD}\fD^2\weight^2)L
    \end{equation}
    such that if $\sX_{\eta,\omega}$ is the collection of $x\in \sX_{\eta}$ with fixed \textbf{plane} color $\omega$, then the collection $\Lambda(x),x\in \sX_{\eta,\gamma}$ have $1$ face congestion.
    The case is similar for $Z$-check with plane coloring function $\omega=\omega_Z:\sZ \to [\zeta_Z]$
\end{theorem}
\begin{proof}
    The proof is exactly the same as that of Theorem \ref{thm:plane-coloring-5D} and thus omitted.
    The only difference is that the summation over bit strings $s,s'\in \dF_2^{\fD}$ and indices $i,j\in [\fD]$ induces a dependence of $\chi_X$ on the qubit grid dimension $\fD$.
\end{proof}

\begin{definition}[Dimension of 5D Grid]
    Let $\chi_X,\chi_Z$ be that defined in Theorem \ref{thm:line-coloring-5D} and $\zeta_X,\zeta_Z$ be that in Theorem \ref{thm:plane-coloring-5D}. Define
    \begin{align}
        L_X &= \max(\chi_X,\zeta_Z) \\
        L_Z &= \max(\chi_Z,\zeta_X)
    \end{align}
    so that the D-dimensional grid is given by $[L_X]\times [L]^{\fD}\times [L_Z]$ and
    \begin{equation}
        L_X,L_Z = O(2^{\fD} \fD^2 \weight^3 \qubit^2)L 
    \end{equation}
\end{definition}

\begin{definition}[Contracting Layers]
    For each $x$-check, define the \textbf{contracting layer} as
    \begin{equation}
        \eta\Lambda\omega(x)=\eta(x)\times \Lambda(x) \times \omega(x)
    \end{equation}
    so that the collection $\eta\Lambda\omega(x),x\in \sX$ is embedded in $[L_X]\times [L]^{\fD}\times [L_Z]$ with 1 (face) congestion.
    The case is similar for the $z$-checks.
    Let $R_X,R_Z$ denote the repetition codes on $L_X,L_Z$ vertices, respectively.
\end{definition}

\begin{definition}[Check layers]
    \label{def:check-layers-anyD}
    Define the $x$-check layer as
    \begin{equation}
        C(x) = \eta\lambda(x)\otimes R_Z +\eta\Lambda\omega(x)
    \end{equation}
    Similarly, define the $z$-check layer as
    \begin{equation}
        C(z)=R_X\otimes \lambda\eta (z) +\omega\Lambda\eta(z)
    \end{equation}
\end{definition}

\begin{remark}[Layer Congestion]
    \label{rem:layer-congestion-anyD}
    Similar to the 5D case in Remark \ref{rem:layer-congestion-5D}, note that $\eta\lambda(x),x\in \sX$ has 1 edge congestion, but after tensoring with $R_Z$, it's possible that $\eta\lambda(x)\otimes R_Z$ have $\fD$ edge congestion for edges along the $\hat{z}$ axis since at most $\fD$ distinct $\eta\lambda(x)$ can share the same vertex.
    Similarly, any edge in the qubit grid $[L]^{\fD}$ can be shared by $\fD-1$ many contracting layers.
    Hence, $C(x),x\in \sX$ has
    \begin{itemize}
        \item 0 edge congestion for edges along the $\hat{x}$ axis
        \item $\max(1,D-3)$ edge congestion for edges in the qubit grid
        \item $D-2$ edge congestion for edges along $\hat{z}$ axes
    \end{itemize}
    The case is similar for $Z$-check layers with $\hat{x}\lr \hat{z}$ switched.
\end{remark}

\begin{remark}[$D\ge 6$ Layer Codes]
    \label{rem:anyD-layer-codes}
    Suppose that an analogue of Theorem \ref{thm:contracting-cycles-5D} can be proven for (any) higher dimension $D\ge 6$.
    Then the $D$-dimensional Layer Code construction follows exactly as in Section \ref{sec:together-5D}, with main results given by those in Section \ref{sec:main-result-5D} with 5D replaced by $D$-dimensions.
\end{remark}

\section{Numerical Calculations}
\label{sec:numerics}

The embeddings of Sections \ref{sec:4Dembed} and \ref{sec:5Dembed} are explicit, as is the $3$D Layer Code of \cite{williamson2023layer} that they generalize.
All three can therefore be assembled cell by cell for a given input code and their output parameters read off exactly, rather than estimated.
This appendix records those parameters for a family of quantum Tanner inputs and compares them against the guaranteed scaling $n(C)=\Theta(n^{D/(D-2)})$.

\subsection{Output Size on Quantum Tanner Inputs}
\label{sec:numerics-tanner}

We take as input the quantum Tanner codes of Leverrier and Z\'emor \cite{leverrier2022quantum} on the left-right Cayley complex of a finite abelian group $G$ with symmetric, disjoint generating sets $A,B$ of size $\Delta$, and local codes $C_A,C_B\subseteq \dF_2^{\Delta}$.
The qubits are the squares of the complex, so that $n=|G|\Delta^2$; the $X$-checks are indexed by $V_{00}\cup V_{11}$ and carry a basis of $C_A\otimes C_B$, the $Z$-checks by $V_{01}\cup V_{10}$ and carry a basis of $C_A^{\perp}\otimes C_B^{\perp}$.
Varying $G$, $\Delta$ and the local codes gives a spread of maximum check weight $\weight$ and qubit degree $\qubit$ at comparable $n$, which is what the constants of the embedding depend on.

For each input we build the $D=3,4,5$ Layer Codes on the qubit grid $[L]^{\fD}$ with $\fD=D-2$ and $L=\lceil n^{1/\fD}\rceil$, inside the ambient box $[L_X]\times [L]^{\fD}\times [L_Z]$, and count the output qubits exactly.
Table \ref{tab:numerics} reports $n(C)$ normalised by the guaranteed scaling, i.e. $n(C)/n^{3}$, $n(C)/n^{2}$ and $n(C)/n^{5/3}$ for $D=3,4,5$ respectively.
The logical dimension is $k(C)=k$ in all three cases by Theorem \ref{thm:logical-4D}, Theorem \ref{thm:logical-5D} and \cite{williamson2023layer}, and is not recomputed here.

\begin{table*}[t]
\caption{\label{tab:numerics}Output size of the $3$D, $4$D and $5$D Layer Codes on quantum Tanner inputs \cite{leverrier2022quantum} over abelian $G$, normalised by the guaranteed scaling $n(C)=\Theta(n^{D/(D-2)})$. Here $\Delta$ is the degree of the left-right Cayley complex, $C_A/C_B$ the local codes, $n,k$ the parameters of the input code, and $\weight,\qubit$ its maximum check weight and total qubit degree. Rows are ordered by $n$.}
\begin{ruledtabular}
\begin{tabular}{lccrrcrrr}
$G$ & $\Delta$ & $C_A/C_B$ & $n$ & $k$ & $(\weight,\qubit)$ & $n(3\mathrm{D})/n^{3}$ & $n(4\mathrm{D})/n^{2}$ & $n(5\mathrm{D})/n^{5/3}$ \\
\hline
$\dZ_2^3$ & $2$ & $[2,1]/[2,1]$ & $32$ & $8$ & $(4,4)$ & $1.423$ & $34.3$ & $72.6$ \\
$\dZ_3^2$ & $2$ & $[2,1]/[2,1]$ & $36$ & $2$ & $(4,4)$ & $1.431$ & $32.0$ & $138.9$ \\
$\dZ_{12}$ & $2$ & $[2,1]/[2,1]$ & $48$ & $4$ & $(4,4)$ & $1.448$ & $34.4$ & $148.1$ \\
$\dZ_2^4$ & $2$ & $[2,1]/[2,1]$ & $64$ & $16$ & $(4,4)$ & $1.461$ & $25.9$ & $86.9$ \\
$\dZ_2^3$ & $3$ & $[3,1]/[3,2]$ & $72$ & $14$ & $(6,8)$ & $1.155$ & $80.0$ & $389.8$ \\
$\dZ_2^5$ & $2$ & $[2,1]/[2,1]$ & $128$ & $32$ & $(4,4)$ & $1.481$ & $32.2$ & $209.8$ \\
$\dZ_2^4$ & $3$ & $[3,1]/[3,2]$ & $144$ & $28$ & $(6,8)$ & $1.170$ & $57.0$ & $551.0$ \\
$\dZ_3^2$ & $4$ & $[4,2]/[4,2]$ & $144$ & $24$ & $(9,10)$ & $1.483$ & $56.9$ & $402.8$ \\
$\dZ_3^2$ & $4$ & $[4,1]/[4,3]$ & $144$ & $42$ & $(8,10)$ & $0.831$ & $62.7$ & $595.9$ \\
$\dZ_2^6$ & $2$ & $[2,1]/[2,1]$ & $256$ & $64$ & $(4,4)$ & $1.490$ & $27.9$ & $305.0$ \\
$\dZ_2^4$ & $4$ & $[4,2]/[4,2]$ & $256$ & $60$ & $(9,10)$ & $1.490$ & $90.4$ & $501.5$ \\
$\dZ_{20}$ & $4$ & $[4,2]/[4,2]$ & $320$ & $46$ & $(9,10)$ & $1.492$ & $54.2$ & $626.8$ \\
$\dZ_5^2$ & $4$ & $[4,2]/[4,2]$ & $400$ & $56$ & $(9,10)$ & $1.494$ & $55.8$ & $667.1$ \\
$\dZ_3^3$ & $4$ & $[4,2]/[4,2]$ & $432$ & $68$ & $(9,10)$ & $1.494$ & $53.3$ & $726.5$ \\
$\dZ_2^7$ & $2$ & $[2,1]/[2,1]$ & $512$ & $128$ & $(4,4)$ & $1.495$ & $47.0$ & $137.9$ \\
$\dZ_2^8$ & $2$ & $[2,1]/[2,1]$ & $1024$ & $256$ & $(4,4)$ & $1.498$ & $28.9$ & $455.2$ \\
\end{tabular}
\end{ruledtabular}
\end{table*}

\begin{remark}[The three columns and the routing they require]
    \label{rem:numerics-columns}
    The three columns are not on the same footing, and the reason is how much routing each dimension needs.

    At $\fD=1$ there is no routing and no coloring: the box is exactly $|\sX|\times n\times |\sZ|$, one coordinate per check. Writing $|\sX|=\alpha n$ and $|\sZ|=\beta n$, the qubit layers and the two families of check layers give the exact count
    \begin{equation}
        n(3\mathrm{D}) = 6\alpha\beta\, n^3 - 2(\alpha+\beta+\alpha\beta)\, n^2 + n
    \end{equation}
    so that the $3$D ratio is fixed by the input alone. For these inputs $\alpha = 2k_Ak_B/\Delta^2$ and $\beta = 2(\Delta-k_A)(\Delta-k_B)/\Delta^2$, so $6\alpha\beta$ is $3/2$ for $[2,1]/[2,1]$ and $[4,2]/[4,2]$, $32/27$ for $[3,1]/[3,2]$ and $27/32$ for $[4,1]/[4,3]$ --- the only three values the column takes. It records the rate profile of the input, not the quality of the embedding.

    At $\fD=2$ the routing is the AB gadget of Definition \ref{def:AB-graph}: $\AB(x)$ is the union of the row and column lines of $[L]^2$ through the qubits of $x$, and the coning layer is the whole grid at a single level. Nothing here is chosen --- $\AB(x)$ is determined by which rows and columns $\supp(x)$ occupies --- so only the line and plane colorings are heuristic, and the ambient box stays close to the qubit grid: $L_X/L$ lies between $1.9$ and $8.0$ across the whole table, with no trend in $n$. The $4$D column is correspondingly stable, spanning a factor $1.8$ at fixed $(\weight,\qubit)$: $25.9$--$47.0$ at $(4,4)$ and $53.3$--$90.4$ at $(9,10)$.

    At $\fD=3$ the AB gadget no longer suffices and is replaced by the color route of Theorem \ref{thm:color-route-5D}: each packet acquires an intermediate coordinate from an edge coloring, $\lambda(x)$ is the union of the axis lines through all the resulting points, and $\Lambda(x)$ fills these into planes. Much more is chosen, and the choices interact. Here $L_X/L$ ranges over $2.0$ to $13.5$ and the $5$D column over $72$ to $727$, a factor of ten, with $\chi_X$ still between $12\%$ and $89\%$ of $|\sX|$ rather than the $\Theta(L)$ of Theorem \ref{thm:line-coloring-5D}. How much of this is arithmetic rather than structure is visible in the two inputs with $L^3=n$ exactly, $n=64$ and $n=512$: they fill the grid with no slack and give the two smallest ratios in their neighbourhoods, $86.9$ and $137.9$. Placing the same $n=512$ code on $L=9$ instead raises $\chi_X$ from $32$ to $54$ and the ratio from $137.9$ to $357.6$. The $5$D column is therefore what one particular routing produces at these sizes, not an estimate of a constant.
\end{remark}

\begin{remark}[Dependence on the placement]
    \label{rem:numerics-placement}
    The theorems fix the ambient box but not which point of $[L]^{\fD}$ each input qubit occupies. Above, qubits are placed in lexicographic index order, which for these group-algebra inputs is far from generic: consecutive indices land on a common grid line. Against five random relabelings of each input, the $4$D ratio under the lexicographic placement is $0.35$ to $1.24$ times the random median, while that median itself settles near $80$ and becomes almost independent of the family. Both are legitimate embeddings; the table records the lexicographic one.
\end{remark}

\begin{remark}[Method and verification status]
    \label{rem:numerics-method}
    The $3$D counts come from the $\fD=1$ specialization of the same framework and agree with an explicit cell-by-cell assembly for $n\le 72$.
    The $4$D counts come from Definition \ref{def:check-layer-4D}, with the line and plane colorings produced greedily on the conflict graph of Theorem \ref{thm:line-coloring-4D} exactly as in its proof. For $n\le 144$ they were reproduced by a full sparse assembly of the total complex, which also verifies $D^2=0$, the weight bounds of Theorem \ref{thm:logical-4D}, $L_\infty$ locality $1/2$ and the congestion counts; for $n\le 64$ the $\dF_2$ ranks are also within reach and give $H_s(C)\cong H_s(A)$. Above $n=144$ only the cell counts were computed.
    The $5$D counts come from Definition \ref{def:check-layers-5D} and were reproduced at every $n$ in the table by two independent implementations; the hypothesis $H_1(\Lambda(x))=0$, $H_0(\Lambda(x))=\dF_2$ of Theorem \ref{thm:contracting-cycles-5D} was checked on a sample of four check layers per instance. We have not verified $D^2=0$ or the homology of the $5$D total complex at any of these scales.
\end{remark}

\bibliographystyle{alphaurl}
\bibliography{refs.bib}

@article{layerreduction,
    title={Quantum Weight Reduction with Layer Codes},
    author={Yuan, Andrew C. and Baspin, Nouédyn and Williamson, Dominic J.},
    journal = {arXiv preprint arXiv:2603.04883},
    year = {2026},
    url={https://arxiv.org/abs/2603.04883}
}

@article{portnoy2023local,
  title={Local quantum codes from subdivided manifolds},
  author={Portnoy, Elia},
  journal={arXiv preprint arXiv:2303.06755},
  year={2023},
  url={https://arxiv.org/abs/2303.06755}
}

@article{li2026almost,
  title={Almost optimal geometrically local quantum LDPC codes in any dimension},
  author={Li, Xingjian and Lin, Ting-Chun and Wills, Adam and Hsieh, Min-Hsiu},
  journal={Nature Communications},
  year={2026},
  publisher={Nature Publishing Group UK London},
  url={https://doi.org/10.1038/s41467-026-69031-w}
}

@misc{gottesman2024surviving,
  title={Surviving as a Quantum Computer in a Classical World},
  author={Gottesman, Daniel},
  url={https://www.cs.umd.edu/~dgottesm/QECCbook-2024.pdf},
  year={2024}
}

@article{knill1998resilient,
    author = {Knill, Emanuel and Laflamme, Raymond and Zurek, Wojciech H.},
    title = {Resilient quantum computation: error models and thresholds},
    journal = {Proceedings of the Royal Society A: Mathematical, Physical and Engineering Sciences},
    volume = {454},
    number = {1969},
    pages = {365-384},
    year = {1998},
    month = {01},
    issn = {1364-5021},
    doi = {10.1098/rspa.1998.0166},
    url = {https://doi.org/10.1098/rspa.1998.0166},
    eprint = {https://royalsocietypublishing.org/rspa/article-pdf/454/1969/365/633972/rspa.1998.0166.pdf},
}

@article{aharonov2008faulttolerant,
author = {Aharonov, Dorit and Ben-Or, Michael},
title = {Fault-Tolerant Quantum Computation with Constant Error Rate},
journal = {SIAM Journal on Computing},
volume = {38},
number = {4},
pages = {1207-1282},
year = {2008},
doi = {10.1137/S0097539799359385},
URL = { 
        https://doi.org/10.1137/S0097539799359385
}
}

@inproceedings{shor1996faultolerant,
  author={Shor, P.W.},
  booktitle={Proceedings of 37th Conference on Foundations of Computer Science}, 
  title={Fault-tolerant quantum computation}, 
  year={1996},
  volume={},
  number={},
  pages={56-65},
  doi={10.1109/SFCS.1996.548464}}

@article{yoshida2013information,
   title={Information storage capacity of discrete spin systems},
   volume={338},
   ISSN={0003-4916},
   url={http://dx.doi.org/10.1016/j.aop.2013.07.009},
   DOI={10.1016/j.aop.2013.07.009},
   journal={Annals of Physics},
   publisher={Elsevier BV},
   author={Yoshida, Beni},
   year={2013},
   month=nov, pages={134–166} }

@article{yuan2025unified,
	title        = {{Unified framework for quantum code embedding}},
	author       = {Yuan, Andrew C.},
	year         = 2026,
	month        = {Feb},
	journal      = {Phys. Rev. A},
	publisher    = {American Physical Society},
	volume       = 113,
	pages        = {022438},
	doi          = {10.1103/gtlz-3q9d},
	url          = {https://link.aps.org/doi/10.1103/gtlz-3q9d},
	issue        = 2,
	numpages     = 23
}

@book{rotman2009introduction,
  title = {An Introduction to Homological Algebra},
  ISBN = {9780387683249},
  url = {http://dx.doi.org/10.1007/b98977},
  DOI = {10.1007/b98977},
  publisher = {Springer New York},
  author = {Rotman,  Joseph J.},
  year = {2009}
}

@article{bravyi2010tradeoffs,
   title={Tradeoffs for Reliable Quantum Information Storage in 2D Systems},
   volume={104},
   ISSN={1079-7114},
   url={http://dx.doi.org/10.1103/PhysRevLett.104.050503},
   DOI={10.1103/physrevlett.104.050503},
   number={5},
   journal={Physical Review Letters},
   publisher={American Physical Society (APS)},
   author={Bravyi, Sergey and Poulin, David and Terhal, Barbara},
   year={2010},
   month=feb
}

@article{bravyi2009nogo,
   title={A no-go theorem for a two-dimensional self-correcting quantum memory based on stabilizer codes},
   volume={11},
   ISSN={1367-2630},
   url={http://dx.doi.org/10.1088/1367-2630/11/4/043029},
   DOI={10.1088/1367-2630/11/4/043029},
   number={4},
   journal={New Journal of Physics},
   publisher={IOP Publishing},
   author={Bravyi, Sergey and Terhal, Barbara},
   year={2009},
   month=apr, pages={043029} 
}

@article{williamson2023layer,
	author = {Williamson, Dominic J. and Baspin, Nou{\'e}dyn},
	date = {2024/11/04},
	doi = {10.1038/s41467-024-53881-3},
	id = {Williamson2024},
	isbn = {2041-1723},
	journal = {Nature Communications},
	number = {1},
	pages = {9528},
	title = {Layer codes},
	url = {https://doi.org/10.1038/s41467-024-53881-3},
	volume = {15},
	year = {2024}
}

@inproceedings{panteleev2022asymptotically,
author = {Panteleev, Pavel and Kalachev, Gleb},
title = {Asymptotically good Quantum and locally testable classical LDPC codes},
year = {2022},
isbn = {9781450392648},
publisher = {Association for Computing Machinery},
address = {New York, NY, USA},
url = {https://doi.org/10.1145/3519935.3520017},
doi = {10.1145/3519935.3520017},
booktitle = {Proceedings of the 54th Annual ACM SIGACT Symposium on Theory of Computing},
pages = {375–388},
numpages = {14},
location = {Rome, Italy},
series = {STOC 2022}
}

@inproceedings{leverrier2022quantum,
author = { Leverrier, Anthony and Zemor, Gilles },
booktitle = { 2022 IEEE 63rd Annual Symposium on Foundations of Computer Science (FOCS) },
title = {{ Quantum Tanner codes }},
year = {2022},
volume = {},
ISSN = {},
pages = {872-883},
doi = {10.1109/FOCS54457.2022.00117},
url = {https://doi.ieeecomputersociety.org/10.1109/FOCS54457.2022.00117},
publisher = {IEEE Computer Society},
address = {Los Alamitos, CA, USA},
month =Nov}

@inproceedings{dinur2023good,
author = {Dinur, Irit and Hsieh, Min-Hsiu and Lin, Ting-Chun and Vidick, Thomas},
title = {Good Quantum LDPC Codes with Linear Time Decoders},
year = {2023},
isbn = {9781450399135},
publisher = {Association for Computing Machinery},
address = {New York, NY, USA},
url = {https://doi.org/10.1145/3564246.3585101},
doi = {10.1145/3564246.3585101},
booktitle = {Proceedings of the 55th Annual ACM Symposium on Theory of Computing},
pages = {905–918},
numpages = {14},
location = {Orlando, FL, USA},
series = {STOC 2023}
}

@article{horsman2012surface,
doi = {10.1088/1367-2630/14/12/123011},
url = {https://doi.org/10.1088/1367-2630/14/12/123011},
year = {2012},
month = {dec},
publisher = {IOP Publishing},
volume = {14},
number = {12},
pages = {123011},
author = {Horsman, Dominic and Fowler, Austin G and Devitt, Simon and Meter, Rodney Van},
title = {Surface code quantum computing by lattice surgery},
journal = {New Journal of Physics}
}

@article{sabo2024weight,
  title={Weight-reduced stabilizer codes with lower overhead},
  author={Sabo, Eric and Gunderman, Lane G and Ide, Benjamin and Vasmer, Michael and Dauphinais, Guillaume},
  journal={PRX Quantum},
  volume={5},
  number={4},
  pages={040302},
  year={2024},
  publisher={APS},
  url = {https://link.aps.org/doi/10.1103/PRXQuantum.5.040302}
}

@misc{hsieh2025simplified,
      title={Simplified Quantum Weight Reduction with Optimal Bounds}, 
      author={Min-Hsiu Hsieh and Xingjian Li and Ting-Chun Lin},
      year={2025},
      eprint={2510.09601},
      archivePrefix={arXiv},
      primaryClass={quant-ph},
      url={https://arxiv.org/abs/2510.09601}, 
}

@article{baspin2024wire,
  title={Wire codes},
  author={Baspin, Nou{\'e}dyn and Williamson, Dominic},
  journal={arXiv preprint arXiv:2410.10194},
  year={2024},
  url={https://arxiv.org/abs/2410.10194}
}

@misc{dai2024localityvsquantumcodes,
      title={Locality vs Quantum Codes}, 
      author={Samuel Dai and Ray Li},
      year={2024},
      eprint={2409.15203},
      archivePrefix={arXiv},
      primaryClass={quant-ph},
      url={https://arxiv.org/abs/2409.15203}, 
}

@misc{temme2016thermalization,
      title={Thermalization time bounds for Pauli stabilizer Hamiltonians}, 
      author={Kristan Temme},
      year={2016},
      eprint={1412.2858},
      archivePrefix={arXiv},
      primaryClass={quant-ph},
      url={https://arxiv.org/abs/1412.2858}, 
}

@misc{alicki2008thermal,
      title={On thermal stability of topological qubit in Kitaev's 4D model}, 
      author={R. Alicki and M. Horodecki and P. Horodecki and R. Horodecki},
      year={2008},
      eprint={0811.0033},
      archivePrefix={arXiv},
      primaryClass={quant-ph},
      url={https://arxiv.org/abs/0811.0033}, 
}

@article{brown2016quantum,
   title={Quantum memories at finite temperature},
   volume={88},
   ISSN={1539-0756},
   url={http://dx.doi.org/10.1103/RevModPhys.88.045005},
   DOI={10.1103/revmodphys.88.045005},
   number={4},
   journal={Reviews of Modern Physics},
   publisher={American Physical Society (APS)},
   author={Brown, Benjamin J. and Loss, Daniel and Pachos, Jiannis K. and Self, Chris N. and Wootton, James R.},
   year={2016},
   month=Nov }

@article{baspin2022connectivity,
   title={Connectivity constrains quantum codes},
   volume={6},
   ISSN={2521-327X},
   url={http://dx.doi.org/10.22331/q-2022-05-13-711},
   DOI={10.22331/q-2022-05-13-711},
   journal={Quantum},
   publisher={Verein zur Forderung des Open Access Publizierens in den Quantenwissenschaften},
   author={Baspin, Nouédyn and Krishna, Anirudh},
   year={2022},
   month=May, pages={711} }

@article{flammia2017limits,
   title={Limits on the storage of quantum information in a volume of space},
   volume={1},
   ISSN={2521-327X},
   url={http://dx.doi.org/10.22331/q-2017-04-25-4},
   DOI={10.22331/q-2017-04-25-4},
   journal={Quantum},
   publisher={Verein zur Forderung des Open Access Publizierens in den Quantenwissenschaften},
   author={Flammia, Steven T. and Haah, Jeongwan and Kastoryano, Michael J. and Kim, Isaac H.},
   year={2017},
   month=Apr, pages={4} }

@article{brown2014entropic,
   title={Entropic Barriers for Two-Dimensional Quantum Memories},
   volume={112},
   ISSN={1079-7114},
   url={http://dx.doi.org/10.1103/PhysRevLett.112.120503},
   DOI={10.1103/physrevlett.112.120503},
   number={12},
   journal={Physical Review Letters},
   publisher={American Physical Society (APS)},
   author={Brown, Benjamin J. and Al-Shimary, Abbas and Pachos, Jiannis K.},
   year={2014},
   month=Mar }

@article{siva2017topological,
   title={Topological order and memory time in marginally-self-correcting quantum memory},
   volume={95},
   ISSN={2469-9934},
   url={http://dx.doi.org/10.1103/PhysRevA.95.032324},
   DOI={10.1103/physreva.95.032324},
   number={3},
   journal={Physical Review A},
   publisher={American Physical Society (APS)},
   author={Siva, Karthik and Yoshida, Beni},
   year={2017},
   month=Mar }

@article{lin2023geometrically,
  title={Geometrically local quantum and classical codes from subdivision},
  author={Lin, Ting-Chun and Wills, Adam and Hsieh, Min-Hsiu},
  journal={arXiv preprint arXiv:2309.16104},
  year={2023},
  url={https://arxiv.org/abs/2309.16104}
}

@article{freedman2021building,
  title = {Building manifolds from quantum codes},
  volume = {31},
  ISSN = {1420-8970},
  url = {http://dx.doi.org/10.1007/s00039-021-00567-3},
  DOI = {10.1007/s00039-021-00567-3},
  number = {4},
  journal = {Geometric and Functional Analysis},
  publisher = {Springer Science and Business Media LLC},
  author = {Freedman,  Michael and Hastings,  Matthew},
  year = {2021},
  month = jun,
  pages = {855–894}
}

@misc{gu2025layer,
      title={Layer codes as partially self-correcting quantum memories}, 
      author={Shouzhen Gu and Libor Caha and Shin Ho Choe and Zhiyang He and Aleksander Kubica and Eugene Tang},
      year={2025},
      eprint={2510.06659},
      archivePrefix={arXiv},
      primaryClass={quant-ph},
      url={https://arxiv.org/abs/2510.06659}, 
}

@misc{baspin2025free,
      title={The Free Energy Barrier: An Eyring-Polanyi bound for stabilizer Hamiltonians, with applications to quantum error correction}, 
      author={Nouédyn Baspin},
      year={2025},
      eprint={2509.17356},
      archivePrefix={arXiv},
      primaryClass={quant-ph},
      url={https://arxiv.org/abs/2509.17356}, 
}

@misc{williamson2025partial,
      title={Partial Self-Correction in Layer Codes}, 
      author={Dominic J. Williamson},
      year={2025},
      eprint={2510.09218},
      archivePrefix={arXiv},
      primaryClass={quant-ph},
      url={https://arxiv.org/abs/2510.09218}, 
}

@book{diestel2025graph,
  title = {Graph Theory},
  ISBN = {9783662701072},
  ISSN = {2197-5612},
  url = {http://dx.doi.org/10.1007/978-3-662-70107-2},
  DOI = {10.1007/978-3-662-70107-2},
  journal = {Graduate Texts in Mathematics},
  publisher = {Springer Berlin Heidelberg},
  author = {Diestel,  Reinhard},
  year = {2025}
}

@article{vizing1964estimate,
  title={On an estimate of the chromatic class of a p-graph},
  author={Vizing, Vadim G},
  journal={Diskret analiz},
  volume={3},
  pages={25--30},
  year={1964}
}

@misc{proofwiki:odd-vertices-edge-disjoint-trails,
  author       = {{ProofWiki}},
  title        = {Odd Vertices Determines Edge-Disjoint Trails},
  howpublished = {\url{https://proofwiki.org/wiki/Odd_Vertices_Determines_Edge-Disjoint_Trails}},
  note         = {ProofWiki, accessed 2026-03-26}
}

@article{brell2016proposal,
   title={A proposal for self-correcting stabilizer quantum memories in 3 dimensions (or slightly less)},
   volume={18},
   ISSN={1367-2630},
   url={http://dx.doi.org/10.1088/1367-2630/18/1/013050},
   DOI={10.1088/1367-2630/18/1/013050},
   number={1},
   journal={New Journal of Physics},
   publisher={IOP Publishing},
   author={Brell, Courtney G},
   year={2016},
   month=Jan, pages={013050} }

@article{michnicki2014memory,
  title={{3-d} topological quantum memory with a power-law energy barrier},
  author={Michnicki, Kamil P.},
  journal={Physical Review Letters},
  volume={113},
  pages={130501},
  year={2014},
  doi={10.1103/PhysRevLett.113.130501},
  url={https://arxiv.org/abs/1406.4227}
}

@misc{portnoy2024sparse,
  title={Sparse maps},
  author={Portnoy, Elia},
  year={2024},
  eprint={2407.05137},
  archivePrefix={arXiv},
  primaryClass={math.MG},
  url={https://arxiv.org/abs/2407.05137}
}

\end{document}